%% file: arxiv.tex
\documentclass[11pt,a4paper]{article}

\usepackage{style}

\usepackage[margin=1in]{geometry}
\usepackage{amsmath}
\usepackage{amsthm}
\usepackage{xspace}
\usepackage{amssymb}

\newtheorem{theorem}{Theorem}[section]
\newtheorem{lemma}[theorem]{Lemma}

\newtheorem{definition}[theorem]{Definition}
\newtheorem{corollary}[theorem]{Corollary}

\title{Parallel Minimum Cost Flow in Near-Linear Work and Square Root Depth for Dense Instances}
\author{
Jan van den Brand\thanks{Georgia Institute of Technology, Atlanta, Georgia, USA}
\and 
Hossein Gholizadeh\thanks{Karlsruhe Institute of Technology, Karlsruhe, Germany} \and 
Yonggang Jiang\thanks{Max Planck Institute for Informatics and Saarland University, Saarbr{\"u}cken, Germany} \and 
Tijn de Vos\thanks{TU Graz, Graz, Austria}
}

\date{}

\begin{document}

%\begin{titlepage}

\maketitle
\begin{abstract}
    For $n$-vertex $m$-edge graphs with integer polynomially-bounded costs and capacities, we provide a randomized parallel algorithm for the minimum cost flow problem with $\tilde O(m+n^ {1.5})$ work and $\tilde O(\sqrt{n})$ depth\footnote{We use $\tO{\cdot}$ to hide polylogarithmic factors.}. On moderately dense graphs ($m>n^{1.5}$), our algorithm is the first one to achieve both near-linear work and sub-linear depth. 
    Previous algorithms are either achieving almost optimal work but are highly sequential \cite{ChenKLPGS22}, or achieving sub-linear depth but use super-linear work \cite{LS14,OrlinS93}.
    Our result also leads to improvements for the special cases of max flow, bipartite maximum matching, shortest paths, and reachability. Notably, the previous algorithms achieving near-linear work for shortest paths and reachability all have depth $n^{o(1)}\cdot \sqrt{n}$ \cite{LiuJS19,FHLRS25}.

    Our algorithm consists of a parallel implementation of \cite{BrandLLSS0W21}. One important building block is a parallel \emph{batch-dynamic}  expander decomposition, which we show how to obtain from the recent parallel expander decomposition of \cite{ChenMGS25}.     
\end{abstract}

\thispagestyle{empty}
\newpage
\thispagestyle{empty}
\tableofcontents

\newpage
\input{introduction}

\newpage
\input{preliminaries}

\newpage
\input{overview}

\newpage
\input{expander}

%\clearpage

\newpage
\input{algorithm}

\newpage
\section*{Acknowledgments}
\addcontentsline{toc}{section}{Acknowledgments}

We thank Thatchaphol Saranurak for pointing out the paper \cite{ChenMGS25} at the early stage of this project.

Jan van den Brand was supported by NSF Award CCF-2338816.

Yonggang Jiang is supported by Google PhD Fellowship.

This research was funded in whole or in part by the Austrian Science Fund (FWF) \url{https://doi.org/10.55776/P36280} and \url{https://doi.org/10.55776/I6915}. For open access purposes, the author has applied a CC BY public copyright license to any author-accepted manuscript version arising from this submission.

\newpage
\addcontentsline{toc}{section}{References}
\bibliography{references}{}
\bibliographystyle{alpha}

\appendix

\newpage
\input{parallelization}

\newpage
\input{heavyhitters}

\newpage
\input{LeverageScore}

\newpage
\input{primalGradientMaintenance}

\newpage
\input{graphDataStructures}

\end{document}

%% file: introduction.tex
\section{Introduction}
Minimum cost flow is one of the most fundamental questions in algorithm design. It has been widely studied in the sequential setting, but relatively few results are known in the parallel setting.
All previous results for min-cost flow, including its important special cases max flow and bipartite maximum matching, either achieve optimal work but are highly sequential~\cite{ChenKLPGS22,BrandLLSS0W21}, or achieve sublinear depth of non-optimal work: Lee and Sidford~\cite{LS14} have $\tO{m\sqrt n}$ work and $\tilde O(\sqrt n)$ depth, and the matrix multiplication technique which gives $\tilde O(m^{\omega+4})$ work\footnote{Here, $\omega\approx 2.37$~\cite{DuanWZ23,WilliamsXXZ24} denotes the exponent of current
matrix multiplication.} and $\tO{1}$ depth~\cite{OrlinS93}. This brings us to the natural question. 

\begin{mdframed}
    \centering
    \textbf{Question:} 
    \textit{Can we achieve both linear work and sublinear depth for min-cost flow?}
\end{mdframed}

Although parallel flow has been studied for decades (see, e.g. \cite{ShiloachV82a,KarpUW86,ramachandran1990parallel}), only very recently
breakthrough results answer this question positively for the special cases of reachability \cite{Fineman20,LiuJS19}, shortest path \cite{RozhonGHZL22,CaoF23} and negative weight SSSP \cite{AshvinkumarBCGH24}, where the current record is near-linear work and $n^{1/2+o(1)}$-depth. 
However, the question of getting linear work and sublinear depth still remains widely open for bipartite maximum matching, max flow, and the most general min-cost flow.

In this paper, we resolve the most general problem for moderately dense graphs: for $m\ge n^{1.5}$, we give a near-linear work and $\tilde O(\sqrt n)$-depth algorithm for min-cost flow, which matches the current progress on shortest path related problems, where it even improves the small $n^{o(1)}$ factor.

\begin{theorem}[Informal]\label{thm:mainInformal}
     There exists a randomized algorithm that computes exact min-cost flow in $\tilde O(m+n^{1.5})$ work and $\tilde O(\sqrt n)$ depth.
\end{theorem}

\paragraph{Techniques.}
To obtain this result, we use an interior point method based on the sequential algorithm by v.d.Brand et al.~\cite{BrandLLSS0W21}. Their interior point method contains $\tO{\sqrt{n}}$ iterations, and they show how to perform each iteration in amortized $\tO{m/\sqrt{n}+n}$ time by using sequential data structures maintaining the information needed for each iteration. For an overview, see \Cref{sec:overview_IPM}.

Our contribution is to parallelize the data structures so that they cost $\tO{m/\sqrt{n}+n}$ work and $\tO{1}$ depth per iteration, which then implies a $\tO{m+n^{1.5}}$ work and $\tO{\sqrt{n}}$ depth flow algorithm. 
In particular, most of the operations used by the data structures involve multiplying matrices and are thus easily parallelizable. The bottleneck is the following data structure involving parallel batch-dynamic expander decomposition: given edge deletions (or insertions) in batches, maintaining an expander decomposition of the graph in work proportional to the batch size and depth $\tO{1}$.

Previous work shows dynamic expander decomposition can be done in sequential time proportional to the batch size \cite{SaranurakW19}. However, their algorithm involving push-relabel is highly sequential. A more recent work \cite{ChenMGS25} shows how to use parallel push-relabel to get a parallel static expander decomposition. Our main contribution is adapting their algorithms to get a parallel batch-dynamic expander decomposition that supports low-depth updates. We provide a detailed overview in \Cref{sec:overview_expander}.

In the next section, we present our main result and its corollaries more formally, also comparing it in detail to the state of the art.

\subsection{Our Results}
\paragraph{Min-Cost Flow and Max Flow.}
In the \emph{minimum cost flow problem} (min-cost flow), we are given a connected directed graph $G = (V, E, u, c)$ with edges capacities $u \in \mathbb{R}_{\ge 0}^E$ and costs $c \in \mathbb{R}^E$. Given $s,t\in V$, we call $x \in \mathbb{R}^E$ an \emph{$s$-$t$ flow}  if $x_e \in [0, u_e]$ for all $e \in E$ and for each vertex $v \in V\setminus \{s, t\}$ the amount of
flow entering $v$ equals the amount of flow leaving $v$, i.e., $\sum_{(u,v)\in E}e_{(u,v)}=\sum_{(v,u)\in E}e_{(v,u)}$. The
\emph{value} of the flow is the amount of flow leaving $s$ (or equivalently, entering $t$): $\sum_{(s,u)\in E}e_{(s,u)}=\sum_{(u,t)\in E}e_{(u,t)}$. The \emph{maximum flow problem} (max flow) is to compute an $s$-$t$ flow of maximum value. In the \emph{minimum cost maximum flow problem} (min-cost flow), the goal is to compute a maximum $s$-$t$ flow $x$ of minimum cost, $x^\top c$. 

The min-cost flow problem generalizes many important graph problems, among which the max flow problem, maximum bipartite matching, the negative-weight shortest path problem, and reachability.

\begin{restatable}{theorem}{MainTheorem}\label{thm:main}
    There exists an algorithm that, given a directed graph $G=(V,E,u,c)$ with capacities $u \in \mathbb{Z}_{\ge 0}^E$ and costs $c \in \mathbb{Z}^E$, and $s,t\in V$, computes with high probability the exact minimum cost maximum $s$-$t$ flow in $\tilde O((m+n^{1.5}\log(CW))\log(CW))$ work and $\tilde O(\sqrt n\log(CW))$ depth, where $W := ||u||_\infty$ and $C:= ||c||_\infty$. 
    %\jan{this should be $m\log(CW)+n^{1.5}\log(CW)^2$}
\end{restatable}

This improves upon the work by Lee and Sidford~\cite{LS14}, which runs in $\tilde O(\sqrt n)$ depth and $\tO{m\sqrt n}$ work.

\paragraph{Corollaries.}
We note that this also gives new results for some of the problems that reduce to min-cost flow. Firstly, we obtain an algorithm for max flow. Here, we again improve upon the result by Lee and Sidford~\cite{LS14}. We provide further related work in \Cref{sec:related_work}. For other special cases, improvements on \cite{LS14} are already known or additional work-depth trade-offs are possible. We compare our results for the cases of bipartite maximum matching, negative-weight shortest paths and reachability. 

\paragraph{Bipartite Maximum Matching.}
For the case of bipartite maximum matching, we obtain the following result. 
\begin{corollary}
    There exists an algorithm that, given a bipartite graph $G=(V,E)$, computes with high probability a maximum matching in $\tilde O(m+n^{1.5})$ work and $\tilde O(\sqrt{n})$ depth.  
\end{corollary}

Again, this improves upon the $\tilde O(m\sqrt{n})$-work and $\tilde O(\sqrt{n})$-depth algorithm of Lee and Sidford~\cite{LS14}. An alternative remains the algorithm by Mulmuley, Vazirani, and Vazirani \cite{MulmuleyVV87}, which has more work, $\tilde O(n^\omega)$, but lower depth, $\tilde O(1)$.

\paragraph{Negative-Weight Shortest Paths.}
We also obtain the following corollary for the negative-weight shortest path problem. 
\begin{corollary}
    There exists an algorithm that, given a directed graph $G=(V,E,w)$ with negative edge weights and a source $s\in V$, computes SSSP from $s$ with $\tilde O(m+n^{1.5})$ work and $\tilde O(\sqrt n)$ depth. 
\end{corollary}
For dense graphs ($m\ge n^{1.5}$), this improves with a $n^{o(1)}$ factor over the state of the art: Fisher et al.~\cite{FHLRS25} provide an algorithm with $\tilde O(m)$ work and $n^{0.5+o(1)}$ depth.

\paragraph{Reachability.}
Even for the special case of reachability, we obtain a similar improvement. 

\begin{corollary}
    There exists an algorithm that, given a directed graph $G=(V,E)$ and a source $s\in V$, computes reachability from $s$ with $\tilde O(m+n^{1.5})$ work and $\tilde O(\sqrt n)$ depth. 
\end{corollary}
Again, for dense graphs, we improve upon the state of the art by a $n^{o(1)}$-factor: Liu, Jambulapati, and Sidford~\cite{LiuJS19} provide an algorithm with $\tilde O(m)$ work and $n^{0.5+o(1)}$ depth. 
For other trade-offs between work and depth, see \Cref{sec:related_work} and \Cref{tab:reachability}.

Our algorithm is based on an interior point method for linear programs. Previous work on near-linear work reachability is limited. There is a folklore algorithm consisting of a parallel BFS, giving $\tilde O(n)$ depth. The more recent algorithms~\cite{Fineman20,LiuJS19} are based on \emph{shortcuts}. With our paper, we show a different approach for the case of dense graphs. 
% It is an interesting open question whether an approach based on an interior point method can lead to a $\tilde O(m)$ work and $\tilde O(\sqrt n)$ depth algorithm. 
% the interior point method presented in this paper could be simplified for the special case of reachability to obtain such a result. 

\paragraph{Expander Decompositions.}
Expander decompositions have become an essential tool in algorithm design. They have been used in the first almost-linear time algorithms for many fundamental questions, including but not limited to, max flow \cite{KelnerLOS14}, min-cost flow \cite{ChenKLPGS22}, electrical flow \cite{SpielmanT04}, Gomory-Hu trees \cite{Abboud0PS23}, and vertex/edge connectivity \cite{KawarabayashiT19,Li21,LiNPSY21}. In many such applications, a \emph{dynamic} expander decomposition is used as a subroutine.   

In the parallel setting, the first expander decompositions are given by Chang and Saranurak~\cite{ChangS19,ChangS20}, with the state of the art given by Chen et al.~\cite{ChenMGS25} (see \Cref{sec:overview_expander} for more details). We show how to use the latter to obtain a parallel \emph{batch-dynamic} expander decomposition. 

\begin{lemma}[Informal version of \Cref{lem:dynamicExpanderDecomposition}]\label{lem:dynamicExpanderDecompositionInformal}
     There exists a randomized data structure that maintains a parallel $\phi$-expander decomposition\footnote{Here we decompose the \emph{edge set} of $G$ into $\phi$-expanders, such that each vertex is in at most $\tilde O(1)$ many expanders.}, where a batch of $E'$ updates uses $\tO{1/\phi^4}$ depth and amortized $\tO{|E'|/\phi^5}$ work.
\end{lemma}

\subsection{Related Work}\label{sec:related_work}
\paragraph{The Laplacian Paradigm.}
Many of the results on flow in the last two decades make use of techniques from the \emph{Laplacian paradigm}. This line of research was initiated by the seminal work of Spielman and Teng~\cite{SpielmanT04}, who showed that linear equations in the Laplacian matrix of a graph can be solved in near-linear time. The Laplacian matrix of a weighted graph $G$ is defined as $L(G):={\rm{Deg}}(G)-A(G)$, where ${\rm{Deg}}(G)$ is the diagonal weighted degree matrix: ${\rm{Deg}}(G)_{uu}:= \sum_{(u,v)\in E} w(u,v)$ and ${\rm{Deg}}(G)_{uv}:=0$ for $u\neq v$, and $A(G)$ is the adjacency matrix: $A(G)_{uv}:=w(u,v)$.

More efficient linear system solvers have been presented since~\cite{SpielmanT04}, see, e.g., \cite{KOSZ13,KoutisMP14,KoutisMP11,CohenKMPPRX14,KyngS16,KyngLPSS16}. The Laplacian paradigm has achieved many successes, including but not limited to flow problems~\cite{Madry13,Sherman13,KelnerLOS14,Madry16,Peng16,CohenMSV17,LiuS2020FasterDivergence,LiuS20,AxiotisMV20}, and bipartite matching~\cite{BrandLN+20}

Later, it was also shown that linear systems can be solved efficiently in the parallel setting \cite{KoutisM07,BlellochGKMPT14,peng2013efficientparallelsolversdd,KoutisX16,LeePS15,SachdevaZ23}.
Again, this has had many applications, including but not limited to  flow~\cite{LS14,AgarwalKLPWWZ24} and shortest paths~\cite{Li20,AndoniSZ20}.

\paragraph{Parallel Min-Cost Flow and Max Flow.}
The first algorithms for parallel max flow had $\tilde O(mn)$ work and $\tilde O(n^2)$ depth \cite{ShiloachV82a,ramachandran1990parallel}, or unspecified polynomial work and $\tilde O(1)$ depth~\cite{KarpUW86}. The latter comes from a reduction to maximum matching and only holds for uncapacitated graphs. 
Recently, a (relatively) simple combinatorial algorithm was given by Peretz and Fischler~\cite{PeretzF22}, with $\tilde O(mn)$ work and $\tilde O(n)$ depth. 
A significantly faster algorithm was already provided by Lee and Sidford~\cite{LS14}, who provided an algorithm with $\tO{m\sqrt n}$ work and $\tilde O(\sqrt n)$ depth by using an interior point method. This algorithm also holds for the more general problem of min-cost flow. 

% When the flow value is integral small, say at most $k$, then Lingas and Persson~\cite{LingasP15} provide an efficient algorithm with $\tilde O(2^k \poly(kn))$ work and $\tilde O(k)$ depth. 

For \emph{approximate} flow, faster algorithms are known.  
An early result by Serna and Spirakis~\cite{SernaS91} showed how to obtain a $(1+\eps)$-approximation of max flow in $\tilde O(\log(1/\eps))$ depth, via a reduction to maximum matching. This algorithm comes with unspecified polynomial work.   
Within the Laplacian paradigm, the interior point method of Madry~\cite{Madry13} combined with the SDD solver of Peng and Spielman~\cite{peng2013efficientparallelsolversdd} gives an algorithm for max flow with $\tilde O(m^ {10/7})$ work and $\tilde O(m^{3/7})$ depth at the cost of a $1/\poly(n)$ additive error. 
More recently, Agarwal et al.~\cite{AgarwalKLPWWZ24} gave a $(1+\eps)$-approximation of max flow for undirected graphs with $\tilde O(m\eps^{-3})$ work and $\tilde O(\eps^{-3})$ depth. 

The algorithm of Madry~\cite{Madry13} is based on an interior point method with $\tilde O(m^{3/7})$ iterations. 
We note that other efficient sequential flow solvers based on interior point methods, e.g., ~\cite{Madry16,CohenMSV17,KathuriaLS24}, also have the potential to be implemented using a parallel SDD solver. However, these algorithms provide \emph{approximate} solutions. Depending on which version of the min-cost or max flow LP they solve, this does not give an \emph{exact} solution. To be precise, \cite{LS14, BrandLLSS0W21} can round their approximate solution directly to obtain an exact solution. On the other hand, \cite{Madry16,CohenMSV17,KathuriaLS24} need a polynomial number of augmenting paths to derive an exact solution. In the parallel setting, computing an augmenting path currently uses $\tilde O(\sqrt n)$ depth, making it an expensive subroutine. Although we believe such results provide sublinear depth, we are not aware of any works specifying the exact trade-off. We note however, that our solver \Cref{thm:main} is as fast as solving the augmenting path subroutine. Since this is called polynomially many times, our solver will have a polynomial advantage.

\paragraph{Iteration Count and Depth.}
The more modern, fast flow algorithms are all based on interior point methods. In particular, Chen et al.~\cite{ChenKLPGS22} gave an almost-linear time algorithm for min-cost flow. This algorithm (and also its deterministic version~ \cite{Brand0PKLGSS23}) uses $\Omega(m)$ iterations, which seems to render it hard to implement it efficiently in a parallel setting, as any intuitive implementation uses at least one round per iteration. The algorithms with lowest iteration counts have either $\tilde\Theta(\sqrt{n})$ iterations~\cite{LS14, BrandLLSS0W21}, which is the lowest in terms of $n$, or $\tilde \Theta(m^{1/3})$ iterations~\cite{KathuriaLS24,AxiotisMV20}, which is lowest in terms of $m$. 
%$\Theta(m^{3/7})$ iterations~\cite{Madry16,CohenMSV17}, or .
This means that without significant improvements in the iteration count of the max flow interior point methods, a depth of $\tilde \Theta(\sqrt{n})$ is currently the best we can hope for in the parallel setting using interior point methods.

\begin{table}[ht]
    \centering
\hspace{-10pt}
\begin{tabular}{|c|c|c|}
     \hline
     \textbf{Work} & \textbf{Depth} & \\ \hline
     $m^{1+o(1)}$ & $m^{1+o(1)}$ & \cite{ChenKLPGS22} \\ \hline
     $\tilde O(m+n^{1.5})$ & $\tilde O(m+n^{1.5})$ & \cite{BrandLLSS0W21} \\ \hline
     $\tilde O(m^{\omega+4})$ & $\tilde O(1)$ & \cite{OrlinS93} \\ \hline
     $\tilde O(m\sqrt{n})$ & $\tilde O(\sqrt{n})$ & \cite{LS14} \\ \hline
     $\tilde O(m+n^{1.5})$ & $\tilde O(\sqrt{n})$ & \textbf{This paper}\\ \hline  
\end{tabular}
\begin{tabular}{|c|c|c|}
     \hline
     \textbf{Work} & \textbf{Depth} & \\ \hline
     $O(m)$ & $\tilde O(n)$ & Parallel BFS \\ \hline     
     $\tilde O(n^\omega)$ & $\tilde O(1)$ & \begin{tabular}{c}Parallel\\ Trans.\ Closure\end{tabular} \\ \hline
     $\tilde O(m+n\rho^2)$ & $\tilde O(n/\rho)$ & \cite{Spencer97} \\ \hline
     $\tilde O(m\rho+\rho^4/n)$ & $\tilde O(n/\rho)$ & \cite{UllmanY91} \\ \hline
   %  $\tilde O(m)$ & $\tilde O(n^ {2/3})$ & \cite{Fineman20} \\ \hline
     $\tilde O(m)$ & $n^{0.5+o(1)}$ & \cite{LiuJS19} \\ \hline
     $\tilde O(m+n^{1.5})$ & $\tilde O(\sqrt{n})$ & \textbf{This paper}\\ \hline  
\end{tabular}

\caption{On the left, an overview of parallel min-cost flow. On the right, an overview of the special case of parallel reachability. The latter is adapted from \cite{LiuJS19}. }
    \label{tab:reachability}
\end{table}

\paragraph{Bipartite Maximum Matching.}
Early on, it was shown that bipartite maximum matching can be computed with polylogarithmic depth~\cite{Lovasz79,KarpUW86}. Later, it was shown by Mulmuley, Vazirani, and Vazirani~\cite{MulmuleyVV87} how to reduce the problem to matrix inversion and hence to matrix multiplication. This rendered an algorithm with $\tilde O(n^\omega)$ work and $\tilde O(1)$ depth. For sparser graphs, different work-depth trade-offs were given by the reduction to max flow or min-cost flow. In particular, the algorithm of Lee and Sidford~\cite{LS14} gives $\tO{m\sqrt n}$ work and $\tilde O(\sqrt n)$ depth.

\paragraph{Negative-Weight Shortest Paths.}
Even in the sequential setting, the negative shortest path problem admitted no algorithms with less than $\tO{m\sqrt n}$ time until recently. Up until then, the scaling framework from Goldberg~\cite{Gabow85,GabowT89,Goldberg95} provided the state of the art. In 2022, Bernstein, Nanongkai, and Wulff-Nilsen provided an algorithm with $\tilde O(m)$ time~\cite{BernsteinNW25}. 

In the parallel setting, a parallel version of Goldberg's algorithm was given by Cao, Fineman, and Russell~\cite{CaoFR22} with $\tO{m\sqrt n}$ work and $n^{0.5+o(1)}$ depth. 
Later, a parallel version of~\cite{BernsteinNW25} was given by Ashvinkumar et al.~\cite{AshvinkumarBCGH24} with $m^{1+o(1)}$ work and $n^ {0.5+o(1)}$ depth. 
This was further improved by Fisher et al.~\cite{FHLRS25}, who provide an algorithm with $\tilde O(m)$ work and $n^{0.5+o(1)}$ depth.

\paragraph{Reachability.}
Two folklore algorithms for reachability are running a parallel breadth first search (BFS) or computing the transitive closure. The former has $O(m)$ work and $\tilde O(n)$ depth, and the latter has $\tilde O(n^\omega)$ work and $\tilde O(1)$ depth. 
Spencer~\cite{Spencer97} and Ullman and Yannakakis~\cite{UllmanY91} provide algorithms with parameterized work-depth trade-offs. However, for near-linear work, they do not improve on the $\tilde O(n)$ depth of parallel BFS. 
More recently, the breakthrough paper of Fineman~\cite{Fineman20} provides $\tilde O(m)$ work and $\tilde O(n^{2/3})$ depth. This paper uses a tool called \emph{shortcuts}. These are edges added to the graph that decrease the diameter of the graph, without impacting reachability. 
Building on these concepts, Liu, Jambulapati, and Sidford~\cite{LiuJS19} improved this to $\tilde O(m)$ work and $n^{0.5+o(1)}$ depth. 
For an overview of these results, also see \Cref{tab:reachability}.

\paragraph{Distributed Min-Cost Flow.}
Exact min-cost flow has also been studied in the related CONGEST model~\cite{ForsterGLPSY21,Vos23} and the (broadcast) congested clique~\cite{ForsterV22,ForsterV23}. The CONGEST model and (broadcast) congested clique are message passing models, where a network of processors is modeled as a graph. Problems are solved in synchronous rounds, in which the processors can perform (unlimited) local computation and exchange bounded-size messages with their neighbors. The complexity of a problem is measured by the number of rounds it requires. The difference between the CONGEST model and the congested clique lies in which communication network is used. \cite{Vos23} achieves $n^{1.5+o(1)}$ rounds for the min-cost flow problem in the CONGEST model, and \cite{ForsterV22} achieves $\tilde O(\sqrt n)$ rounds in the broadcast congested clique.  

Since in both these models each node has infinite computing capacities, the total work is not taken into account and hence the results there are orthogonal to the contributions of this paper. 

\subsection{Open Questions}
Our min-cost flow algorithm, \Cref{thm:main}, consists of an interior point method with $\tO{\sqrt n} $ iterations. Since each iteration takes $\tO{m/\sqrt n +n}$ work, this leads to $\tO{m+n\sqrt n}$ work in total. Obtaining an algorithm with a lower iteration count is an obvious open question, but also hard -- even in the sequential setting. An interesting open question is whether we can reduce the per-iteration cost to $\tO{m/\sqrt n}$. Currently, our algorithm involves updating an $n$-dimensional vector in each iteration. This would have to be circumvented to obtain an algorithm with lower work. 

Our batch-dynamic parallel $\phi$-expander decomposition, \Cref{lem:dynamicExpanderDecompositionInformal}, takes $\tO{1/\phi^4}$ depth and $\tO{|E'|/\phi^5}$ work. This depth matches the static state of the art~\cite{ChenMGS25}, so we expect any improvement for the dynamic depth to also require improving the static depth. The work for our dynamic algorithm is a factor $1/\phi^3$ worse than its static counterpart (\Cref{lem:staticofdyanmicexpander}). It is an interesting open question if this overhead can be reduced.

\subsection*{Outline}
We give our preliminaries in \Cref{sec:preliminaries}, and then, in \Cref{sec:overview}, we give a high-level overview of our result. Among other things, this section describes how expander decompositions, \Cref{lem:dynamicExpanderDecompositionInformal}, are used for our main result, \Cref{thm:main}.  
Then, in \Cref{sec:expander} we prove \Cref{lem:dynamicExpanderDecompositionInformal}. And finally, in \Cref{sec:alg}, we give our main algorithm for \Cref{thm:main}. The subroutines of this algorithm are deferred to the appendix. In \Cref{overview:organization} we give an overview of these sections, after having defined the subroutines.

% \paragraph{Open questions.}
% This leads to three obvious open questions:
% \begin{itemize}
%     \item Can we get $\Ot{(\sqrt n)}$ depth and $\Ot{(m)}$ work?
%     \item Can we get any $o(\sqrt n)$ depth algorithm?
%     \item Can we get an efficient parallel combinatorial algorithm?\\
%     Combinatorial max flow in sequential: \cite{BernsteinBST24}
%     \item Reachability via IPM
% \end{itemize}

%% file: preliminaries.tex
\section{Preliminaries}\label{sec:preliminaries}

%\jan{2nd preliminaries in appendix? move a lot of this stuff there since it's not needed in first 10pages}
%\jan{not sure if 3p is really a prelim. maybe be split into a "parallelization basics" section...}
%\jan{some of it can also go to overview directly, eg what is a linear program}

\paragraph{Parallel Computing.}
The parallel model we consider is the widely used \emph{\pram{} (Parallel RAM) model}. In this model, we consider multiple synchronous processors with shared memory. The access to this memory can differ: exclusive-read-exclusive-write, concurrent-read-exclusive-write, and concurrent-read-concurrent-write. However, up to polylogarithmic factors, these have shown to be equivalent. We measure the efficiency of an algorithm by its \emph{work}, the total number of operation performed, and its \emph{depth}, the longest sequence of dependent operations. 

\paragraph{Notation.}
We let $[n] \defeq \{1,\dots,n\}$ denote the set of the first $n$ natural numbers. 
We use the $\tO{\cdot}$ notation to hide $\poly(\log \epsilon^{-1},\log n)$ factors, where $n$ denotes the number of vertices.
We write \emph{with high probability} or \emph{w.h.p.} to describe a probability of $1 - n^c$ for some constant $c > 0$.
Further, we write $\mathbf{1}_{\text{condition}}$ for the indicator variable, 
which is $1$ if the condition is true and $0$ otherwise.

\subsection{Linear Algebra}
\paragraph{Norms}
We write $\| \cdot \|_p$ for the $\ell_p$-norm, i.e., for any vector $v \in \R^d$, 
%\begin{align*}
    $\|v\|_p := \left(\sum_i |v_i|^p\right)^{1/p}$, $\|v\|_\infty = \max_i | v_i |$,    
%\end{align*}
and $\| v \|_0$ denotes the number of non-zero entries of $v$.
Further, for a vector $\tau \in \R^d$, we define 
$\| v \|_\tau := \left(\sum_i \tau_i v_i^2\right)^{1/2}$ and the mixed norm 
    $\|v\|_{\tau + \infty} := \| v \|_\infty + C \log(4m/n) \|v\|_\tau$    
for a large constant $C$. 
With regards to this norm, for given vectors $h, \tau \in \R^n$, we define
\begin{align*}
    h^{\flat(\tau)} := \argmax_{\|x\|_{\tau + \infty} \le 1} \langle h, x \rangle.    
\end{align*}

\paragraph{Matrix and Vector Operations.}
Given vectors $u,v \in \R^d$ for some $d$,
we perform arithmetic operations $\cdot,+,-,/,\sqrt{\cdot}$ element-wise.
For example, $(u\cdot v)_i = u_i\cdot v_i$ or $(\sqrt{v})_i = \sqrt{v_i}$.
For the inner product, we write $\langle u, v \rangle$ and $u^\top v$ instead.
For a vector $v \in \R^d$ and a scalar $\alpha \in \R$, we let $(\alpha v)_i = \alpha v_i$ and $(v + \alpha)_i = v_i + \alpha$. Additionally, given a vector $v \in \R^d$, we define $\mV \in \R^{d \times d}$ as the diagonal matrix whose diagonal entries are the elements of $v$, i.e., $\mV_{i,i} = v_i$ for all $i \in [d]$. 
For a function $\phi:\R^m \to \R^m$, the diagonal matrix $\mathbf \Phi(x)$ is defined analogously.
% Diagonal matrices $\mPhi ''(x)$ and $\mTau(x)$ are defined from $\phi''(x)$ and $\tau(x)$ analogously.

For a positive definite matrix $\mM \in \R^{n \times n}$, we denote the weighted Euclidean $\mM$-norm of a vector $x$ as $\norm{x}_\mM = \sqrt{x^\top \mM x}$. Furthermore, for symmetric matrices $\mA, \mB \in \R^{n \times n}$, we use $\preceq$ to denote the Loewner ordering, i.e., $\mB \preceq \mA$ if and only if $\norm{x}_{\mA - \mB} \geq 0$ for all $x \in \R^n$.

% For $\epsilon>0$, we write $\mA \approx_\epsilon \mB$ to denote matrix $\mA$ being a $\exp(\pm\epsilon)$-spectral approximation of $\mB$. Similarly, we extend this notation for vectors, letting $u \approx_\epsilon v$ if and only if $\exp(-\epsilon) v_i$ $\leq u_i \leq \exp(\epsilon) v_i$ for all $i$. Observe that $\exp(\pm\epsilon)$ is close to $(1\pm\epsilon)$ for small $\epsilon >0$.

In this context, we write $\mA \approx_\varepsilon \mB$ if and only if $\exp(-\varepsilon) \mB \preceq \mA \preceq \exp(\varepsilon) \mB$. Observe that $\exp(\pm\epsilon)$ is close to $(1\pm\epsilon)$ for small $\epsilon >0$. Similarly, we extend this notation for vectors, letting $u \approx_\varepsilon v$ if and only if $\exp(-\varepsilon) v \leq u \leq \exp(\varepsilon) v$ entrywise. This implies that $u \approx_\varepsilon v \approx_\delta w$ yields $u \approx_{\varepsilon + \delta} w$, and $u \approx_\varepsilon v$ implies $u^\alpha \approx_{\varepsilon \cdot |\alpha|} v^\alpha$ for any $\alpha \in \R$.

For any matrix $\mA \in \R^{m \times n}$, we denote $a_i \in \R^n$ as the $i^\mathrm{th}$ row of $\mA$, represented as a column vector. Further, let $\nnz(\mA)$ denote the number of non-zero entries in $\mA$, and similarly, $\nnz(a_i)$ the number of non-zero entries in $a_i$.

% As our approach is mainly based on the interior point method (IPM) of \cite{BrandLLSS0W21} and \cite{BrandLN+20}, we also use similar notation to theirs.

\paragraph{Leverage Scores and Lewis-Weights.}
For a full-rank matrix $\mA \in \R^{m \times n}$, let $\sigma(\mA) \in \R^m$ denote its \emph{leverage scores}, defined as $\sigma(\mA)_i \defeq a_i^\top (\mA^\top \mA)^{-1} a_i$ for each $i \in [m]$.  

For $p \in (0, \infty)$ and a full-rank matrix $\mA \in \R^{m \times n}$, the $\ell_p$-Lewis weights are defined as the solution $w \in \R^m_{>0}$ to the equation $w = \sigma(\mW^{\frac{1}{2} - \frac{1}{p}} \mA)$, where $\mW = \mdiag(w)$. As we follow the approach of \cite{BrandLLSS0W21}, similar to them, we use regularized Lewis weights, i.e., we let $p = 1 - \frac{1}{4 \log (4m / n)}$ and consider the solution $w \in \R^m_{>0}$ to
\begin{align*}
    \sigma(\mW^{\frac{1}{2} - \frac{1}{p}} \mA) + \frac{m}{n} \mathbf{1}.
\end{align*}

\paragraph{Graph Matrices.}
For a directed graph $G = (V, E)$, we define the edge-vertex incidence matrix $\mA \in \{-1,0,1\}^{E \times V}$ via $\mA_{e,u} = -1$ and $\mA_{e,v} = 1$ for any edge $e=(u,v) \in E$. We typically refer to the number of edges by $m$ and the number of vertices by $n$, so that the incidence matrix is an $m \times n$ matrix, which allows us to use the notation $\mA_{i,j}$ for $i \in [m]$ and $j \in [n]$, assuming an ordering of edges and vertices.

%\cite{BrandLLSS0W21}'s IPM requires an initial feasible point to reach an optimal solution. As it is typically hard to find this initial point, the usual approach is to modify the input graph to a new graph by adding a bi-directional star rooted at a new vertex $z$, so that it has a trivial initial feasible point. Formally, given $G = (V,E)$, define $\tG = (V \cup \{z\}, E \cup \tE)$, where $\tE = \{(v,z), (z,v) \mid v \in V\}$. By setting the costs of the newly added edges to large enough values and choosing the capacity of these edges in a sophisticated way, it is easy to see that an optimal solution to the modified graph $\tG$ gives us an optimal solution of $G$. \hossein{enough explanation?}

%For technical reasons, \cite{BrandLLSS0W21}'s IPM requires the input (incidence) matrix to have full rank. This can be achieved by removing one arbitrary column of $\tmA$, the incidence matrix of $\tG$. For convenience, we remove the column associated with the vertex $z$. Nevertheless, it is easy to show that the feasibility condition, i.e., $\tmA \tf = b$ for a feasible flow $\tf$, still holds even after removing the column (see Fact 7.3, \cite{BrandLLSS0W21}). 

%Since $\tG$ has only $2n$ more edges than $G$, here, we generally assume that the input matrix $\mA \in \R^{m \times n}$ is a full-rank matrix, which has been obtained by removing one column from the incidence matrix of a graph with $m$ edges and $n+1$ vertices.

For technical reasons, \cite{BrandLLSS0W21}'s IPM requires the input (incidence) matrix to have full rank. This can be achieved by removing one column of the incidence matrix $\mA$, without hurting the feasibility condition (see Fact 7.3, \cite{BrandLLSS0W21}). Thus, in the following, we generally assume the matrix $\mA$ has full rank and is obtained by removing one column of the incidence matrix of a graph $G$.

\subsection{Graphs}
\paragraph{Expander Graphs.} For an \emph{undirected graph} $G=(V,E)$, we use $E(A,B)$ for two disjoint vertex sets $A,B$ to denote the edge set between $A$ and $B$. We use $\deg_G(A)$ to denote $\sum_{v\in A}\deg_G(v)$ where $\deg_G(v)$ is the degree of $v$ in $G$. We say $G$ is a \emph{$\phi$-expander} if for every $S\subseteq V$ with $S\not=\emptyset,S\not=V$, we have
\[\frac{|E(S,V\backslash S)|}{\min(\deg_G(S),\deg_G(V\backslash S))}\ge\phi\]

In this overview, to simplify notation, when we say a graph is an \emph{expander}, we mean a $\phi$-expander for some $\phi=\frac{1}{\polylog(n)}$.

\paragraph{(Edge-partitioned) expander decomposition.} For an undirected graph $G=(V,E)$, \emph{(edge-partitioned) $\phi$-expander decomposition} is a partition $E=E_1\cup E_2\cup \ldots\cup E_z$ such that $E_i$ is a $\phi$-expander for every $i$ (as a subgraph), and every vertex is contained in at most $\tO{1/\phi}$ many different $E_i$'s (as subgraphs). We mostly care about the case when $\phi=1/\polylog(n)$, so to simplify notation, in this overview, when we say expander decomposition, we mean an edge-partitioned $\phi$-expander decomposition for some $\phi=1/\polylog(n)$.

Note that the expander decomposition operates by ignoring edge directions and is therefore defined on undirected graphs. This simplification is appropriate because the decomposition is primarily used to construct the HeavyHitter data structure, which maintains edge weights and supports weight-based queries, rather than relying on edge directions.

\paragraph{Parallel batch-dynamic expander decomposition.} The problem of parallel batch-dynamic expander decomposition asks for maintaining the expander decomposition of an undirected graph $G=(V,E)$ where $G$ undergoes edge updates (both insertions and deletions). Moreover, the edge updates are given as batches: in each batch, a set of edges is given to the algorithm to be deleted (or inserted), and the algorithm needs to perform this update and give the expander decomposition after this batch update in small work (ideally proportional to the number of updated edges) and small depth (ideally $\polylog(n)$).

% these symbols to not exist with our current packages
%\newcommand\DDelta{\boldsymbol{\mathit{\Delta}}}
%\newcommand\nnabla{\boldsymbol{\mathit{\nabla}}}
\newcommand\DDelta{{\Delta}}
\newcommand\nnabla{{\nabla}}

\paragraph{Flows.} A \textit{flow} $f$ on an undirected graph $G=(V,E)$ assigns non-negative real values to simple paths in $G$, meaning the total value of flow over this path. The size of the flow is the total value among all paths. A source (sink) demand vector $\DDelta$ assigns a real value to each vertex. We say a flow routes a source demand $\DDelta$ to a sink demand $\nnabla$ if at least $\DDelta(v)$ flows are starting from $v$ for every $v$ and at most $\nnabla(v)$ flows are ending at $v$ for every $v$.
%\jan{why delta and nabla? did we use that}

Note that the notion of flow is defined for undirected graphs in the context of the expander decomposition algorithm, where it is exclusively used.

%% file: overview.tex
\section{Overview}\label{sec:overview}
This overview consists of two main parts. In \Cref{sec:overview_IPM}, we give a summary of previous work and show how the interior point method is set up to obtain our main result, \Cref{thm:main}. We describe which subroutines are necessary. In particular, we show that the main technical ingredient is a parallel batch-dynamic expander decomposition. We describe how we obtain the latter in \Cref{sec:overview_expander}.
% Some necessary notation and definitions are provided in \Cref{sec:overview_notation}.

\subsection{A (Parallel) Interior Point Method}\label{sec:overview_IPM}
%In this section, we provide an overview of the algorithm that attains \Cref{thm:main}. 

In this section we recap interior point methods for solving minimum cost flows. The interior point method presented here stems from previous work on sequential algorithms. We outline the framework and the computational tasks that must be solved in the parallel setting. Readers familiar with this framework may skip ahead to \Cref{sec:overview_expander}, where we outline our contributions, i.e., solving the necessary subroutines in a parallel setting.

Consider the linear programming definition of minimum cost flow: for an $m\times n$ incidence matrix $\mA\in\{-1,0,1\}^{m\times n}$, edge-cost vector $c\in\R^m$, demand vector $b\in\R^n$, and edge capacities $u\in\R^m_{\ge 0}$, the following linear program models minimum-cost flows.
\begin{align*}
    \min_{x\in\R^E} c^\top x ~\text{ subject to }~
    \mA^\top x = b,~~
    0 \le x \le u
\end{align*}

A common approach for solving linear programs is central path methods.
These reduce solving linear programs to solving a sequence of linear systems. The idea is to define a potential function $f_\mu$ for $\mu\in\R_{>0}$ such as
\begin{align}
f_\mu:&~\{x\mid \mA^\top x = b\} \rightarrow \R%$$
%$$
~~\text{ where }\notag\\ %~~
f_\mu(x) =&~ c^\top x - \mu\cdot\sum_{i=1}^m \log x_i + \log (u_i - x_i).
\label{eq:simple_step}
\end{align}
Observe that the log-terms go to $-\infty$ as $x_i$ approaches 0 or capacity $u_i$. Hence to minimize $f_\mu(x)$, we need that $x$ stays far from the boundaries $0 \le x_i \le u$.
However, as $\mu\rightarrow 0$, the first term $c^\top x$ starts to dominate, so minimizing $c^\top x$ becomes more important.
In particular, for $\mu\rightarrow 0$, the minimizer $x$ converges towards the optimal solution of the linear program.
The curve of minimizers $c(\mu) = \argmin_{\mA^\top x = b} f_\mu(x)$, is referred to as ``the central path'' as it traces a path from the center of the feasible space towards the optimal solution of the linear program.

The idea of central path methods is to start with an initial point $x$ that is the minimizer of $f_\mu$ for large $\mu$, and then follow the central path by repeatedly decreasing $\mu$ and moving $x$ closer to the new minimizer of $f_\mu$ via Newton steps. 

Once $\mu$ is small enough (i.e., some $1/\poly(m,\|u\|_\infty,\|c\|_\infty)$) then $x$ is close to the optimal solution of the linear program. In particular for min-cost flow, where the optimal solution is guaranteed to be integral, we can simply round all coordinates of $x_i$ to the nearest integer to obtain the optimal min-cost flow.

\paragraph{Solving min-cost flow in $\sqrt{n}$ depth.}

In \cite{lee2020solvinglinearprogramssqrtrank}, it was shown that for a slightly different choice of $f_\mu$, the number of iterations is $\tilde O(\sqrt{n})$.
They use a variation of \eqref{eq:simple_step} where the log-terms are weighted by Lewis-weights $\tau\in\R^m_{\ge0}$.
\begingroup
\allowdisplaybreaks
\begin{align}
f_\mu(x) =&~ c^\top x + \mu\cdot\sum_{i=1}^m \tau_i\cdot\phi(x)_i\notag\\
\phi(x)_i =&~ -\log(x_i) - \log(u_i-x_i), \notag\\%~~
\phi'(x)_i =&~ -\frac{1}{x_i} + \frac{1}{u_i - x_i},\notag\\% ~~
\phi''(x)_i =&~ \frac{1}{x_i^2} + \frac{1}{(u_i-x_i)^2}\notag\\
\tau =&~ \sigma(\mT^{1/2-1/p}\Phi''(x)^{-1/2}\mA) + \frac{n}{m} \label{eq:lewisweight}
\end{align}
\endgroup
Here the $\tau$ that satisfy the recursive equation \eqref{eq:lewisweight} for $p=1-1/(4\log(4m/n))$ are referred to as $\ell_p$ Lewis-weights \cite{CohenP15,lee2020solvinglinearprogramssqrtrank}.
%on intuitive level, lewis weights measure importance of rows of A wrt Lp norm. 
%in case where A is incidence matrix, they thus measure the importance of individual edges.

To measure the distance of $x$ towards the central path, it is useful to introduce a variable $s\in\R^m$ of form $s = c-\mA y$ for $y\in\R^n$. The optimality condition for $f_\mu(x)$ is the existence of $y\in\R^n$ with $0 = \nabla f_\mu(x) - \mA y$ (using that the gradient $\nabla f_\mu$ must be orthogonal to the feasible space $\{x\mid \mA^\top x = b\}$), and thus $ 0 =c + \mu\tau\phi'(x) - \mA y = s +\mu\tau\phi'(x)$.
So the distance to the central path can also be measured via the length of the vector $s + \mu\tau\phi'(x)$ rather than the value of $f_\mu(x)$.
In \cite{lee2020solvinglinearprogramssqrtrank}, it was shown that for certain potential function $\Psi$, measuring the length of $s + \mu\tau\phi'(x)$, the following iterative steps allow convergence towards the optimal solution of the linear program in only $\tilde O(\sqrt{n})$ iterations.
\begin{align}
    x\gets&~ x + \delta_x,~~s\gets s+\delta_s \label{eq:ls_step}\\
    \delta_x =&~\Phi''(x)^{-1/2}g \notag\\
    &~- \mT^{-1}\Phi''(x)^{-1}\mA~(\mA^\top \mT^{-1}\Phi''(x)^{-1}\mA)^{-1}~ \mA^\top\Phi''(x)^{-1/2} g \notag\\
    \delta_s =&~ \mu \mA~(\mA^\top \mT^{-1}\Phi(x)^{-1}\mA)^{-1}~ \mA^\top\Phi''(x)^{-1/2} g \notag\\
    g =&~ \nabla \Psi\left(\frac{s+\mu\tau\phi'(x)}{\mu\tau\sqrt{\phi''(x)}}\right)^{\flat(\tau)} \notag
\end{align}
where we define for vectors $v,\tau\in\R^m$
\begin{align*}
    v^{\flat(\tau)} := \argmax_{\|w\|_{\tau+\infty}\le1} \langle w, v\rangle.
\end{align*}
Essentially, $g$ is the largest step we can take in direction of the gradient $\nabla\Psi$ while the step is still bounded in the $\|\cdot\|_{\tau+\infty}$ norm.
%Here $\Psi$ is a potential function that measures the distance to the central path (i.e, $\Psi$ is 0 when $s-\mu\tau\phi(x) = 0$, and $\Psi$ is large when the vector is non-zero).

Performing steps as in \eqref{eq:ls_step}, then decreasing $\mu$ by some $1-O(1/\sqrt{n}))$ factor, and repeating, takes $\tilde O(\sqrt{n})$ iterations to solve the linear program. In particular, it implies a $\tilde O (\sqrt n)$ depth algorithm for min-cost flow, but with large work to calculate the steps in each iteration.
Thus the next task is to reduce the total amount of work without substantially increasing the depth.

\paragraph{Robust Interior Point Method.}
Observe that even when ignoring the time for calculation, just writing down the vectors $\delta_x,\delta_s$ of \eqref{eq:ls_step} takes $\Theta(m)$ work per iteration.
%calculation of tau is computationally expensive given (i) its recursive definition, and (ii) that it involves solving m linear systems.
To reduce time per iteration, sequential work \cite{CohenLS21,Brand20,LeeSZ19,JiangSWZ21,BrandLLSS0W21,BrandLN+20} has shown that computing the Newton steps approximately suffices. These methods use only approximations of $x$, $s$, $\tau$ and the matrix inverse $(\mA\mT^{-1}\Phi''(x)^{-1}\mA)^{-1}$. %, instead of their precise values.
Concretely, \cite{BrandLLSS0W21} showed the following method converges within $\tilde O(\sqrt{n})$ iterations despite crude approximations (here $\epsilon = O(1/\log m)$):
\begingroup
\allowdisplaybreaks
\begin{align}
    &\text{Pick }\ox \approx_\epsilon x,~~\os \approx_\epsilon s,~~\otau \approx_\epsilon \tau(\Phi''(x)^{-1/2}\mA),~~~\omu \approx_\epsilon \mu \label{eq:first_step}\\
    &g = -\gamma \nabla \Psi\left(\frac{\os+\omu\otau\phi'(\ox)}{\omu\otau\sqrt{\phi''(\ox)}}\right)^{\flat(\otau)} \notag\\
    &\mH \approx_\epsilon \mA^\top \omT^{-1} \Phi''(x)^{-1} \mA ~~~~\text{(spectral sparsifier)} \notag\\
    &\delta_c = \mH^{-1}(\mA^\top x - b),~~ %\\
    \delta_y = \mH^{-1} \mA^\top \mPhi''(x)^{-1/2}g \notag\\
    &\mR = \text{random $m\times m$ diagonal matrix, }\notag\\
    &\mR_{i,i} = \begin{cases}
        1/p_i & \text{with probability } p_i\\
        0 & \text{else}
    \end{cases} \label{eq:sparsifydelta}\\
    &~~~\text{where }p_i \ge \min\left(1,~\frac{m}{\sqrt{n}}\cdot \frac{((\omT\Phi''(\ox))^{-1}\mA (\delta_y+\delta_c))_i^2}{\|(\omT\Phi''(\ox))^{-1}\mA (\delta_y+\delta_c)\|_2^2}+\frac{1}{\sqrt{n}}+\otau_i \right)\notag\\
    &\delta_x = \Phi''(x)^{-1/2}g - \mR\omT^{-1}\Phi''(x)^{-1} \mA (\delta_y + \delta_c),~~ %\\
    \delta_s = \mu \mA \delta_y \notag\\
    &x\gets x + \delta_x,~~~ s\gets s+\delta_s,~~~ \mu \gets \mu (1-\tilde O(1/\sqrt{n})) \notag\\
    &\text{Repeat from \eqref{eq:first_step}} \notag
\end{align}
\endgroup
Unlike \eqref{eq:ls_step}, here the additional term $\delta_c$ is needed, because by using the spectral approximation $\mH$, we no longer have $\mA^\top \delta_x =0$. Thus $\mA^\top (x+\delta_x) \neq b$ which is corrected with the $\delta_c$ step.
To support fast calculation of $\mA^\top x-b=\mA x^{(\rm{old})} + \mA^\top \delta_x$, part of the vector $\delta_x$ is sparsified via the random matrix $\mR$ in \eqref{eq:sparsifydelta}.

It is important for \eqref{eq:first_step} that $x,s,\tau$ are not actually computed and only \emph{defined} as reference points for the approximation $\ox,\os,\otau$. Only these approximate values are computed by the algorithm.
This allows for (amortized) sublinear-work per iteration, because one can prove that the vectors $\ox, \os, \otau$ require only a few changed entries per iteration to stay valid approximations. So there is no need to recompute $m$ entries of $\ox, \os ,\otau$ from scratch in each iteration.
%, it suffices update only a few indvidiual entries in each iteration.
By developing data structures tailored to this task of updating entries of $\ox,\os,\otau$, \cite{BrandLLSS0W21} obtained a sequential algorithm that solves min-cost flow in $\tilde O (m+n^{1.5})$ time.
However, their result does \emph{not} imply an $\tilde O( \sqrt n)$-depth algorithm, because their data structures need up to $\Theta(m)$ depth in some iterations. 
Our main contribution is the development of low-depth equivalents of their data structures.
The major bottlenecks we need to solve are as follows. 

\paragraph{Sampling $\mR$.}
Computing the sampling probabilities for $\mR$ as in \eqref{eq:sparsifydelta} would require a matrix vector product of form $\mW\mA h$ for some vector $h\in\R^n$ and $m\times m$ diagonal matrix $\mW$.
This calculation would require $O(m)$ time per iteration.
However, let us assume for simplicity that $\mA$ is incidence matrix of an expander graph and $\mW_{i,i}=\mW_{j,j}$ for all $i,j$. We will argue that for this case, the sampling task is simple.

Assume without loss of generality that $h \bot \mD\mathbf{1}$ where $\mD \in \R^{V\times V}$ is a diagonal matrix with $\mD_{v,v}=\deg(v)$ (we can add a multiple of the all-1-vector to $h$ to satisfy this. This does not change the task as $\mA \mathbf{1} = 0$).
Iterate over the vertices, and for each $v\in V$ sample each incident edge $\{u,v\}$ independently with probability proportional to $h_v^2 > |h_u - h_v|^2/2 = (\mA h)_{uv}$.
This sampling can be implemented to take time proportional to the number of returned edges, i.e., $\sum_{v\in V} 1+ \deg(v)\cdot h_v^2 \le \tilde O(n+ \|\mA h\|_2^2)$ in expectation by Cheeger-inequality (which states 
$\phi^2_G/2 \le \lambda_2(\mD^{-1/2}\mL\mD^{-1/2})$, and thus bounds
$\sum_{v\in V} \deg(v)\cdot h_v^2 = \|\mD^{1/2} h\|_2 \le \tilde O((\mD^{1/2}h^\top) (\mD^{-1/2}\mL\mD^{-1/2})(\mD^{1/2} h)) = \tilde O(h^\top \mL h)$ $= \tilde O(\|\mA h\|_2^2)$ for $1/\polylog(n)$ expander graph and $h\bot \mD\mathbf{1}$, i.e., $\mD^{-1/2}h$ is orthogonal to the $(\lambda_1$$=$$0)$-eigenspace of $\mD^{-1/2}\mL\mD^{-1/2}$).

Thus, if we can decompose the weighted graph represented by $\mW\mA$ into a collection of expanders such that each expander has (almost) uniform edge weights, then sampling can be done efficiently without computing $\mW \mA h$ in $O(m)$ time.
%
%
%calculating these probabilities would take m time per iteration.
%however, assume that incidence matrix A is an expander.
%then ... explain why sampling is simple.
%
%in general A is not an expander, but if we can reduce the problem ot this case via expander decomposition.
%given an expander decomposition G= union Gi where each Gi is expander, then this is equivalent to decomposing mxn matrix A into smaller mi xn matrices, each of which is the inidence matrix of an expander.
%since sum L2norm squared = L2normsquared, we can thus perform the sampling task on each expander independently.
%
Hence the main problem is to maintain an expander decomposition for the graph where edge weights are $w:=(\otau\phi''(\ox))^{-1}$, and thus change from one iteration to the next.
In previous work, this expander decomposition could not be done in low depth.
We solve this issue by developing a parallel batch-dynamic expander decomposition that is efficient in the parallel setting, see \Cref{sec:overview_expander}.

\paragraph{Maintaining $\os,\otau$.}
We also need to compute $\os$ and $\otau$ in each iteration. Since they are $m$-dimensional vectors, we cannot afford to recompute them in each iteration.
Instead, we aim to only update a few of their entries per iteration. We outline the idea here for $\os$ but it extends to $\otau$ as well.

If we already have some approximation $\os \approx s$, and perform one more iteration $s \gets s+\delta_s$, then $\os_i$ is still a good approximation of $s_i$ if $(\delta_s)_i \ll s_i$ was sufficiently small.
In particular, instead of computing the entire $m$-dimensional vector $\delta_s$ in each iteration, it suffices to only calculate information about entries $(\delta_s)_i$ where $(\delta_s)_i\gg \epsilon \cdot s_i$ for some threshold $\epsilon$.

We thus reduce the problem to a data structure task where, given a vector $\delta_y\in\R^n$, we must detect all indices $i$ where $(\omS^{-1}\delta_s)_i = (\omS^{-1} \mA \delta_y)_i > \epsilon$.
This too can be solved easily when incidence matrix $\omS^{-1}\mA$ represents an expander graph. E.g., it could be solved by repeatedly sampling indices proportional to $\omS^{-1}\mA \delta_y$, which reduces the task to the previously outlined sampling problem. %We also give a deterministic alternative in \Cref{sec:HeavyHitter}.
In summary, for an efficient parallel algorithm for min-cost flow, all that is left is to develop a dynamic expander decomposition of low-depth and work.

\subsection{Expander Decomposition}\label{sec:overview_expander}

In this section, we provide an overview of our algorithm for parallel batch-dynamic expander decomposition (\Cref{sec:expander}), i.e., we want to maintain an expander decomposition\footnote{As defined in \Cref{sec:preliminaries}, in this overview, for simplicity, we use \emph{expander decomposition} to refer to decomposing the edge set of a graph into subgraphs where each subgraph is an \emph{expander}, and \emph{expander} refers to a $\phi$-expander for some $\phi=1/\polylog(n)$. It is guaranteed that each vertex appears in at most $\tO{1}$ subgraphs.} of an \textit{undirected} graph while this graph undergoes edge insertions and deletions in batches, and for each batch update of edge set $E'$, we wish to maintain the expander decomposition in $\tO{1}$ depth and amortized $\tO{|E'|}$ work. See \Cref{lem:dynamicExpanderDecomposition} for a detailed definition. 

\paragraph{Reduction to decremental setting.} It is known that in the sequential setting, fully dynamic expander decomposition (edges can be both inserted and deleted) can be reduced to two subroutines, (i) the static setting (compute an expander decomposition of an undirected graph once) and (ii) the decremental setting (edges can only be deleted). The reduction can be found in \cite{BernsteinBGNSS022}. This reduction can be naturally implemented in the \pram{} model under batch updates (see \Cref{sec:expander}, Proof of \Cref{lem:dynamicExpanderDecomposition} for more details). In a nutshell, we partition the edge set $E$ into $O(\log n)$ edge sets $E_1,...,E_{O(\log n)}$ where $|E_i|<2^i$, and the algorithm will maintain expander decomposition for every $E_i$. When a batch of edges $E'$ is inserted into the graph, the algorithm will find the smallest $i$ such that $2^i>|E'|+\sum_{j\le i}|E_i|$ and use the static algorithm to recompute an expander decomposition of $\cup_{j\le i}E_i\cup E'$ and make it the new $E_i$. Notice that the static algorithm takes work proportional to the number of edges inserted. Thus, we only need to handle the static case and the decremental case.

For the static computation of expander decomposition, a recent paper \cite{ChenMGS25} gives a near-linear time and polylogarithmic depth algorithm, so it is done. For the rest of this section, we focus on giving the decremental algorithm.

\paragraph{Reduction to expander pruning.} It is known from \cite{SaranurakW19} that decremental expander decomposition can be reduced to the following expander pruning task.

\textbf{Input.} An expander graph $H$.

\textbf{Updates.} An online sequence of batch deletions $E_1,...,E_k$. Denote $H_i=H-\cup_{j\le i}E_i$ the graph after the $i$-th deletion.

\textbf{Outputs.} After each update, a set of \emph{pruned} vertices $V_i$ such that after deleting $\cup_{j\le i}V_i$ from $H_i$ (and all their adjacent edges), the graph becomes an expander again. It is required that the number of adjacent edges in $V_i$ in $H_i$ is at most $\tO{|E_i|}$, i.e., at most $\tO{|E_i|}$ edges can be pruned out after the $i$-th update.

We also wish to perform the $i$-th update in $\tO{|E_i|}$ amortized work and $\tO{1}$ depth. 

In a nutshell, the reduction works as follows. In the decremental setting, i.e., deleting an edge set $E'$ from an expander decomposition, we `prune out' sets of vertices (and adjacent edges) from each decomposed expander of total size $\tO{|E'|}$ so that the remaining subgraphs are still expanders, where the `pruned out' edges are inserted into the graph again which is handled by edge insertion in the previous paragraph (remember that edge insertion can be handled by static computation of expander decomposition).

For the rest of this section, we only need to consider parallel expander pruning, the formal definition is in \Cref{thm:pruning}.

\paragraph{Trimming.} In \cite{SaranurakW19}, expander pruning is based on modifying a procedure called \emph{trimming}, which can be viewed as a static version of expander pruning: a set of edges $E'$ is deleted from an expander $H$, and we want to delete another set of edges of size $\tO{|E'|}$\footnote{In the previous paragraph, we mentioned deleting a set of vertices, but that is equivalent to deleting all the adjacent edges and for convenience here we write deleting a set of edges.} so that the remaining subgraph is still an expander. Notice the difference: expander pruning requires many batches of deletions where expander trimming only requires one batch of deletions.

However, the trimming algorithm in \cite{SaranurakW19} is highly sequential. In a more recent paper~\cite{ChenMGS25}, a parallel trimming algorithm is provided, with $\tO{m}$ work and $\tO{1}$ depth. The challenge is to use it to get expander pruning, i.e., instead of just deleting one batch of edges from $H$, it needs to support deleting an online sequence of edge sets $E_1,E_2,...,E_k$, where each deletion only uses $\tO{|E_i|}$ work and $\tO{1}$ depth. This is not so obvious to do as the algorithmic ideas from \cite{ChenMGS25} are quite different from \cite{SaranurakW19}.

\paragraph{Algorithm overview of \cite{ChenMGS25}.} We first give a brief overview of the trimming algorithm in \cite{ChenMGS25}. Given an expander $H$ and an edge set $E'$ to be deleted from $H$, it outputs an edge set of size $\tO{|E'|}$ deletions such that after these deletions $H$ is an expander again. The algorithm uses $\tO{m}$ work and $\tO{1}$ depth. 

%Before stating the algorithm, it is necessary to give definitions of flow. A flow $f$ supported by $H$ is a collection of path $\cP$ in $H$, a source demand (or sink demand) is a vector $\nabla\in \mathbb{N}^V$. We say $f$ routes the source demand $\nabla$ if the number of paths in $\cP$ that start at $v\in V$ is exactly $\nabla(v)$ for every $v\in V$. Similarly, we say $f$ routes to the sink demand $\nabla$ if the number of paths in $\cP$ that end at $v\in V$ is exactly $\nabla(v)$ for every $v\in V$. We say $\cP$ has congestion $c$ if the maximum number of paths passing through an edge is $c$. 

The intuition of the algorithm comes from a structural lemma (see \Cref{lem:certificate}), for which we first define a \emph{certificate}. We call a flow $f$ supported by $H-E'$ a \emph{certificate} if: (i) $f$ routes the source demand $\DDelta(u)=\tO{\deg_{E'}(u)}$\footnote{Here we use $\deg_{E'}(u)$ for the edges in $E'$ that are adjacent to $u$.}, i.e., each edge $(u,v)\in E'$ contributes roughly $\tO{1}$ source demands on both $(u,v)$, (ii) $f$ routes the sink demand $\nnabla(u)=\deg_H(u)/\polylog(n)$. \Cref{lem:certificate} shows that if a certificate exists, then $H-E'$ is an expander. 

A natural idea is to solve a max flow problem to try to find a certificate in case $H-E'$ is already an expander and we do not need to delete any further. One worry is that we might not be able to solve the max flow problem in $\tO{m}$ work and $\tO{1}$ depth. However, \cite{ChenMGS25} shows that we do not need to solve the flow problem perfectly but just need to use the push relabel algorithm up to $\tO{1}$ layers. After which, either we route all the demand, which certifies that $H-E'$ is an expander, so we are done, or we can find a set of edges to be deleted from $H-E'$, denoted by $E''$, of size $\tO{|E'|}$. It is guaranteed that if we substitute $E'$ by $E'\cup E''$, we can set up the flow problem again and the source demand size reduces by a constant factor. Thus, $\tO{1}$ iterations suffice to finally find a certificate, and the remaining graph is an expander. The flow algorithm is described in detail in \Cref{subsec:parallelunitflow}.

\paragraph{A better analysis of the algorithm.} Although in \cite{ChenMGS25}, the stated work is $\tO{m}$, we can actually give a better analysis to make the work $\tO{|E'|}$. The intuition is that the sink demand for every vertex $u$ is $\deg_H(u)/\polylog(u)$, so the flow cannot go out of $u$ unless roughly $\deg_H(u)$ amount of flow is absorbed by $u$. In other words, if an edge $(u,v)$ is in the support of the flow, then either $u$ or $v$ is saturated so that we can charge the size of the support by the total source demand size, which is $\tO{|E'|}$. The algorithm has work proportional to the support of the flow as it is using push-relabel.

\paragraph{Supporting decremental updates.} Now, we explain how we convert the trimming algorithm into an pruning algorithm supporting many batches of edge deletions. Suppose that after the first batch of edge deletions $E_1$ arrives, we call the trimming algorithm on $H,E_1$ and get a set $E'_1=\tO{|E_1|}$ such that $H-E_1-E'_1$ is an expander. A natural idea is to substitute $H$ by $H_2:=H-E_1-E'_1$ and run the algorithm again on $H_2$ with the next batch of deletion $E_2$. However, there is a technical issue that makes this simple idea wrong: the trimming algorithm has an inherent property that the expansion of $H_2$ decreases by a constant factor compared to $H$, which means after $k$ updates, the expansion becomes exponentially smaller in terms of~$k$. To avoid this, instead of replacing $H$ by $H_2$ and feeding it into the trimming algorithm, we will try to run the trimming algorithm again on $H,E_1\cup E_2$. This results in high work since the work for the second update is $\tO{|E_1\cup E_2|}$ instead of $\tO{|E_2|}$. We resolve this in a similar way to \cite{SaranurakW19}, i.e., by reusing the flow information from the first trimming algorithm on $H,E_1$ so that we only need to route the flow contributed by $E_2$. Let us call the flow instance routing the contribution of $E_i$ by $f_i$. The certificate after deleting $E_1,...,E_i$ from $H$ will be $f_1+...+f_i$. 

However, this idea has an additional issue in the parallel setting. The parallel trimming algorithm by \cite{ChenMGS25} has good parallel work and depth mainly because each time the sink demand for flow $f_i$ is relatively large, i.e., at least $\deg(u)/\polylog(n)$. That means that the flow cannot go out from $u$ unless it gets absorbed by an amount proportional to the degree of $u$. Also, remember that we compute each $f_i$ using the idea from \cite{ChenMGS25}, which means $f_i$ must have sink demand at least $\deg(u)/\polylog(n)$ for every $i$. 
Thus, the total flow summing over all $f_i$ is polynomially larger than $\deg(u)$ if the number of updates are polynomial. This \emph{violates} the definition of a certificate flow.
In other words: we can only support the batch updates up to $\polylog(n)$ times.

In order to support an arbitrary number of batch updates, we prove a lemma of batch number boosting (see \Cref{lem:fully_framework} for more details). Similar ideas appear in \cite{NanongkaiSW17,0001S21}. The boosting is simple: every time the algorithm undergoes $2^i$ batch updates, it rolls back to the beginning and combines all the $2^i$ updates into a single batch update, and does the $2^i$ to $2^{i+1}$ updates using similar rolling back techniques. Finally, we can support polynomially many decremental batch updates. 

\paragraph{Vertex Decomposition.} Although we presented our algorithm for maintaining an (edge-partitioned) expander decomposition, the same algorithm should work for the more prevalent vertex-partitioned expander decomposition, where the vertex set is partitioned into $V_1,...,V_z$ such that each induced subgraph is an expander, and there are at most $\tO{1/\phi}\cdot m$ inter-cluster edges. This is because expander pruning will give a pruned \emph{vertex} set instead of and \emph{edge} set, and all the arguments above should work. Note that since our result specifically concerns edge-partitioned decompositions, the vertex-partitioned variant is not directly relevant to our main result.

\subsection{Organization}\label{overview:organization}
In \Cref{sec:expander}, we provide full details for our parallel batch-dynamic expander decomposition. 
In \Cref{sec:alg}, we describe our outer algorithm in more detail and prove our main result, \Cref{thm:main}. The data-structures and subroutines we need for this are all parallel implementations of the version in \cite{BrandLLSS0W21}. We provide full details in the appendix, more precisely:
We provide parallel primitives such as linear system solving and maintaining lists in \Cref{sec:parallelization}. 
In \Cref{sec:HeavyHitter}, we provide our heavy hitter data structure. 
In \Cref{sec:MaintainingRegularizedLewis-Weights}, we show how to main the regularized Lewis-weights. 
In \Cref{sec:primal}, we show how to maintain an approximation to the primal solution and the gradient of the potential function. 
Finally, in \Cref{sec:graphDS}, we provide additional data structures needed, including an algorithm that maintains an approximate dual solution.

%% file: expander.tex
\section{Parallel Batch-Dynamic Expander Decomposition}\label{sec:expander}

\newcommand\veczero{\boldsymbol{0}}
\newcommand\vecone{\boldsymbol{1}}
\newcommand\ff{\boldsymbol{\mathit{f}}}
\renewcommand{\deg}{\mathbf{deg}}
\newcommand{\ex}{\mathbf{ex}}
\SetKw{KwBreak}{break}

In this section, we will prove the following lemma. It is a parallel version of Theorem 4.3 in \cite{BernsteinBGNSS022}.

\begin{lemma}
	\label{lem:dynamicExpanderDecomposition}
        There exists a randomized data structure against an adaptive adversary that, given an undirected graph $G=(V,E)$ (initially empty) and $\phi<1/\log^C n$ for sufficiently large constant $C$, maintains subgraphs $G_1,...,G_t$ of $G$. It is guaranteed that $G_1,...,G_t$ partitions the edge set of $G$, $G_i$ is a $\phi$-expander for every $i$ with high probability, and $\sum_{i}|V(G_i)|=\tO{n}$. The data structure supports batch updates, i.e., given a set of edges $E'$ to be deleted or added to $G$, the update uses $\tO{1/\phi^4}$ depth and amortized $\tO{|E'|/\phi^5}$ work. 
        %\jan{in this parallel model, how is the output defined? i just want to make sure that the mere act of "printing" the output/changes to $G_i$ doesn't break the depth statement.}\yonggang{$G_i$'s are maintained in the memory as adjacent lists. In the literature there are fundamental data structures like parallel inserting/deleting batches of elements of a sorting list, so insertion/deletion of edges/vertices from the adjacent list representations of these graphs in parallel shouldn't be a problem.}
\end{lemma}

Notice that by following the convention, expander decomposition means decomposing the vertex set instead of the edge set, see the theorem below.

\begin{theorem}[Parallel Expander Decomposition, Theorem 1 of \cite{ChenMGS25}]\label{thm:exp-decomp}
Given a graph $G=(V,E)$ of $m$ edges and a parameter $\phi\in(0,1)$, there is a randomized parallel algorithm that with high probability finds a partition of $V$ into $\phi$-expanders $V_1,..,V_k$ such that $\sum_{i 
< j} |E_G(V_i, V_j)|=\widetilde{O}(\phi m)$. The total work of the algorithm is $\widetilde{O}(m/\phi^2)$ with depth $\widetilde{O}(1/\phi^4)$.
\end{theorem}

By repeatedly applying \Cref{thm:exp-decomp}, we can get a static algorithm for \Cref{lem:dynamicExpanderDecomposition}. %\jan{Cref instead of cref. Theorem X, Lemma Y are capitalized because they are the names of that theorem/lemma.} 
To get the dynamic version, we need the following parallel expander pruning lemma, it is the parallel version of Theorem 1.3 in \cite{SaranurakW19}. 

\begin{lemma}[Parallel expander pruning]\label{thm:pruning}

There is a deterministic data structure that, given an undirected graph $G=(V,E)$ which is a $\phi$-expander, and
an online sequence of edge set deletions $E_1,E_2,...,E_k$, maintains a \emph{pruned set} $P\subseteq V$ such that the following property holds. Let $G_i$ and $P_{i}$ be the graph $G$ and the set $P$ after the $i$-th deletion.
We have, for all $i$, 
\begin{enumerate}
\item $P_{0}=\emptyset$ and $P_{i}\subseteq P_{i+1}$,
\item $\vol(P_{i})\le \tO{\sum_{j<i}|E_j|/\phi}$, and
\item $G_i[V-P_{i}]$ is a $\phi/(6\log n)$-expander. 
\end{enumerate}
There is no initialization required. Assuming the graph is given by an adjacency list, the work and depth for updating each $P_{i}$ is $\tO{|E_i|/\phi^{4}}$ and $\tO{1/\phi^3}$ for every $i$. 
\end{lemma}

Proving \Cref{thm:pruning} is the main technical contribution of this section. The algorithm in \cite{SaranurakW19} is highly sequential. We will adjust the parallel trimming algorithm in \cite{ChenMGS25} to get our result, see \Cref{subsec:parallelexpanderpruning}. For completeness, let us first see how we can combine \Cref{thm:exp-decomp} and \Cref{thm:pruning} to get \Cref{lem:dynamicExpanderDecomposition}, by following the same algorithm as in \cite{BernsteinBGNSS022}.

    To prove \Cref{lem:dynamicExpanderDecomposition}, let us first prove the static version:
    \begin{lemma}\label{lem:staticofdyanmicexpander}
        There exists a randomized algorithm that given an undirected graph $G=(V,E)$ and $\phi<1/\log^Cn$ for sufficiently large constant $C$, outputs subgraphs of $G$ denoted by $G_1,...,G_t$. It is guaranteed that $G_1,...,G_t$ partitions the edge set of $G$, $G_i$ is a $\phi$-expander for every $i$ with high probability, and $\sum_{i}|V(G_i)|=\tO{n}$. The algorithm runs in $\tO{m/\phi^2}$ work and $\tO{1/\phi^4}$ depth.

    \end{lemma}
    \begin{proof}
        We use \Cref{thm:exp-decomp} on $G$ and $\phi$ to get $V_1,...,V_k$, and let $G[V_i]$ for every $i$ to be the output subgraphs. Then we delete the edges in $G[V_i]$ for every $i$ from the graph, and we are guaranteed that the remaining number of edges is at most $\tO{\phi m}<m/2$, according to \Cref{thm:exp-decomp} and $\phi<1/\log^Cm$ for sufficiently large $C$. We repeat the same procedure on the remaining graph, after $O(\log n)$ iterations the graph will be empty and we have successfully partitioned the edge set into $\phi$-expanders. Moreover, each vertex is contained in at most $O(\log n)$ many subgraphs since in each iteration, the vertex sets of the subgraphs generated in this iteration are disjoint.
    \end{proof}

\begin{proof}[Proof of \Cref{lem:dynamicExpanderDecomposition}]
    %Now we are ready to give the dynamic algorithm. 
    The data structure maintains $O(\log n)$ graphs $G_1,G_2,...$ so that they are subgraphs of the dynamic graph $G$, their union is equal to $G$, and $|E(G_i)|\le 2^i$. For each $G_i$, we will maintain a partition of the edge set of $G_i$ into expanders $G_{i,1},..,G_{i,k_i}$ such that each vertex appears in $O(\log n)$ many of them. If we can maintain that, we get the desired partition of the whole graph $G$.
    
    When an edge set is inserted into the graph, we will first insert them into $G_1$ and perform some operations on $G_1$; when an edge set is deleted, we will split the deleted edge set into $O(\log n)$ sets and delete the corresponding edge set in each $G_i$. In the next two paragraphs, we describe how we handle the insertion or deletion of an edge set from a specific $G_i$. 

    If some set of edges $I$ is to be inserted into $G_i$, then we consider two cases: (i) If $|E(G_i)\cup I|>2^i$, then we set $G_i$ to be an empty graph and insert $E(G_i)\cup I$ into $G_{i+1}$. (ii) If on the other hand $|E(G_i)\cup I|\le 2^i$, then we perform the algorithm of \Cref{lem:staticofdyanmicexpander} on $G_i$ to obtain a partition of $G_i$ into $(6\log n\phi)$-expanders (which is also a $\phi$-expander), and re-initialized the pruning algorithm for each partitioned subgraph.

    If some set of edges $I$ is to be deleted from $G_i$, we use the pruning algorithm of \Cref{thm:pruning} on each of the $(6\log n\phi)$-expander of the partitioned subgraphs of $G_i$. In that way, we delete some nodes from $G_i$, along with their adjacent edges. For all of those edges, we insert them into $G_1$ again. Notice that we can run the pruning algorithm on $G_i$ is because we can guarantee that there will only be deletions on $G_i$ without insertions: once an insertion happens in $G_i$, according to the last paragraph, either the whole $G_i$ will be renewed into another partition of expanders, or will be deleted. 

    The correctness and depth of each update follow from \Cref{thm:pruning,lem:staticofdyanmicexpander} and the definition of the algorithm. Next, we analyze the amortized work. 

    For each $G_i$, we define a \emph{life-time} of $G_i$ to start from one insertion into $G_i$ and to end at the next insertion into $G_i$. We first show that the number of edge updates during a life-time of $G_i$ is at least $2^i/(\phi\log^Cn)$ for some sufficiently large constant $C$: if $G_i$ gets some edges to be inserted, that must be from the edges in $G_{i-1}$, when $G_{i-1}$ becomes larger than $2^{i-1}$ according to the definition of insertion. Analogously, this is due to $G_{i-2}$, down to $G_1$. So all $G_1$ to $G_{2^{i-1}}$ must be empty for each insertion to $G_i$, thus, between two insertions, at least $\Omega^{2^i}$ edges are inserted. These insertions can be from a direct insertion to $G$ or from the pruning procedure, which according to \Cref{thm:pruning} must correspond to at least $\tO{1/\phi}$ multiplied by the number of deleted edges. Thus, the number of edge update in a life-time is $2^i/(\phi\log^Cn)$. Next, we show that the work done during a life-time is small: it corresponds to one call to \Cref{lem:staticofdyanmicexpander} on $2^i$ edges or simply deleting these edges, and then multiple calls to \Cref{thm:pruning} which together cost at most $\tO{2^i/\phi^4}$ work. Thus, the amortized work is at most $\tO{1/\phi^5}$.
\end{proof}

\subsection{Parallel Expander Pruning (Proof of \Cref{thm:pruning})}\label{subsec:parallelexpanderpruning}

Later we will show that we can only get expander pruning for $\tO{1}$ batch updates at most, instead of supporting an arbitrary number of updates. Luckily, this is not a problem according to the following batch number boosting lemma. Similar ideas were used in \cite{NanongkaiSW17,0001S21}.

First we define some notation. A data structure (decremental, incremental, or fully dynamic) requires preprocessing, then can undergo a bunch of \emph{batch updates} where each update is guaranteed to have some complexity, in our case work and depth. The data structure has \emph{batch number} $b$ if it can only support $b$ batch updates, after which the data structure cannot support any updates. 

\def\dsinitialize{\textsc{DS-Initialize}}
\def\dsupdate{\textsc{DS-Update}}
\def\rti{t_{\mathrm{pre}}}
\def\rtu{w_{\mathrm{amor}}}
\def\parametertimes{\xi}
\def\parameterlength{w}
\def\bat{b}
\def\nbat{\bar{b}}
\begin{lemma}[Batch number boosting]\label{lem:fully_framework}
Let $G$ be a graph undergoing batch updates (decremental, incremental, or fully dynamic). 
Assume there is a data structure $\mathfrak{D}$ maintaining some graph properties of $G$ with batch number~$\bat$, preprocessing time~$\rti$, amortized update work~$\rtu$, update depth~$d$.
Here $\bat,\rti,\rtu$ and $d$ are functions that map the upper bounds of some graph measures throughout the update, e.g.\ maximum number of edges, to non-negative numbers.

Then, for an arbitrary function $\nbat$, there is a data structure (decremental, incremental, or fully dynamic, correspondingly) undergoing batch updates with batch number~$\nbat$, preprocessing time~$\rti$, amortized update work $O(\bat\cdot (\nbat)^{1/\bat}\cdot \rtu)$ and update depth $O(d\cdot\bat+\log_{\bat}\nbat)$, maintaining the same graph properties of $G$ as $\mathfrak{D}$.
\end{lemma}
\begin{proof}

    The case for $\nbat\le \bat$ is trivial. Let us consider the instance where $\nbat> \bat$. The data structure first initializes the data structure $\mathfrak{D}$. When we update the data structure $\mathfrak{D}$ we will memorize our updates so that in the future we can make an operation called \emph{rollback}, which will restore the data structure to a previous state where the updates have not been made.
    
    Define $D_k=(\nbat)^{k/\bat}$. When the $i$-th batch update arrives, the data structure checks for every integer $k>0$ if there exist integers $a$ and $b$ such that $i=a\cdot D_k+b\cdot D_{k-1}$ with $b<D_1=D_k/D_{k-1}$ or not. It finds the largest such $k>0$ (it must exist since when $k=1$, $D_{k-1}=1$ and $a,b$ are simply the module of $D_1$) and rollback the data structure to the point where the $a\cdot D_k$-th update has just been made. Then it updates the batch combining all the updates from $a\cdot D_k+1$ to $i$. 

    %After that, if no such $k$ is found, the data structure find the largest integer $j$ such that $i>j\cdot D_0$, rollback $\mathfrak{D}$ to the point where the $j\cdot D_0$-th update has just been made, and then update the batch combining all the updates from $j\cdot D_0+1$ to $i$.

    \paragraph{Correctness.} Notice that after the arriving of the $i$-th update, the data structure maintains the correct $\mathfrak{D}$ which undergoes the updates combined by all the $1$ to $i$-th updates, by induction on the correctness. Moreover, $\mathfrak{D}$ has always been updated by at most $\bat$ batches: according to the definition of the data structure, if a batch update is generated by rolling back to the $a\cdot D_k$-th update, then the previous batch update must be rolling back to the $a'\cdot D_{k+1}$ for some $a'$, which means the batch size goes up by $D_1$ at each time, so that the maximum batch number is at most $\log_{D_1}(\nbat)=\bat$.

    \paragraph{Complexity.} Finding the $k$ and finding $a,b$ costs $O(\log_{\bat}\nbat)$ work and $O(\log_{\bat}\nbat)$ depth. Then for the largest $k$, rolling back the data structure costs depth $O(d\cdot b)$ since as we previously proved, the batch number for the data structure is at most $\bat$ and rolling back one batch costs $d$ depth as it requires to reverse all the operations made by $\mathfrak{D}$. Now we show the amortized work. Notice that the rolling back procedure can be charged to the update procedure, so we only need to analyze the updates. We only need to prove that for every $i$, the $i$-th update batch gets involved in at most $O(b\cdot (\nbat)^{1/b})$ many batch updates for $\mathfrak{D}$. For that, we only need to prove that for every $k$ (at most $b$ different $k$'s), it is involved in at most $O(\nbat)^{1/b})$ updates. Now fix $k$ and let $a$ be the largest integer such that $a\cdot D_k<i$. The $i$-th update can only be involved in the batch update from $a\cdot D_k+1$ to $a\cdot D_k+b\cdot D_{k-1}$ for $b<D_1=(\nbat)^{1/b}$, so we are done. 
\end{proof}

Now we give the data structure for a small number of batch updates, captured by the following lemma.

\begin{lemma}\label{thm:pruninglowbatch}
There is a deterministic data structure that, given an undirected graph $G=(V,E)$ which is a $\phi$-expander, and
an online sequence of edge set deletions $E_1,E_2,...,E_b$ with $b<(\log n)/2$, maintains a \emph{pruned set} $P\subseteq V$ such that the following
properties holds. Let $G_i$ and $P_{i}$ be the graph $G$ and the set $P$ after the $i$-th deletion.
We have, for all $i$, 
\begin{enumerate}
\item $P_{0}=\emptyset$ and $P_{i}\subseteq P_{i+1}$,
\item $\vol(P_{i})\le \tO{\sum_{j<i}|E_j|/\phi}$, and
\item $G_i[V-P_{i}]$ is a $\phi/6\log n$-expander. 
\end{enumerate}
There is no initialization required. Assuming the graph is given by an adjacency list, the work and depth for updating each $P_{i}$ is $\tO{|E_i|/\phi^{4}}$ and $\tO{1/\phi^3}$ for every $i$. 

\end{lemma}

Our proof of \Cref{thm:pruninglowbatch} will make use of the following \Cref{lem:trimmingcorrectness}, proven in \Cref{subsec:trimming}, which is adapted from the trimming algorithm in \cite{ChenMGS25}.

\begin{restatable}{lemma}{trimmingcorrectness}\label{lem:trimmingcorrectness}
        The algorithm $\textsc{Trimming}(G = (V, E), A, \phi)$ with inputs 
        \begin{enumerate}
            \item a graph $G = (V, E)$ accessed by adjacency list, 
            \item a set $A \subseteq V$ accessed by identity query, i.e., given a vertex, it answers whether it is in $A$ or not, a parameter $\phi \in \R_{\geq 0}$, 
            \item explicitly given $E(A, V \setminus A)$, such that $|E(A, V \setminus A)| \leq \phi \cdot m$,
            
        \end{enumerate}
        outputs explicitly the set $V-A'$ and a flow $\ff$ with non-zero entries explicitly given, such that 
    \begin{enumerate}
        \item $\ff$ restricted in $A'$ routes source $\frac{2}{\phi}(\deg_G(v)-\deg_{G[A']}(v))$ to sinks $\nnabla(v)\le \deg_G(v)/\log n$,
        \item $\vol_G(A-A')=\tO{\frac{1}{\phi}}|E(A, V \setminus A)|$,
        %\item $|E(A', V \setminus A')| \leq 2 \cdot |E(A, V \setminus A)|$
    \end{enumerate}
    in deterministic $\tO{|E(A,V\setminus A)|/\phi^4}$ work and $\tO{1/\phi^3}$ depth.
    \end{restatable}

\begin{proof}[Proof of \Cref{thm:pruninglowbatch}]
    Notice that once we have $\sum_{j<i}|E_i|>\phi\cdot m/\log n$, we can let $P_i=V$ and we are done. So throughout the algorithm, we will assume $\sum_{j<i}|E_i|\le \phi\cdot m/\log n$.

    Recall that $G_i$ and $P_{i}$ are the graph $G$ and the set $P$ after the $i$-th deletion. Define $V_i:=V-P_{i}$. When $E_i$ is deleted, the data structure sets up the virtual graph $G'_i$ by starting at the induced subgraph $G_{i-1}[V_{i-1}]$, and inserting a node in the middle of each edge $(u,v)\in E_i\cap E(G_{i-1}[V_{i-1}])$, i.e., it builds a new node $e$ with edges $(u,e),(e,v)$ and deletes the original edge $(u,v)$. Intuitively, this step is to ensure that $G'_i[V_{i-1}]$ does not contain any edges in $E_i$, so $G'_i[V_{i-1}]$ is the same as $G_i[V_{i-1}]$, the difference is that $G'_i[V_{i-1}]$ has more adjacent edges in $G'_i$ than $G_i[V_{i-1}]$. 
    
    Now we can apply \Cref{lem:trimmingcorrectness}. 
    We can call \textsc{Trimming} on $G'_{i},V_{i-1},\phi$, which satisfies all input requirements: the graph $G'_i$ is given as an adjacency list as we know $E_i$ explicitly; $A$ is accessed by identity query; and $E_{G'_i}(A,V\setminus A)$ is explicitly given as we know explicitly $E_i$ (more specifically, $E_{G'_i}(A,V\setminus A)$ contains all two-edges split by inserting a node in the middle of each edge in $E_i\cap E(G_{i-1}[V_{i-1}])$).
    %: notice that before the first update, $V_{i-1}$ can be accessed by identity query since $A=V$, and $E(A,V\setminus A)=\emptyset$ is explicitly given, we will also guarantee this in the next interaction after we have computed $P_i$. 
    The reader can now see more intuition on why we define $G'_i$ in the above way: the additional adjacent edges of $G'_i[V-P_{i-1}]$ will contribute to the source capacity of $\ff$. %Notably, this is easy to maintain once we know explicitly $P_i$, as then we can update $A=V-P_i$ and $E(A,V\setminus A)$ accordingly in work proportional to $\deg_G(P_i)$, which we will show is bounded by $\tO{|E_i|/\phi}$. 

    From \Cref{lem:trimmingcorrectness}, we get the set $V_{i-1}-A'$ and the flow $\ff$. We let $P_i=P_{i-1}\cup (V_{i-1}-A')$. %Since we have $\vol_G(V_{i-1}-A')=\tO{\frac{1}{\phi}}|E_{G'_{i-1}}(V_{i-1}, V(G'_{i-1}) \setminus V_{i-1})|$, which is at most $\tO{|E_i|/\phi}$, we can update $E(A,V\setminus A)$ easily in $\tO{|E_i|/\phi}$ work and $\tO{1}$ depth. 
    The first point of \Cref{thm:pruning} follows by the definition. 
    The second point of \Cref{thm:pruning} follows by the following argument: $\deg_G(P_i)$ contains two parts. One part is from the deleted edges, which is bounded by $\sum_{j<i}|E_j|$. The other part is from increasing by $|\deg_{G'_i}(V_{i-1}-A')|$ for each run of \textsc{Trimming}, which according to \Cref{lem:trimmingcorrectness} is bounded by $\tO{|E_{G_i}(V_{i-1},V(G_i)\setminus V_{i-1})|/\phi}\le \tO{|E_i|/\phi}$ according to the definition of $G_i$. 
    
    The complexity follows as well since 
    $$\tO{|E_{G_i}(V_{i-1},V(G_i)\setminus V_{i-1})|}=\tO{|E_i|}.$$
    It remains to prove that $G_i[V-P_i]$ is a $\phi/6\log n$-expander. We will first prove the following lemma.

    \begin{lemma}\label{lem:flowcertificate}
        For every $i<(\log n)/2$, there exists a flow $\ff_i$ supported by the graph $G_i[V_i]$ that routes the source $\frac{2}{\phi}(\deg_G(v)-\deg_{G_i[V_i]}(v))$ to sinks $\nnabla(v)\le 2i\deg_G(v)/\log n$, with the capacity of each edge bounded by $2i/\phi$.
    \end{lemma}
    \begin{proof}
        We prove it by induction on $i$. Initially, when $i=0$, the source vector is zero because $V_0=V$ and $G_i=G$, so we are done. Now suppose we have a flow $\ff_{i-1}$, we will construct the flow $\ff_i$ as follows. The flow $\ff_{i-1}$ is supported by $G_{i-1}[V_{i-1}]$, we first restrict it to $G'_i[V_i]$, which is the same as $G_i[V_i]$. After this restriction, for every edge $(u,v)$ in $G_{i-1}[V_{i-1}]$ that goes from $V_i$ to $V_{i-1}-V_i$, $u$ will receive at most $2/\phi$ many redundant flows since flow of $\ff_{i-1}$ is cut out here. Moreover, for every edge $(u,v)$ in $G$ that $u\in V_i$, $u$ gets an extra source capacity from $\ff_i$ of $2/\phi$ compared to $\ff_{i-1}$ if either $(u,v)\in E_i$ (newly deleted edge) or $v\in V_{i-1}-V_i$ and $(u,v)\in E(G_{i-1})$. The idea is to use the flow $\ff$ returned by \textsc{Trimming} in the $i$-th round to route all the new demand compared to $\ff_{i-1}$: consider the flow $\ff+\ff$, each $\ff$ routes the demand $\frac{2}{\phi}\cdot(\deg_{G'_i}-\deg_{G'_i[V-P_i]})$ according to \Cref{lem:trimmingcorrectness}, so the first $\ff$ can route the cut-off demand of $\ff_{i-1}$ back into $V_i$ (remember that the boundary of $V_i$ in $G'_i$ contains the same number of edge as the boundary of $G_{i-1}$ although some edges are deleted by $E_i$, because we put a node in the middle of that edge to construct $G'_i$), the second $\ff$ can route the extra source capacity for edges $(u,v)$ with $u\in V_i$ and either $(u,v)\in E_i$, in which case it contributes to $\deg_{G'_i}-\deg_{G'_i[V-P_i]}$ according to the definition of $G'_i$, or $v\in V_{i-1}-V_i$ and $(u,v)\in E(G_{i-1})$. 

        Notice that the sink capacity of every node increases by $2\deg_G(v)\log n$ every round, and the edge capacity increases by $2/\phi$ every round, so the lemma follows.
    \end{proof}
    The following lemma shows that the flow in \Cref{lem:flowcertificate} certifies that $G_i[V_i]$ is an expander. It is similar to Proposition 3.2 of \cite{SaranurakW19}. For its formulation, we say that $G'\subseteq G$ is a \emph{$\phi$-nearly expander} if for all $A'\subseteq V(G')$ with $\deg_{G}(A')\le \deg_{G}(G)/2$ we have $|E_{G}(A',V(G)-A')| \ge \deg_G(A')$. 

    \begin{lemma}\label{lem:certificate}
        Given a graph $G = (V, E)$ and a subgraph $G'\subseteq G$ such that $G'$ is a $\phi$-nearly expander and there exists a flow $f$ supported on the graph $G'$ that routes source $\Delta(v) \defeq 2/\phi\cdot(\deg_G(v) - \deg_{G'}(v))$ to sinks $\nabla(v) \defeq \deg_G(v)$ for $v \in V(G')$ with uniform edge capacity $2\log n/\phi$. Then $G'$ is a $\frac{\phi}{6\log n}$-expander. 
    \end{lemma}
    \begin{proof}
        If $G'$ is not a $\frac{\phi}{6\log n}$-expander, then there exists $A'\subseteq V(G')$, such that more than $1-\frac{\phi}{3\log n}$ fraction of the edges in $E_G(A',V\setminus A')$ are not in $G'$, in which case those edges contribute $\frac{1.9}{\phi}\deg_G(A')\cdot \phi$ many flows to the source capacity of $f$ since $G'$ is a $\phi$-nearly expander. Moreover, the sink capacity of all nodes in $A'$ is at most $\deg_G(A')$, so at least $0.9 \deg_G(A')$ capacity of flow should go out of $A'$ to $V(G')-A'$. However, the capacity of all the edges in $E_{G'}(A',V(G')-A')$ is at most $\frac{\phi}{3\log n}\cdot \deg_G(A')\cdot 2\log n/\phi<0.9\deg_G(A')$, a contradiction.
    \end{proof}

    By combining \Cref{lem:flowcertificate} and \Cref{lem:certificate}, we get that $G_i[V_i]$ is always a $\phi/(6\log n)$-expander. 
\end{proof}

\Cref{thm:pruning} follows from combining \Cref{lem:fully_framework,thm:pruninglowbatch}: we let the data structure $\mathfrak{D}$ be the data structure of Lem.~\ref{thm:pruninglowbatch}, which has batch number $b=\log n/2$, then we get a data structure of batch size $m$ without losing the update work and depth by a polylog factor. Batch size $m$ suffices for \Cref{thm:pruning} since after $m$ updates we can simply set $P_i=V$. 

\subsection{Parallel Unit Flow}\label{subsec:parallelunitflow}

Before proving the trimming lemma, \Cref{lem:trimmingcorrectness}, we need an important subroutine, see the following lemma.

\begin{algorithm}[h]
\caption{\textsc{ParallelUnitFlow}$(G, \cc, \DDelta, \nnabla, h)$}
\label{alg:parallel-unit-flow}
\DontPrintSemicolon
$\ff_0 \gets \veczero$; $\nnabla_0 \gets \veczero$: $\forall v \in V: l(v) = 0$\\
\For{$i = 1, \ldots, 8 \cdot \log_2 n$}{\label{lna:main_for} 
$x_i \gets \sum_{v \in V: l(v) \neq h + 1} \ex^G_{\ff_{i - 1},\DDelta,\nnabla_{i -1}}(v)$ \tcc*{Non settled excess}  $\ff_i' \gets \veczero$;
$\nnabla_i \gets \frac{1}{8 \log_2 n} \nnabla$ \\
\While(\label{lna:at-least-half}){$\sum_{v \in V: l(v) \neq h + 1}  \ex^G_{\ff_{i - 1}+\ff_i',\DDelta,\nnabla_i}(v) \geq x_i/2$}{
    $(\ff_i', l) \gets \textsc{PushThenRelabel}(G, \cc_{\ff_{i - 1}}, \ff_{i}', \ex^G_{\ff_{i - 1},\DDelta,\nnabla_{i -    1}}, \nnabla_i, h, l)$ \\
}
$\ff_i \gets \ff_{i-1}+\ff_i'$\\
}
$\forall v \in V$ s.t.\ $l(v) = h + 1$: $l(v) \gets h$ \\
\Return $(\ff_{8 \cdot \log_2 n}, l)$ \label{lne:final_return}
\end{algorithm}

\begin{algorithm}[h]
\caption{$\textsc{PushThenRelabel}(G, \cc, \ff, \DDelta, \nnabla, h, l)$}
\label{alg:push_then_relabel}
\DontPrintSemicolon
\For{j = h \ldots, 1}{
    In parallel, push all flow from all vertices $v$ with $l(v) = j$ that have excess flow to vertices $u$ with level $l(u) = j - 1$ until there either is no flow left or all the edges to such vertices are saturated. Update $\ff$ accordingly.  
}
For all vertices $v$ that only have saturated edges going to level $l(v) - 1$ and have no remaining sink capacity, increase their level $l(v) \gets \min(l(v) + 1, h + 1)$. \\
\Return $(\ff, l)$
\end{algorithm}

\noindent
Since we do not change the algorithm compared to \cite{ChenMGS25}, the correctness directly follows, as shown in the following lemma.

\begin{lemma}\label{lem:unitflowcorrectness}

Given a height parameter $h$ and a residual flow instance $\Pi = (G, \cc, \DDelta, \nnabla)$ where $\nnabla(v) \geq \gamma \cdot \deg(v)$ for all vertices $v \in V$ for some $0 < \gamma \leq 1$, $\norm{\DDelta}_1 \leq 2m$ and $\DDelta(v) \leq \eta \cdot \deg_G(v)$ for all $v \in V$, and $\|\cc\|_{\infty} \leq \eta$, the algorithm $\textsc{ParallelUnitFlow}(G, 
\cc, \DDelta, \nnabla, h)$ produces a flow $\ff$ and labeling $l:V\longrightarrow\{0,...,h\}$ such that:
\begin{enumerate}[label=(\roman*)]
    \item If $l(u)>l(v)+1$ where $\{u,v\}$ is an edge, then $\{u,v\}$ is saturated in the direction from $u$ to $v$, i.e.\ $\ff(u,v)=\cc(u,v)$.\label{prop:gvalidProp1}
    \item If $l(u)\geq 1$, then $u$'s sink is nearly saturated, i.e.\ $\ff(u)\geq \nnabla(u)/(8 \cdot \log_2 n)$.\label{prop:gvalidProp2}
    \item[(iii)] \label{prop:noExcessAtLowLevels} If $l(u)<h$, then there is no excess mass at $u$, i.e.\ $\ex_{\DDelta,\nnabla,\ff}^G(u) = 0$.
\end{enumerate}

\end{lemma}

Now we analyze the complexity of the algorithm in a more fine-grained way.

\begin{lemma}\label{lem:unitflowcomplexity}
    Given a height parameter $h$ and a residual flow instance $\Pi = (G, \cc, \DDelta, \nnabla)$ where $\nnabla(v) \geq \gamma \cdot \deg(v)$ for all vertices $v \in V$ for some $0 < \gamma \leq 1$, $\norm{\DDelta}_1 \leq 2m$ and $\DDelta(v) \leq \eta \cdot \deg_G(v)$ for all $v \in V$, and $\|\cc\|_{\infty} \leq \eta$, $\textsc{ParallelUnitFlow}$($G$,$\cc$, $\DDelta$, $\nnabla$, $h$) can be implemented to output the flow $\ff$ and labeling $l:V\longrightarrow\{0,...,h\}$ implicitly, i.e., only non-zero entries of $\ff$ and $l$ are stored in the output, in $\|\DDelta\|_0\cdot \tO{\eta h^2/\gamma^2}$ work and $\tO{\eta h^2/\gamma}$ depth. %that requires work $\tilde{O}(mh\eta/\gamma)$ and span $\tilde{O}(h^2\eta/\gamma)$, 
\end{lemma}

    \begin{proof}[Proof of \Cref{lem:unitflowcomplexity}]
        Since we only give implicit output, the initialization of $\ff$, $\DDelta$, $l$ can be skipped. Now we look into the for-loop. Throughout the algorithm, we will have the following invariance.
        \begin{claim}\label{cla:smallnonzerol}
            The total number of edges adjacent to a vertex without remaining sink capacity is at most $\tO{\|\DDelta\|_0/\gamma}$. This implies that the total number of vertices $u$ with excess and the total number of vertices with $l(u)>0$ are both bounded by $\tO{\|\DDelta\|_0/\gamma}$.
        \end{claim}
        \begin{proof}
            A vertex can have excess only if it has no remaining sink capacity, which means it absorbs at least $\tM{\gamma\cdot \deg(v)}$ flows. All vertices can absorb at most $\|\DDelta\|_0$ flows in total. Thus, if we sum up $\gamma\cdot\deg(v)$ for all $v$ with excess, it is at most $\|\DDelta\|_0$, this implies that the total number of edges adjacent to a vertex with excess is at most $\tO{\|\DDelta\|_0/\gamma}$.
            
            According to \textsc{PushThenRelabel}, a vertex has a non-zero level only if it has no sink capacity. So the claim follows.
        \end{proof}
        According to \Cref{cla:smallnonzerol}, we can always compute $x_i$ in $\tO{1}$ depth and $\tO{\|\DDelta\|_0/\gamma}$ work. 

        We can skip initializing $\ff'_i,\nnabla_i$ since they can be accessed entry-wise when in query.

        Now we look at the while-loop: the condition is checked in $\tO{\|\DDelta\|_0/\gamma}$ work and $\tO{1}$ depth according to \Cref{cla:smallnonzerol}. Inside \textsc{PushThenRelabel}, for each $j$, we only look at $v$ with $l(v)>0$, and push flows accordingly. According to \Cref{cla:smallnonzerol}, the push operations can be performed in $\tO{\|\DDelta\|_0/\gamma}$ work and $\tO{1}$ depth for each $j$. The relabel operations are only on vertices that have no remaining sink capacity, so it can be done in $\tO{\|\DDelta\|_0/\gamma}$ work and $\tO{1}$ depth according to \Cref{cla:smallnonzerol}. 
        
        Now we argue the number of times that the while-loop can be ran as the following claim. It is implicitly implied by the proof of Claim 4.1 in \cite{ChenMGS25}.

        \begin{claim}\label{cla:numberofwhile}
            In \textsc{ParallelUnitFlow}, for each $i$, the while-loop terminates in $\tO{\eta h/\gamma}$ loops.
        \end{claim}
        \begin{proof}
            Loot at \textsc{PushThenRelabel}, after the push operation, we claim that all the nodes with excess must have their level increase unless it is at level $h+1$. This is because such a node cannot have remaining sink capacity as otherwise it cannot have excess; also it cannot have non-saturated edges going to level $l(v)-1$, as otherwise the excess would be pushed through that edge in the push operations (during the push operations, if a node has all edges going to level $l(v)-1$ saturated, those edges cannot be un-saturated in the following push operations since push only happens from an upper level to a lower level). Thus, for each run of \textsc{PushThenRelabel}, we should think of all the remaining excess on a vertex less than $h+1$ level being ``raised'' by one level. 

            Now we count the total number of flows that can be raised. A vertex $v$ can raise at most $h\eta\cdot\deg_G(v)$ flows, since the edge capacities are bounded by $\eta$ and $\deg_G(v)$ edges going into this node, and it can increase its level at most $h$ times. Thus, the total number of flows that can be raised should be $h\eta\cdot\sum_{v\text{ has excess}}\deg_G(v)$. Notice that $\sum_{v\text{ has excess}}\deg_G(v)$ can be bounded by $\tO{x_i/\gamma}$ according to the same reason as in \Cref{cla:smallnonzerol}: if $v$ has excess then $v$ has absorbed $\tM{\gamma\cdot\deg_G(v)}$ flows, but the total flow that can be absorbed is $x_i$. Thus, the total number of flows that can be raised is bounded by $\tO{h\eta\cdot x_i/\gamma}$. 

            Now if each while-loop raises $x_i/2$ flows, the number of while-loops is bounded by 
            \[\tO{h\eta\cdot x_i/\gamma}/(x_i/2)=\tO{\eta h/\gamma}.\qedhere\]
        \end{proof}

        The above claim also implies that the support of $\ff'_i$ is bounded by $\tO{\eta h/\gamma}\cdot h\cdot \tO{\|\DDelta\|_0/\gamma}$, so $\ff_i$ can be updated in such work for each iteration. In the end we update $l$ with non-zero entries.

        To summarize, the dominating term of complexity is the push operations, which costs in total $\tO{\eta h/\gamma}\cdot h\cdot \tO{\|\DDelta\|_0/\gamma}=\|\DDelta\|_0\cdot \tO{\eta h^2/\gamma^2}$ work and $\tO{\eta h^2/\gamma}$ depth. 
    \end{proof}

\subsection{Trimming}\label{subsec:trimming}

Our goal in this section is to prove \Cref{lem:trimmingcorrectness}. Just for readability, we restate the lemma here.

\trimmingcorrectness*
\begin{proof}
    We first need two lemmas almost identical to Claim 3.7 and Claim 3.8 in \cite{ChenMGS25}.
    \begin{lemma}\label{lem:lessthanh}
        The while-loop at \Cref{lne:bat_pruning:while} terminates in less than $h$ steps.     
    \end{lemma}
    \begin{proof}
        The proof is identical to the proof for Claim 3.7 in \cite{ChenMGS25} since we do not change \Cref{lne:bat_pruning:while}.
    \end{proof}
    \begin{lemma}\label{lem:whilelesslogn}
        The main loop at \Cref{lne:whileOuter} of \Cref{alg:trimming} terminates after at most $ \log n$ steps.
    \end{lemma}
    \begin{proof}
        The proof is almost identical to the proof of Claim 3.8 in \cite{ChenMGS25} except that we have an extra $\log n$ factor on both $h$ and $\nnabla_i$ which cancel out. To be precise, by following the same argument, let the total excess at the end of the iteration $k$ by $X^k$, we have that
        \[X^{i-1}\ge\deg_G(S_j)/8\log^3n\]
        because every node in $S_j$ absorbs at least a $1/8\log^2n$ fraction of its degree many flows (according to \Cref{lem:unitflowcorrectness} point (ii)). Moreover, similarly, we have that $X^i\le \frac{4}{\phi}\cdot\frac{5\ln m}{h}\cdot\deg_G(S_j)\le 1/32 X^{i-1}$, and the algorithm terminates after $\log n$ iterations.
    \end{proof}

    Now we can prove the first point of our lemma. Almost identical to the proof of Claim 3.5 in \cite{ChenMGS25}, we have that $\ff$ routes the source $\frac{2}{\phi}(\deg_G(v)-\deg_{G[A']}(v))$ to sinks $\nnabla(v)=i\cdot\frac{\deg_G(v)}{\log^2 n}$ according to the definition of the algorithm (it terminates when the excess is $0$ according to \Cref{lne:ifReturn}), since we have \Cref{lem:whilelesslogn}, which means $i\le \log n$, we are done.

    Now we prove the second point, which is almost identical to the proof of Claim 3.9 in \cite{ChenMGS25}. Notice that every node $v$ in $A-A'$ absorbs at least $(1/8\log^3n)\deg_G(v)$ many flows according to point (ii) of \Cref{lem:unitflowcorrectness}, and according to the proof of \Cref{lem:whilelesslogn}, the total amount of flow that ever exists in the algorithm can be bounded by $X^0+X^1+....$ which is bounded by $2X^0=\frac{4}{\phi}\cdot|E(A,V\backslash A)|$ since $X_i\le X_{i-1}/32$. Thus, we have that $2X^0\ge (1/8\log^3n)\deg_G(A-A')$, and we are done. 
\end{proof}
\begin{lemma}\label{lem:trimmingtime}
    Given a graph $G = (V, E)$ accessed by adjacency list, a parameter $\phi \in \R_{\geq 0}$ and a set $A \subseteq V$ along with $E(A,V\setminus A)$ such that $G[A]$ is a $\phi$-nearly expander and $|E(A, V \setminus A)| \leq \phi \cdot m$, the algorithm $\textsc{Trimming}(G = (V, E), A, \phi)$ can be implemented to output $A-A'$ (which implicitly tells what is $A'$) in $\tO{|E(A,V\setminus A)|/\phi^4}$ work and $\tO{1/\phi^3}$ depth. 
\end{lemma}

\begin{algorithm}[h]
\caption{\textsc{Trimming}$(G = (V,E), A, \phi)$} \label{alg:trimming}
\DontPrintSemicolon
 $h \defeq \frac{5120}{\phi} \ \cdot \log_2^3 n\cdot \ln m$\\
$\cc \gets \frac{1}{\phi}\cdot\vecone$\\
$A_0 \gets A$; $\ff_0 \gets \veczero$, $\DDelta_0 \gets \frac{2}{\phi}(\deg_G - \deg_{G[A]})$; $\nnabla_0 \gets \veczero$; $i \gets 0$\\
\While(\label{lne:whileOuter}\tcc*[h]{While we do not have a feasible flow.}){\textbf{true}}{
    $i \gets i +1$\\
    $\nnabla_i \gets \nnabla_{i - 1} + \frac{1}{\log^2 n} \deg_G(v)[A_{i -1}]$ \\
    $({\ff'}_i, {l}_i) \gets \textsc{ParallelUnitFlow}\left(G[A_{i-1}], \cc_{
    \ff_{i - 1}}[A_{i-1}],
    \ex^{G[A_{i - i}]}_{\ff_{i - 1}, \DDelta_{i - 1}, \nnabla_{i - 1}}, \frac{\deg_G(v)[A_{i -1}]}{\log_2 n} , h\right)$ \label{lne:parallel_unit}\\
    $\ff_i \gets \ff_{i - 1} + \ff'_i$ \\
    \lIf(\label{lne:ifReturn}){$\ex^{G[A_i]}_{\ff_i, \DDelta_{i}, \nnabla_i} = \veczero$}{
        \KwBreak
    }

    $j \gets 0$; $S_0 \gets \{v\in A_{i-1}: {l}_i(v)=h\}$

   \While{$|E_{\ff_i}(S_j, A_{i-1} \setminus S_j)| \geq \frac{5\ln m}{h}\cdot \deg_{G}(S_j)$}{ \label{lne:bat_pruning:while}
        $j \gets j + 1$; $S_j \gets \{v \in A_{i-1}: {l}_i(v) \geq h-j\}$ 
    }
    
    $A_i \gets A_{i-1} \setminus S_j$.\label{lne:induceByCut}  \\
    $\DDelta_i \gets \frac{2}{\phi}(\deg_G[A_i] - \deg_{G[A_i]})$; $\nnabla_i \gets \nnabla_i[A_i]$
}

\Return $A' = A_{i-1}$, $\ff=\ff_i$
\end{algorithm}
\begin{proof}[Proof of \Cref{lem:trimmingtime}]
    %We first show how to implement the \textsc{ParallelUnitFlow} subroutine. The following lemma is similar to Lemma 3.4 in \cite{ChenMGS25}, but with a more careful complexity analysis.

    %We will prove \Cref{lem:unitflowcomplexity} later. Now let us assume \Cref{lem:unitflowcomplexity}. We will show how to implement each line of \textsc{Trimming} as follows. 
    
    % We first prove the following invariances throughout the computation.

    % \begin{claim}\label{cla:smallV\setminus Ai}
    %     Throughout the algorithm, we have $\vol_G(V\setminus A_i)=O(\vol_G(V\setminus A)/\phi)$.
    % \end{claim}
    % \begin{proof}
    %     According to the second point of \Cref{lem:trimmingcorrectness}, we have that $\deg_G(V\setminus A_i)-|E(A_i,V\setminus A_i)|-\deg_G(V\setminus A)\le \deg_G(A)-\deg_G(A')\le \tO{|E(A,V\backslash A)|/\phi}\le \tO{\deg_G(V\setminus A)/\phi}$. Thus, we only need to prove that $|E(A_i,V\setminus A_i)|=\tO{\deg_G(V\setminus A)}$. This is from the proof of Claim 3.10 in \cite{ChenMGS25}, in which they proved $|E(A',V\backslash A')|\le 2|E(A,V\backslash A)|$, but the proof actually shows $|E(A_i,V\backslash A_i)|\le 2|E(A,V\backslash A)|$ for every $i$. Therefore, we are done. 
    % \end{proof}

    We skip the explicit initialization of $h,\cc,A_0,\ff_0,\nnabla_0$ as they are all zero and can be updated entry-wise when needed. We initialize $\DDelta_0$ by memorizing all its non-zero entries in $O(|E(A,V\setminus A)|)$ work and $O(1)$ depth.    

    According to \Cref{lem:whilelesslogn}, the main loop on \Cref{lne:whileOuter} terminates after $\tO{1}$ steps, so let us focus on each step of the loop.

    For the updates on $\DDelta_i$, we will only update the entries which are non-trivial, i.e., which have absorbed some flows so that the entry is less than $(i/\log^2 n)\deg_G(v)[A_{i-1}]$. The work for updating $\DDelta_i$ thus can be charged to computing the flow $\ff_i$.

    Now we compute the work and depth for the call to \textsc{ParallelUnitFlow}. According to \Cref{lem:unitflowcomplexity}, initially we have that $\gamma=\Omega(1/\log^2n)$, $\eta=\tO{1/\phi}$, and $\|\DDelta\|_0\le \frac{4}{\phi}|E(A,V\setminus A)|$, so that the work is $\tO{|E(A,V\setminus A)|/\phi^4}$ and depth is $\tO{1/\phi^3}$. Notice that according to the proof of \Cref{lem:lessthanh}, the total source flow at each iteration $X^i$ is only decreasing, and $\gamma$ and $\eta$ remain unchanged, so we have that the work over all iterations is bounded by $\tO{|E(A,V\setminus A)|/\phi^4}$. 

    Both the update of $\ff_i$ and the check of excess can be charged to the work of \textsc{ParallelUnitFlow} as just performed.

    The while-loop in \Cref{lne:bat_pruning:while} contains less than $h$ iterations according to \Cref{lem:lessthanh}. The work for each iteration depends on $\sum_{l_i(v)>0}\deg_G(v)$, which is bounded by $\tO{|E(A,V\setminus A)|/\phi}$ according to \Cref{cla:smallnonzerol}.

    The update on $A_i$ should be implemented by adding the set $S_j$ to $V\setminus A_{i-1}$ since we are maintaining $V\setminus A_i$. The update of $\DDelta_i$ can be done easily since it only contains edges adjacent to $V\setminus A_i$; the update on $\nnabla$ should not be explicitly maintained as we only need to store the non-trivial sink capacities.

    To summarize, the work and depth are dominated by the flow computation, which takes $\tO{|E(A,V\setminus A)|/\phi^4}$ work and $\tO{1/\phi^3}$ depth. 
\end{proof}

%% file: algorithm.tex
\section{Full Algorithm}\label{sec:alg}

\newcommand{\ShortStep}{\textsc{ShortStep}\xspace}
\newcommand{\PathFollowing}{\textsc{PathFollowing}\xspace}

%\jan{todo: inverse maintenance corollary? do we have that? did we remove it?} \hossein{we decided to remove it and simply use laplacian/sdd solver for it (\Cref{lem:parallelsolver}). If you want, I can add it}
%\jan{I need an initial point lemma}
%\jan{I need an isolation lemma}

In this section we combine the data structures of \Cref{sec:HeavyHitter,sec:MaintainingRegularizedLewis-Weights,sec:primal,sec:graphDS} to obtain our main result \Cref{thm:main}.

\MainTheorem*

The overall algorithm for \Cref{thm:main} is identical to the minimum cost flow algorithm by \cite{BrandLLSS0W21}. The main difference is that we replace their data structures with variants that are more efficient in the parallel setting, i.e., our data structures have smaller depth.

Since the data structures of \cite{BrandLLSS0W21} provide the same input/output guarantees as the data structures we constructed in \Cref{thm:dual_maintenance,thm:gradient_maintenance,thm:lewis_weight_maintenance,thm:heavysampler},
%\jan{ref sections}, 
the correctness of our algorithm is directly implied by the correctness proof of \cite{BrandLLSS0W21}.
We are left with analyzing the parallel work and depth complexity of the algorithm when replacing their data structures with ours.
For this, we first state a few definitions and lemmas from \cite{BrandLLSS0W21}.

\paragraph{Preliminaries for full algorithm}

Throughout, we consider a linear program of the following form:
\begin{align*}
    \min c^\top x \\
    \mA^\top x = b \\
    0 \le x \le u
\end{align*}
where $\mA$ is obtained by removing one column (corresponding to one vertex) from an incidence matrix of a directed graph with $n$ vertices and $m$ edges.
We will use the following barrier functions and their derivatives (see \eqref{eq:lewisweight} in \Cref{sec:overview_IPM}):
\begin{align*}
    \phi(x)_i = - \log(x_i) - \log(u_i - x_i),~~\phi'(x)_i = - \frac{1}{x_i} + \frac{1}{u_i - x_i},~~\phi''(x)_i = \frac{1}{x^2_i} + \frac{1}{(u_i-x_i)^2}.
\end{align*}
The closeness to the central path is defined as follows, i.e., we require that the initial point of the algorithm satisfies this condition.
\begin{definition}[{$\eps$-centered point, \cite[Definition 4.7]{BrandLLSS0W21}}]
\label{def:centered}
We say that $(x,s,\mu) \in \R^m \times \R^m \times \R^m_{>0}$ is $\epsilon$-centered for $\epsilon \in (0,1/80]$ if the following properties hold, where %$\cnorm = C / (1-p)$ 
$\cnorm = C \log(4m/n)$ 
for a constant $C \ge 400$.
\begin{enumerate}
\item (Approximate centrality) $\left\| \frac{s+\mu \tau(x)\phi'(x)}{\mu\tau(x)\sqrt{\phi''(x)}}\right\|_\infty \le \epsilon.$
\item (Dual Feasibility) There exists a vector $z \in \R^n$ with $\ma z+s=c$.
\item (Approximate Feasibility) $\| \ma^\top x - b \|_{(\mA^\top(\mT(x)\Phi''(x))^{-1}\mA)^{-1}} \le \eps\gamma/\cnorm$.
\end{enumerate}
\end{definition}

\iffalse
\begin{lemma}[{\cite[Lemma 6.4 \& 4.11 \& 4.12]{BrandLLSS0W21}}]
    \label{lem:correctness}
    There is constant $C_{start}$ such that, given an initial $\epsilon/C_{start}$-centered point $(x^\init, s^\init, \mu)$,
    \PathFollowing (\Cref{alg:implement:pathfollowing}) makes $\tilde O(\sqrt{n} \log (\mu/\mu^\final))$ calls to \ShortStep (\Cref{alg:implement:short_step}). 
    %
    The resulting vector $x^\final$ satisfies
    \begin{itemize}
        \item $\mA^\top x^\final = b$, $x^\final \ge 0$
        \item $c^\top x^\final - \min_{\mA^\top x = b, \ell \le x \le u} c^\top x \le O(n \mu)$
    \end{itemize}
\end{lemma}
\fi
\iffalse
\begin{lemma}[{\cite[Lemma 7.2 \& 4.12]{BrandLLSS0W21}}]
    \label{lem:correctness}
    Consider a linear program
    \begin{align*}
        \min c^\top x \\
        \mA^\top x = b \\
        \ell \le x \le u
    \end{align*}
    where $\mA$ is obtained by removing one column (corresponding to one vertex) from an incidence matrix of a directed graph with $n$ verices and $m$ edges.
    Let $\epsilon = 1/(4C \log(m/n))$ for a large enough constant $C$.
    Given an initial $\epsilon$-centered point $(x^\init, s^\init, \mu)$ and a target $\mu^\target$,
    \PathFollowing (\Cref{alg:implement:pathfollowing}) makes $\tilde O(\sqrt{n} \log (\mu/\mu^\final))$ calls to \ShortStep (\Cref{alg:implement:short_step}) and returns
    an $\epsilon$-centered $(x^\final, s^\final, \mu^\target)$.\jan{define epsilon centered}
\end{lemma}
\fi

\begin{algorithm}[ht]
\caption{Data structure initialization, outer loop, final solution. \cite[Algorithm 4]{BrandLLSS0W21}\label{alg:implement:pathfollowing}}
\SetKwProg{Globals}{global variables}{}{}
\SetKwProg{Proc}{procedure}{}{}
\Globals{}{
	$D^{(x,\nabla)}$ instance of primal/gradient maintenance (\Cref{thm:gradient_maintenance}) using $\gamma/2^{12}$ accuracy \\
	$D^{(s)}$ instance of dual maintenance (\Cref{thm:dual_maintenance}) using $\gamma/2^{12}$ accuracy\\
	$D^{(\tau)}$ instance of Lewis weight data structure (\Cref{thm:lewis_weight_maintenance}) with accuracy $\gamma/2^{12}$ \\
	$D^{(\sample)}$ instance of \textsc{HeavySampler} (\Cref{thm:heavysampler}) \\
	%$D^{(-1)}$ instance of \textsc{InverseMaintenance} (\Cref{def:inverse_maintenance}) \\
	$\otau \in \R^m$ element-wise approximation of $\tau(\ox)$ (multiplicative error)\\
	$\ox \in \R^m$ element-wise approximation of $x$ (error relative to $\Phi''(\ox)$) \\
	$\os \in \R^m$ element-wise approximation of $s$ (multiplicative error) \\
	$\Delta \in \R^n$ (Infeasibility $\Delta = \mA^\top x - b$) \\
	$\omu \in \R$ approximation of $\mu$ \\
	\tcc{Parameters where $C$ is a sufficiently large constant}
	$\alpha \leftarrow \frac{1}{4\log(4m/n)}, \eps \leftarrow \frac{\alpha}{C}, \lambda \leftarrow \frac{C\log(Cm/\eps^2)}{\eps}, \gamma \leftarrow \frac{\eps}{C\lambda}, r \leftarrow \frac{\eps\gamma}{\cnorm\sqrt{n}}$
}
\Proc{$\PathFollowing(\mA, x^\init, s^\init, \mu^\init, \mu^\target)$}{
	$\ox \leftarrow x^\init$, $\os \leftarrow s^\init$, $\omu \leftarrow \mu^\init$, $\mu \leftarrow \mu^\init$, $\Delta \leftarrow 0$ \\
	%Let $c$ be the parameter assumed in \Cref{def:inverse_maintenance}, \Cref{def:heavyhitter}, and \Cref{def:heavysampler}.
    %, then define $z \leftarrow n/m$. \\
	$\otau \leftarrow D^{(\tau)}.\textsc{Initialize}(\mA, \phi''(\ox)^{-1/2}, n/m, 1- 1/(4\log(4m/n)), 2^{32}, \gamma / 2^{16})$ \\%\Comment{$C_{\ref*{lemma:pchange}}$ is the constant suppressed by the first item of Lemma \ref{lemma:pchange}} \\
	$D^{(x, \nabla)}.\textsc{Initialize}(\mA, x^\init, -\gamma \phi(\ox)''^{-1/2}, \otau, \frac{\os+\omu\otau\phi'(\ox)}{\omu\otau\sqrt{\phi''(\ox)}},\phi''(\ox)^{-1/2}, \gamma/2^{16})$\\
	$D^{(s)}.\textsc{Initialize}(\mA, s^\init, \omu\otau\phi''(\ox)^{1/2}, \gamma/2^{16})$ \\
	
	$D^{(\sample)}.\textsc{Initialize}(\mA, \otau^{-1}\phi''(\ox)^{-1/2}, \otau)$\\
	%$D^{(-1)}.\textsc{Initialize}(\mA, \otau^{-1} \phi''(\ox)^{-1}, \otau)$\\
	\While{$\mu > \mu^\target$}{
		$\ShortStep(\mu)$ (\Cref{alg:implement:short_step}) \\
		$\mu \leftarrow (1-r)\mu$ \\
	}
	\Return $D^{(x,\nabla)}.\textsc{ComputeExactSum}()$, $D^{(s)}.\textsc{ComputeExact}()$
    %$x \gets  D^{(x,\nabla)}.\textsc{ComputeExactSum}()$ \label{line:getexact}\\
    %$x^\final \gets x - (\omT\Phi''(x))^{-1} \mA(\mA^\top (\omT\Phi''(x))^{-1}\mA)^{-1}(b- \mA^\top x)$ \label{line:makefeasible}\\
    \Return $x$ %$x^\final$
}
\end{algorithm}

\subsection{Complexity Analysis}

In this section we analyze the work and depth complexity of the algorithm by \cite{BrandLLSS0W21}, when replacing their data structures with ours.
First, we need to state the number of iterations performed by \cite{BrandLLSS0W21}. We then prove in \Cref{lem:pathfollowingcomplexity} that each such iteration can be implemented in $\tilde O(1)$ depth, resulting in a $\tilde O(\sqrt n)$ depth algorithm.

\begin{lemma}[{\cite[Lemma 4.12]{BrandLLSS0W21}}]
    \label{lem:iterationbound}
    Let $\epsilon = 1/(4C \log(m/n))$ for a large enough constant $C$.
    Given an initial $\epsilon$-centered point $(x^\init, s^\init, \mu^\init)$ and a target $\mu^\target$,
    \PathFollowing (\Cref{alg:implement:pathfollowing}) makes $\tilde O(\sqrt{n} \log (\mu^\init/\mu^\final))$ calls to \ShortStep (\Cref{alg:implement:short_step}).
    %and returns
    %an $\epsilon$-centered $(x^\final, s^\final, \mu^\target)$.
\end{lemma}

\begin{lemma}\label{lem:pathfollowingcomplexity}
    %In the setting of \Cref{lem:correctness},
    Consider a linear program
    \begin{align*}
        \min c^\top x \\
        \mA^\top x = b \\
        0 \le x \le u
    \end{align*}
    where $\mA$ is obtained by removing one column (corresponding to one vertex) from an incidence matrix of a directed graph with $n$ vertices and $m$ edges.
    Let $\epsilon = 1/(4C \log(m/n))$ for a large enough constant $C$.
    Given an initial $\epsilon$-centered point $(x^\init, s^\init, \mu^\init)$ and a target $\mu^\target$,
    \PathFollowing (\Cref{alg:implement:pathfollowing}) has 
    $$\tilde O \left(\left(m + n^{1.5}\cdot(\log W + \log |\mu^\init/\mu^\target|)\right)\cdot \log |\mu^\init/\mu^\target|\right)$$ work and 
    $$\tilde O (\sqrt{n} \log |\mu^\init/\mu^\target|)$$ depth.
    Here $W$ is the ratio of largest to smallest entry in the vector $\phi''(x^\init) = 1/(u-x^\init)^2 + 1/(x^\init - \ell)^2$.
\end{lemma}

\begin{proof}
    The work bound follows from previous work \cite[Lemma 7.2]{BrandLLSS0W21}, since the amount of work in our and their data structures is identical.
    We are left with verifying the depth bound, as that is where our data structures differ.

    \paragraph{\PathFollowing (\Cref{alg:implement:pathfollowing})}
    This algorithm first initializes all the data structures.
    Each data structure initializes within $\tilde O (1)$ depth by \Cref{thm:dual_maintenance,thm:gradient_maintenance,thm:lewis_weight_maintenance,thm:heavysampler}.

    The algorithm then proceeds by performing $\tilde O (\sqrt{n} \log |\mu^\init/\mu^\target|)$ iterations (i.e., calls to \ShortStep, \Cref{alg:implement:short_step}) by \Cref{lem:iterationbound}. We argue in a later paragraph that each call to \ShortStep has depth $\tilde O(1)$.
    Thus overall, the depth is bounded by $\tilde O (\sqrt{n} \log |\mu^\init/\mu^\target|)$.

    \paragraph{\ShortStep (\Cref{alg:implement:short_step})}
    First, consider all lines that do not involve any of our new data structures.
    \begin{itemize}
        \item Line~\ref{line:step:H} samples a $\ow_i$, which requires $\tilde O(1)$ depth.
        \item Line~\ref{line:step:delta} updates $\Delta$, which takes $\tilde O(1)$ depth as it's just a sequence of constant many matrix-vector products.
        \item Line~\ref{line:systemsolver} solves a Laplacian system (with one row and column of the Laplacian removed due to removing a column of $\mA$). By \Cref{lem:solver} this takes $\tilde O(1)$ depth and $\tilde O(n)$ work because $\mH$ is a sparsified Laplacian (with one row\&column deleted) with only $\tilde O(n)$ edges.
        \item Lines~\ref{line:step:removeIx},~\ref{line:step:removeItau}, and \ref{line:step:removeIs} update the sets $I_x,$ $I_s$, and $I_\tau$. This  (and updating the vectors $\ox,\os,\otau$ in the next lines) also requires $\tilde O(1)$ depth.
    \end{itemize}
    Now consider all lines that involve data structures, which also includes the calls to \textsc{UpdateForX}, \textsc{UpdateForTau} (\Cref{alg:implement:update}).
    Overall, in each call to \ShortStep (\Cref{alg:implement:short_step}), each data structure receives only a constant number of calls to their methods.
    Since the depth of each method is bounded by $\tilde O(1)$ by \Cref{thm:gradient_maintenance,thm:dual_maintenance,thm:heavysampler,thm:lewis_weight_maintenance}, we have an overall depth bound of $\tilde O(1)$ per call to \ShortStep.
\end{proof}

\begin{algorithm}[t!]
\caption{One step of interior point method via data structures. \cite[Algorithm 5]{BrandLLSS0W21}\label{alg:implement:short_step}}
\SetKwProg{Globals}{global variables}{}{}
\SetKwProg{Proc}{procedure}{}{}
\Globals{}{
	Same variables as in \Cref{alg:implement:pathfollowing}. 
}
\Proc{$\ShortStep(\mu^{\new}>0)$}{
	%\LineComment{Update $\omu$ and data structures that depend on it.}
	\If{$\omu \not\approx_{\gamma/2^{12}} \mu^\new$}{ \label{line:step:update_mu}
		$\omu \leftarrow \mu^\new$ \\
		$D^{(x,\nabla)}.\textsc{Update}(
				[m], 
				-\gamma \phi''(\ox)^{-1/2}, 
				\otau,
				(\os+\omu\otau\phi'(\ox))/(\omu\otau\sqrt{\phi''(\ox)}))$ \\
			$D^{(s)}.\textsc{SetAccuracy}([m], \omu\otau\phi''(\ox)^{1/2})$ \\
	}
	%\LineComment{Perform \textsc{ShortStep} (\Cref{alg:short_step_LS})}
	$h' \leftarrow D^{(x,\nabla)}.\textsc{QueryProduct()}$ \label{line:step:gradient}
			\Comment{$h' = -\gamma \mA^\top \Phi''(\ox)^{-1/2} \nabla\Psi(\ov)^{\flat(\otau)}$} \\
	%\Comment{Here $\oy = (\os+\omu\otau\Phi'(\ox))/(\omu\otau\sqrt{\Phi''(\ox)})$}
      %
       \tcc{Leverage score sampling gives $\mH \defeq \mA^\top \omW \mA \approx_{\gamma/2} \mA^\top \omT^{-1} \Phi''(\ox)^{-1} \mA$}
	$\ow_{i}\leftarrow\begin{cases}
\frac{1}{\min(1,100\otau\log(n)/\gamma^{2})} & \text{with probability}\min(1,100\otau\log(n)/\gamma^{2})\\
0 & \text{otherwise}
\end{cases}$.  \label{line:step:H}\\
      %\tcc{$h'' = \mH^{-1} (h' + (\mA^\top x - b))$, $\delta_r = \omT^{-1}\Phi''(\ox)^{-1/2} \mA h''$}
	$h'' \leftarrow \mH^{-1} (h' + \Delta)$ via Laplacian system solver \Cref{lem:solver} with accuracy $\gamma/2$\label{line:step:h''}\label{line:systemsolver}\\
	% YinTat: I am confused by this line...
	% $\mR \leftarrow \textsc{SamplePrimal}(1, D^{\sample}, \otau, h'')$\\
       % I guess you mean this:
	$\mR \leftarrow D^{(\sample)}.\textsc{Sample}(h'')$\\
	$\Delta \leftarrow \Delta + h' - \mA^\top \mR \omT^{-1}\Phi''(\ox)^{-1} \mA h''$ 
			\Comment{Maintain $\Delta = \mA^\top x - b$} \label{line:step:delta}\\
	
    \bigskip

    \hrule\hfill
    \tcc{\hfill***** Update $\ox$ based on step $x\gets x+\delta_x$ *****\hfill}
    $x^{\tmp}, I_x \leftarrow D^{(x,\nabla)}.\textsc{QuerySum}(-\mR \omT^{-1}\Phi''(\ox)^{-1} \mA h'')$ \label{line:step:xtmp}\\
    Remove indices $i$ from $I_x$ where $|\sqrt{\phi''_i(x^{\tmp}_i)} (x^{\tmp}_i - \ox_i)| \le \gamma / 2^{12} $ \label{line:step:removeIx}\\
	$\ox_{I_x} \leftarrow x^{\tmp}_{I_x}$ \\
	\textsc{UpdateForX}$(I_x)$ \tcp{Notify data structures about changes to $\ox$}
	%\LineComment{Update $\otau$ and data structures that depend on it.}
    
    \bigskip

    \hrule\hfill
    \tcc{\hfill***** Update $\otau$ based on step $x\gets x+\delta_x$ *****\hfill}
	$\tau^{\tmp}, I_\tau \leftarrow D^{(\tau)}.\textsc{Query}()$ \label{line:step:tautmp}\\
    Remove indices $i$ from $I_\tau$ where $\tau^{\tmp}_i \approx_{\gamma/2^{10}} \otau_i$ \label{line:step:removeItau}\\
    $\otau_{I_\tau} \leftarrow \tau^{\tmp}_{I_\tau}$ \\
	\textsc{UpdateForTau}$(I_\tau)$ \tcp{Notify data structures about changes to $\otau$}

    \bigskip

    \hrule\hfill
    \tcc{\hfill***** Update $\os$ based on step $s\gets s+\delta_s$ *****\hfill}
	$s^{\tmp}, I_s \leftarrow D^{(s)}.\textsc{Add}(\mu \mH^{-1}h')$ where $\mH^{-1}h'$ is computed via Laplacian system solver \Cref{lem:solver} using accuracy $\gamma/2$ \label{line:step:stmp}\\
    Remove indices $i$ from $I_s$ where $|\omu^{-1}\otau^{-1}\phi''(\ox)^{-1/2}(s^\tmp_i-\os_i)| \le \gamma/2^{10}$ \label{line:step:removeIs}\\
    $\os_{I_s} \gets s^{\tmp}_{I_s}$\\
    \tcp{Notify data structure about changes to $\os$}
	$D^{(x,\nabla)}.\textsc{Update}(
            {I_s},
            (-\gamma \phi''(\ox)^{-1/2})_{I_s}, 
            ((\os+\omu\otau\phi'(\ox))/(\omu\otau\sqrt{\phi''(\ox)}))_{I_s}
            )$ 
}
\end{algorithm}

\begin{algorithm}[h]
\caption{Notify data structures about changes to $\ox$ and $\otau$. \label{alg:implement:update}}
\SetKwProg{Globals}{global variables}{}{}
\SetKwProg{Proc}{procedure}{}{}
\Globals{}{
	Same variables as in \Cref{alg:implement:pathfollowing}.
}\medskip

\Proc{\textsc{UpdateForX}$(I\subset[m])$}{
	$D^{(x,\nabla)}.\textsc{Update}(
				I_x,
				(-\gamma \phi''(\ox)^{-1/2})_{I_x}, 
				\otau_{I_x},
				((\os+\omu\otau\phi'(\ox))/(\omu\otau\sqrt{\phi''(\ox)}))_{I_x}
				)$ \\
	$D^{(x,\nabla)}.\textsc{SetAccuracy}({I_x}, (\phi''(\ox)^{-1/2})_{I_x})$ \\
	$D^{(\tau)}.\textsc{Scale}({I_x}, (\phi''(\ox)^{-1/2})_{I_x})$ \label{line:step:update_tau_for_x}\\
	$D^{(\sample)}.\textsc{Scale}({I_x}, (\otau^{-1} \phi''(\ox)^{-1/2})_{I_x})$\\
	%$D^{(-1)}.\textsc{Update}({I_x}, (\otau^{-1} \phi''(\ox)^{-1})_{I_x}, \otau_{I_x})$\\
	$D^{(s)}.\textsc{SetAccuracy}({I_x}, (\omu\otau\phi''(\ox)^{1/2})_{I_x})$\\
}\medskip

\Proc{\textsc{UpdateForTau}$(I\subset[m])$}{
	$D^{(x,\nabla)}.\textsc{Update}(I,
            (-\gamma \phi''(\ox)^{-1/2})_I, 
            \otau_I,
            ((\os+\omu\otau\phi'(\ox))/(\omu\otau\sqrt{\phi''(\ox)})_I)
            )$ \\
    $D^{(\sample)}.\textsc{Scale}(I, (\otau^{-1} \phi''(\ox)^{-1/2})_I)$\\
    %$D^{(-1)}.\textsc{Update}(I, (\otau^{-1} \phi''(\ox)^{-1})_I, \otau_I)$\\
    $D^{(s)}.\textsc{SetAccuracy}(I, (\omu\otau\phi''(\ox)^{1/2})_I)$\\
}
\end{algorithm}

\begin{algorithm}[h]
\caption{Complete algorithm for min-cost flow \cite[Section 7]{BrandLN+20} \label{alg:implement:complete}}
\SetKwProg{Globals}{global variables}{}{}
\SetKwProg{Proc}{procedure}{}{}
\Proc{$\textsc{MinCostFlow}(G=V,E),c,u\in\Z^m,d\in\Z^n)$}{
    \tcp{Perturb costs $c$ a bit, so the optimal solution for $G$ becomes unique (Isolation Lemma 7.8 \cite{BrandLLSS0W21})}
    Let $W=\max(\|c\|_\infty, \|d\|_\infty, \|u\|_\infty)$ \label{line:isolationstart}\\
    For each edge $e$, set $c_e \gets c_e + r_e$ where each $r_e$ is an independent uniformly sampled number from set $\{\frac{1}{4m^2W^3},...,\frac{2mW}{4m^2W^3}\}$ \label{line:isolationend}\\
    \bigskip

    \tcp{Modify graph to allow construction of initial point. \cite[Lemma 7.5]{BrandLLSS0W21}}
	Construct a modified min-cost flow instance $G',c',u',b$ as follows: \label{line:initialstart}\\
    Let $G'=(V\cup\{z\}, E\cup E')$ by adding a vertex $z$, and connecting it with every other vertex in both directions.\\
    Let $b$ be the same demands as $d$, but $b_z = 0$. \\
    Let $x^\init$ be the flow with $x^\init_e = u_e/2$ for $e\in E$, and let $x^\init_e = 1$ for $e\in E'$. Then add additional flow to $e\in E'$ such that all demands are satisfied. \\
    Let $c',u'$ be the edge costs and capacities of the new graph. The capacity and cost of each edge $e\in E$ stays the same, but for $e\in E'$ we let $u'_e = 2x^\init_e$ and $c'_e = 50m\|u\|_\infty\|c\|_\infty$.\\
    Let $s^\init = c'$, and $\mu^\init = 100m^2\max(\|u\|_\infty,\|c\|_\infty)^3/\epsilon$ \label{step:initialpoint}\label{line:initialend}\\
    \bigskip
    
    \tcp{Run linear program solver}
    Let $\mA$ be the incidence matrix of $G'$\\
    $\mu^\target \gets 1/(12m^3W^4)$\\
    $x,s \gets \PathFollowing(\mA, x^\init, s^\init, \mu^\init, \mu^\target)$ (\Cref{alg:implement:pathfollowing})\\
    \bigskip
    
    \tcp{Make solution feasible, i.e., $\mA^\top x = b$. \cite[Lemma 7.7]{BrandLLSS0W21}}
    $x \gets \omT\Phi''(x)^{-1}\mA(\omT \Phi''(x))^{-1} \mA)^{-1}(b-\mA^\top x)$ where solving it $1/\poly(mW)$-approximately via a Laplacian solver suffices.\label{line:makefeasible}\tcp{$x$ is entry-wise close to opt}
    Round each $x_e$ to the nearest integer. \tcp{opt must be integral flow, and $x$ is entry-wise close to opt, so rounding yields opt}
    \Return $x$ %$x^\final$
}
\end{algorithm}

The complexity bound of \Cref{lem:pathfollowingcomplexity} requires that the input is $\epsilon$-centered. This is guaranteed by the initial point construction in \Cref{alg:implement:complete}, as stated in the following lemma:

\begin{lemma}[{\cite[Lemma 7.5]{BrandLLSS0W21}}]\label{lem:initialcorrectness}
    The initial point $(x^\init, s^\init,\mu^\init)$ constructed in Lines~\ref{line:initialstart} to \ref{line:initialend} of \Cref{alg:implement:complete} is $\epsilon$-centered for the linear program
    \begin{align}
        \min c'^\top x \label{eq:modifiedLP}\\
        \mA^\top x = b \notag\\
        0 \le x \le u'\notag
    \end{align}
    where $\mA$ is the edge-vertex-incidence matrix of $G'$ with any one removed column.
\end{lemma}

\begin{proof}[Proof of \Cref{thm:main}]
    \Cref{alg:implement:complete} is the same as in \cite{BrandLLSS0W21} (up to our new data structures), so correctness follows from their analysis. We are left with the complexity analysis.
    %We recap the algorithm here to analyze its parallel complexity, i.e., work and depth.
    
    Given the graph $G$, the random perturbations to the cost as described in Lines~\ref{line:isolationstart} and \ref{line:isolationend}, take $O(m)$ work and $\tilde O(1)$ depth.
    %. This is done in $O(m)$ works and $O(1)$ depth.

    Constructing the initial point via Lines~\ref{line:initialstart} to \ref{line:initialend} also takes $O(m)$ work and $\tilde O(1)$ depth. The depth bottleneck is the calculation of the flow on the newly added edges to fix the demand on each vertex.

    Then the interior point method \PathFollowing is run for $\mu^\init/\mu^\target = \poly(m\|u\|_\infty\|c\|_\infty)$ iterations,
    which takes $\tilde O(\sqrt{n} \log (m\|u\|_\infty\|c\|_\infty))$ depth and 
    $$\tilde O\left(\left(m+n^{1.5}\log(\|u\|_\infty\|c\|_\infty)\right)\log W')\log(\|u\|_\infty\|c\|_\infty)\right)$$ work by \Cref{lem:pathfollowingcomplexity}. 
    Here $W'$ is the ratio of largest to smallest entry in 
    $$\phi''(x^\init) = \frac{1}{(x^\init)^2} + \frac{1}{(u' - x^\init)^2}.$$
    By construction of $u' = 2x^\init$, and the initial flow is either $x^\init_e = u_e/2$ or routes flow to fix the demands (\Cref{alg:implement:complete}). Thus, the ratio is bounded by $O((\|d\|_\infty+\|u\|_\infty)^2)$.
    Here in the case of $st$-max flow, we have $\|d\|_\infty = 0$ since we look for a circulation where the back edge $(t,s)$ has negative $\poly(nC)$ cost.
    
    Solving the Laplacian system in Line~\ref{line:makefeasible} and rounding each entry to the nearest integer, takes $\tilde O (m)$ work and $\tilde O(1)$ depth by \Cref{lem:solver}.

    Overall, the work is bounded by 
    $$\tilde O\left(\left(m+n^{1.5}\log (\|u\|_\infty\|c\|_\infty)\right)\log (\|u\|_\infty\|c\|_\infty)\right)$$ and $\tilde O(\sqrt{n}\log(\|u\|_\infty\|c\|_\infty))$ depth.
\end{proof}

%% file: parallelization.tex
\section{Parallelization Primitives}\label{sec:parallelization}
In this section, we provide some basic parallel routines, that we will use repeatedly in our algorithms.

\paragraph{Parallel Tools.}
\cite{BrandLLSS0W21}'s IPM must repeatedly solve linear systems of the form $(\mA^\top \mD \mA) x = b$, where $\mD \in \R^{m \times m}$ is a diagonal matrix with positive entries. It is not hard to see that the matrix $\mA^\top \mD \mA$ is an SDD, i.e., $\mA^\top \mD \mA$ is symmetric and for all $i \in [n]$, it satisfies $(\mA^\top \mD \mA)_{ii} \geq \sum_{j \in [n]} |(\mA^\top \mD \mA)_{ij}|$.

For these systems, we use the following linear system solver:

\begin{lemma}[Parallel SDD Solver, see e.g. \cite{ParallelSDDSolver}]\label{lem:parallelsolver}
    Assume $\mA \in \R^{m \times n}$ is a matrix obtained by removing one column from an edge-vertex incidence matrix of a directed graph with $n + 1$ vertices and $m$ edges, and $\mD \in \R^{m \times n}$ a diagonal matrix with positive entries. Furthermore, assume $b \in \R^n$ and $0 < \epsilon < 1$ are given. Assuming there exists a vector $x \in \R^n$ with $(\mA^\top \mD \mA) x = b$, there is an algorithm, that, w.h.p., returns an $\epsilon$-approximate solution to $(\mA^\top \mD \mA) x = b$, i.e., a vector $\bar{x} \in \R^n$ such that $\|\ox - x\|_{\mA^\top \mD \mA} \leq \varepsilon \|x\|_{\mA^\top \mD \mA}$, in $\tO{\nnz(\mA) \log W \log \epsilon^{-1}}$ work and $\tO{1}$ depth, where $W$ is the ratio of the largest to smallest entry of $\mD$.
\end{lemma}

An important tool for implementing the data structures required for the IPM in the parallel setting is a sorted list data structure (in most cases, it does not necessarily have to be sorted), which can be initialized with efficient work and depth and allows efficient bulk search, insertion, and deletion, i.e., searching, inserting, or deleting a set of elements efficiently. For this, we use following data structure.

\begin{lemma}[Parallel Sorted List Maintenance] \label{lem:parallelSortedList}
    There exists a deterministic data structure $ T $ that stores a list of elements from a universe $\cA$, sorted according to a defined comparison criterion, and supports the following operations:
    \begin{itemize}
        \item \textsc{Initialize}$(L \subseteq \cA, R \subseteq \cA \times \cA)$:  
        The data structure is initialized with a given set $L$, storing its elements sorted based on the comparison criterion $R$. It assumes, $xRy$, can be computed in $O(1)$ work and depth. This operation requires $O(|L| \log |L|)$ work and $O(\log |L|)$ depth in the EREW model, or $O(1)$ work and depth if $L = \emptyset$.

        \item \textsc{Search}$(I \subseteq \cA)$:  
        Searches for the elements of $I$ in the data structure, and returns a list $L \subseteq [0,1]^I$, specifying whether each entry of $I$ is in the data structure or not. This operation requires $O(|I|)$ work and $O(\log |I| + \log |T|)$ depth in the EREW model, where $|T|$ is the current number of elements in the data structure. 

        \item \textsc{Insert}$(I \subseteq \cA)$:  
        Inserts the elements of $I$ into the list while maintaining the sorted order. This operation requires $O(|I|)$ work and $O(\log |I| + \log |T|)$ depth in the EREW model, where $|T|$ is the current number of elements in the data structure.

        \item \textsc{Delete}$(I \subseteq \cA)$:  
        Removes the elements of $I$ from the list if they exist in the data structure, while maintaining the sorted order. This operation requires $O(|I|)$ work and $O(\log |I| + \log |T|)$ depth in the EREW model, where $|T|$ is the current number of elements in the data structure.

        \item \textsc{RetrieveAll}():  
        Returns an array of the sorted elements in the data structure. This operation requires $O(|T|)$ work and $O(\log |T|)$ depth.
    \end{itemize}
\end{lemma}

\begin{proof}[Proof of \Cref{lem:parallelSortedList}]
    Such a data structure can be implemented using a binary search tree, which is well-studied in the parallel setting. Here, we use the parallel implementation of red-black trees by \cite{PARK2001415}. As they show, their implementation allows for constructing a red-black tree from a sorted list $L$ in $O(|L| / \log \log |L|)$ work and $O(\log \log |L|)$ depth in the EREW PRAM model. However, sorting $L$ can be done with $O(|L| \log |L|)$ work and $O(\log |L|)$ depth, confirming the complexity of \textsc{Initialize}. 

    Furthermore, they demonstrate that searching, inserting, and deleting $k$ unsorted items in a tree with $n$ nodes can be done in $O(k)$ work and $O(\log k + \log n)$ depth, confirming the complexities of \textsc{Search}, \textsc{Insert}, and \textsc{Delete}. 
    Finally, retrieving all elements of the tree can be achieved by traversing it, resulting in $O(n)$ work and $O(\log n)$ depth, as required by \textsc{RetrieveAll}.
\end{proof}

\paragraph{Sampling with work bounded by output size}

Throughout the algorithm, multiple sampling procedures are employed. In general, in the following sections, we assume that sampling can be implemented such that the expected work is bounded by the expected number of sampled entries (up to an additive $\log W$ factor, where $W$ denotes the ratio of the largest nonzero entry to the smallest). The depth of these procedures is assumed to be $\tO{1}$.
Here, we are given a sampling distribution that samples indices $i\in[m]$ proportional to a given $\tau\in\R^m$. This distribution changes slowly, i.e., individual entries of $\tau$ might change over time.

To implement these sampling procedures, we use the following data structure.
%, which answers the question: Given a graph where each edge $e$ has a (normalized) weight $\tau_i$, how can we efficiently sample edges with probabilities proportional to these weights?

\begin{theorem}[Parallel \textsc{$\tau$-Sampler}] \label{thm:tausampler}
    There exists a data structure that supports the following operations:
    \begin{itemize}
        \item \textsc{Initialize}$(\tau \in \R^m_{>0})$: 
        Initializes the data structure in $\tO{m}$ work and $\tO{1}$ depth.
        
        \item \textsc{Scale}$(I \subseteq [m], a \in \R^I_{>0})$:  
        For each $i \in I$, updates $\tau_i \leftarrow a_i$ in $\tO{|I|}$ work and $\tO{1}$ depth.
        
        \item \textsc{Sample}$(K \in \R_{>0})$:  
        Returns a set $M \subseteq [m]$, where each $i \in [m]$ is included in $M$ independently with probability
        \begin{align*}
            \P[i \in M] \geq K \cdot \frac{n \tau_i}{\|\tau\|_1}.
        \end{align*}
        With high probability, the output size $|M|$ is bounded by $\tO{Kn}$. This operation takes, with high probability, $\tO{Kn + \log W}$ work and $\tO{1}$ depth, where $W$ is an upper bound on the ratio of the largest to smallest nonzero entry in $\tau$.

        \item \textsc{Probability}$(I \subseteq [m], K \in \R_{>0})$:  
        For each $i \in I$, returns the probability $p_i$ with which $i$ is selected in \textsc{Sample}$(K)$. This operation takes $\tO{|I|}$ work and $\tO{1}$ depth.
    \end{itemize}
\end{theorem}

When using this data structure, we always have $\tau =\Omega( n/m)$ and $\|\tau\|_1 = \tO{1}$, so the term $\log W$ is bounded by $\tO{1}$.
We use this sampling routine to sample $\tO{n}$ edges of an $m$-edge graph. Here, sampling with higher probability than the edges' effective resistances yields a spectral sparsifier. So WLOG we can assume each edge is sampled with probability at least $n/m$ by adding $n/m$ to the sampling probabilities.

%We remark that in most cases, we want an output of size $\tO{n}$, and sampling elements with higher probability is fine as long the the output is bounded. So we can WLOG assume $\tau_i \ge n/\|\tau\|_1$ which

\begin{proof}[Proof of \Cref{thm:tausampler}]
    The main idea behind the implementation of this data structure is similar to the \textsc{Sample} procedure in \Cref{lem:large_entry_datastructure}. Specifically, we distribute the weights $\tau_i$ into $\log W$ buckets, where each bucket contains weights of similar value. Then, for each bucket, we sample the elements within it using the same probability. This approach ensures that the overall sampling process achieves the desired complexity.

    \paragraph{\textsc{Initialize}.}  
    Begin by storing $\tau$ internally within the data structure. Then, for each $i \in [1,m]$ in parallel, compute the index $j_i$ such that $\tau_i \in [2^{j_i}, 2^{j_i+1})$, and store these indices in an internal $m$-dimensional vector $k$ by setting $k_i \leftarrow j_i$. Next, initialize $\log W$ buckets $B_j$, where  
    \begin{align*}
        B_j = \{i \mid \tau_i \in [2^j, 2^{j+1})\}.
    \end{align*}  
    Additionally, compute and store $\|\tau\|_1 = \sum_{i \in [1,m]} \tau_i$ internally within the data structure.  

    For managing the buckets, we use the data structure from \Cref{lem:parallelSortedList}. Thus, the initialization procedure requires $\tO{m}$ work and $\tO{1}$ depth.
    
    \paragraph{\textsc{Scale}.}
    For each $i \in I$ in parallel, determine its new bucket by computing $j_i$ such that $a_i \in [2^{j_i}, 2^{j_i+1})$. Then, update the bucket assignments by removing all $i \in I$ from its previous bucket $B_{k_i}$ in parallel and adding it to its new bucket, also in parallel, while updating $k_i$ accordingly. 

    Since we use the list data structure from \Cref{lem:parallelSortedList}, it is not hard to see that, this can be done in $\tO{|I|}$ work and $\tO{1}$ depth.

    Additionally, we update the value of $\|\tau\|_1$ as follows:
    \begin{align*}
        \|\tau\|_1 \leftarrow \|\tau\|_1 + \sum_{i \in I} (a_i - \tau_i).
    \end{align*}
    This update requires $\tO{|I|}$ work and $\tO{1}$ depth.

    \paragraph{\textsc{Sample}.}
    For each $i \in [1,m]$, we sample $i$ with probability
    \begin{align*}
        p_i = K \cdot \frac{n 2^{j+1}}{\|\tau\|_1} > K \cdot \frac{n \tau_i}{\|\tau\|_1},
    \end{align*}
    assuming $i$ belongs to bucket $B_j$, i.e., $k_i = j$.
    
    To implement this (similar to \textsc{Sample} from \Cref{lem:large_entry_datastructure}), we proceed as follows: for each bucket $B_j$ in parallel, we first sample a binomial $x \in [0,|B_j|]$, then select $x$ indices uniformly at random from $[1, |B_j|]$, and finally retrieve the corresponding entries of $B_j$. This approach bounds the required work for sampling by the expected number of sampled $i$'s, which can be bounded by:
    \begin{align*}
        \E[X] &= \sum_{B_j} \sum_{i \in B_j} K \cdot \frac{n 2^{j+1}}{\|\tau\|_1} \\
        &\leq \sum_{B_j} \sum_{i \in B_j} K \cdot \frac{n 2 \tau_i}{\|\tau\|_1} \\
        &= \sum_{i \in [1,m]} 2 K \cdot \frac{n \tau_i}{\|\tau\|_1} \\
        &= 2 K n \in O(K n).
    \end{align*}
    
    Thus, the work complexity of \textsc{Sample} is bounded by $\tO{\log W + Kn}$, and its depth by $\tO{1}$.

    \paragraph{\textsc{Probability}.}
    For this operation, we essentially need to compute 
    \begin{align*}
        p_i = K \cdot \frac{n 2^{j+1}}{\|\tau\|_1},
    \end{align*}
    for each $i \in I$, assuming $i$ belongs to bucket $B_j$, i.e., $k_i = j$. 
    
    Since this computation is independent for each $i \in I$, it can be performed in parallel, resulting in $\tO{|I|}$ work and $\tO{1}$ depth.
\end{proof}

%% file: heavyhitters.tex
\section{Heavy Hitters}\label{sec:HeavyHitter}

In this section, we prove the following lemma, which is the parallel version of Lemma 5.1 in \cite{BrandLN+20}. We will use exactly the same algorithm so the correctness follows trivially, but we describe how we implement each step in parallel. The implementation requires \Cref{lem:dynamicExpanderDecomposition}, especially. 

\begin{lemma}\label{lem:large_entry_datastructure}
    There exists a data structure \textsc{HeavyHitter} %
	that supports the following operations.
	\begin{itemize}
		\item \textsc{Initialize($\mA \in \{-1,0,1\}^{m\times n}, g \in \R^m_{\geq 0}$)}: 
		The data structure is given the edge incidence matrix $\mA$ of a directed graph and a scaling vector $g$.
		It initializes in $\tO{m}$ work and $\tO{1}$ depth. 
		\item \textsc{Scale$(I \subseteq [m], s \subseteq \mathbb{R}^I_{\ge 0})$}: 
		Updates $g_i \leftarrow s_i$ for each $i\in I$ in $\tO{|I|}$ amortized work and $\tO{1}$ depth.
		\item \textsc{HeavyQuery($h\in \R^n$, $\epsilon \in \R_{>0}$)}:
		With high probability, the data structure returns all $e\in [m]$ such that %
		$\left|\left(\mdiag(g) \mA h\right)_e\right| \geq \epsilon$		in work
		\begin{align*}
		\tO{\|\mdiag(g) \mA h\|_2^2 \epsilon^{-2} + n\log W},
		\end{align*}
		where 
		$W$ is the ratio of the largest to the smallest non-zero entries in $g$, and depth $\tO{1}$.
		\item \textsc{Sample}$(h \in \R^n, K \in \R_{>0})$:
		With high probability, in $\tO{K + n\log W}$ work and $\tO{1}$ depth,
		the data-structure returns independently sampled indices of $\mdiag(g) \mA h\in \R^m$ (i.e., edges in graph $G$), where each edge $e=(u,v)$ is sampled with some probability $q_e$ which is at least
		\begin{align*}
		\min\left\{K \cdot \frac{(g_e (h_u - h_v))^2}{16 \|\mdiag(g) \mA h\|_2^2 \log^{8} n},1\right\},
		\end{align*}
		and with high probability there are at most $O(K \log n)$ entries returned.
		\item \textsc{Probability$(I\subset [m], h\in \R^n, K\in \R_{>0})$}: 
		Given a subset of edges $I$, this procedure returns for every $e\in I$ the probabilities $q_e$
		that $e$ would be sampled in the procedure \textsc{Sample$(h,K)$} . 
		The work is $\tO{|I|+n\log W}$ and depth is $\tO{1}$.
		\item \textsc{LeverageScoreSample}$(K' \in \R_{>0})$:
		With high probability, in $\tO{K' n \log W}$ work and $\tO{1}$ depth
		\textsc{LeverageScoreSample} returns a set of sampled edges 
		where every edge $e$ is included independently 
		with probability $p_e \ge K' \cdot \sigma(\mdiag(g)\mA)_e$, 
		and there are at most $O(K' n \log^{11} n \log W)$ entries returned.
		
            \item \textsc{LeverageScoreBound}$(K', I \subset [m])$:
		Given a subset of edges $I$, this procedure returns for every $e\in I$ the probabilities $p_e$ 
		that $e$ would be sampled in the procedure \textsc{LeverageScoreSample$(K')$}. 
		The work is $\tO{|I|}$ and depth is $\tO{1}$. %

	\end{itemize}
\end{lemma}

\begin{proof}[Proof of \Cref{lem:large_entry_datastructure}]
	We describe each of the operation independently.
	\paragraph{\textsc{Initialization}.}
	Let $G=(V, E)$ be the graph corresponding to the given incidence matrix $\mA$ with weight $g_e$ of each edge $e$.
	Partition the edges into subgraphs denoted as $G_i=(V, E_i)$, where 
	\begin{align}
	E_i=\{e \mid g_e\in [2^{i}, 2^{i+1})\} \label{eq:graph weight classified}
	\end{align}
	i.e. each $G_i$ is an unweighted subgraph of $G$ consisting of edges of roughly the same $g_e$ values. 
Let $G_{-\infty}$ be the subgraph induced by all zero-weight edges. 
This step can be done in parallel by folklore parallel sorting all the edges in $\tO{m}$ work and $\tO{1}$ depth, then partitioning the sorted list into different $E_i$'s by finding the correct transition entries in parallel.
	Next, we choose $\phi = 1/\log^4 n$ and initialize the expander decomposition algorithm in \Cref{lem:dynamicExpanderDecomposition}
	by treating all the edges in $G_i$ as \emph{undirected} edges and inserting all the edges in one batch, which results in $\phi$-expanders (when viewing each directed edge as undirected) $G_{i,1},...,G_{i,t_i}$ for each $i$. This step takes $\tO{m}$ work and $\tO{1}$ depth.

    %Note this can be implemented so the running time does not depend on the ratio $W$ of largest and smallest non-zero entries in $g$ if we only spend time on the non-empty $G_i$'s (and if a later update of $g$ creates a new non-empty subgraph $G_i$, we attribute the time to initialize the expander decomposition of $G_i$ to the \textsc{Scale} operation. In total we spend $O(m \log^9 n)$ time for the initialization since the average time for each edge is $O(\log^9 n)$.
	\paragraph{\textsc{Scale}.}
	Changing $g_e \leftarrow s$ means the edge weight of $e$ is changed to $s$.
	Thus we may need to delete the given edge $e$ from its current graph $G_{i}$
	and insert it into some other $G_{i'}$, 
	so that \eqref{eq:graph weight classified} is maintained. 
        For each corresponding change in $I$, we can check them in parallel for which $G_i$ to be deleted and which $G_i$ to be inserted. Then we again do a parallel sorting for each of the update and collect the insertion and deletion batch update for each $G_i$ in parallel, use \Cref{lem:dynamicExpanderDecomposition} to do all the batch updates. It takes 
	$\tO{1}$ amortized work per edge update, so the total amortized work for $O(|I|)$ updates is $\tO{|I|}$. The depth for updating is $\tO{1}$.

	\paragraph{\textsc{HeavyQuery}.} For every $i\neq -\infty$ we do the following for all the $\phi$-expander $G_{i,j}$ in parallel. Let $m'$ and $n'$ denote the number of edges and nodes in $G_{i,j}$ respectively.  
	Let $\mB$ be the incidence matrix of $G_{i,j}$; thus rows and columns of $\mB$ correspond to edges and nodes in $G_{i,j}$, respectively. 
	To simplify notations, we pretend rows of $\mb$ are in $\R^n$ instead of $\R^{n'}$ (i.e. keep a ghost column even for nodes in $G$ but not in $G_{i,j}$).\footnote{\label{foot:query:assumption}Without this assumption, $\mB$ might not contain all column in $\mA$ and we need to define vector $\hat{h}$ which contains $|V(G_{i,j})|$ entries of $h$ corresponding to nodes in $G_{i,j}$, which introduce unnecessary notations. However, it is crucial to actually work with $n'$-dimensional vectors for efficiency.}
	Note that $G_{i,j}$ is unweighted and each row of $\mB$ corresponds to an edge in $G_{i,j}$ and also appears as a row in $\mA\in \R^{m\times n}$; thus $\mB\in \R^{m'\times n}$ for some $m'\leq m$. 
	All $m'$-row/column (respectively $n$-row/column) matrices here have each row/column correspond to an edge (respectively node) in $G_{i,j}$, so we index them using edges $e$ and nodes $v$ in $G_{i,j}$; for example, we refer to a row of $\mB$ as $\mB_e$ and an entry of $h\in \R^n$ as $h_v$. 
	To answer the query (finding all $e$ such that $|(\mdiag(g)\mA h)_e|\geq \epsilon$), it suffices to find all edges $e$ in $G_{i,j}$ such that 
	\begin{align}
	|(\mdiag(g)\mB h)_e|\geq \epsilon \label{eq:query:original goal}
	\end{align}
	because rows of $\mB$ is a subset of rows of $\mA$. Finding $e$ satisfying \eqref{eq:query:original goal} is done as follows.\\
	{\em Step 1:} Shift $h$ so it is orthogonal to the degree vector of $G_{i,j}$: 
	\begin{align}
	h' \leftarrow h - 
	\onevec_n \cdot (\onevec_n^\top \mD h) / (\onevec_n^\top \mD \onevec_n),
	\label{eq:query:modify_h}
	\end{align} 
	where $\mD \in \R^{n \times n}$ is the diagonal degree matrix with $\mD_{v,v} = \deg_{G_{i,j}}(v)$.\\
	{\em Step 2:} Let $\delta=\epsilon/2^{i+1}$, and find all $e$ with 
	\begin{align}
	|(\mB h')_e|\geq \delta.\label{eq:query:new goal}
	\end{align}
	as follows.
	For every node $v$ with $|h'_v| \geq 0.5 \delta$,
	find all edges $e$ incident to $v$ that satisfies \eqref{eq:query:new goal}. 
	Among these edges, return those satisfying \eqref{eq:query:original goal}.   
	
	Since we did not change the algorithm, the correctness directly follows form the same proof as in \cite{BrandLN+20}.
    
	\smallskip\noindent{\em Time Complexity:} 
    
	Step 1 (computing $h'$) can be implemented to take $O(|V(G_{i,j})|)$ work (see \Cref{foot:query:assumption}) and 
        $\tO{1}$ depth since both $\onevec_n^\top \mD h$ and $\onevec_n^\top \mD \onevec_n$ only requires multiplication and summation over at most $|V(G_{i,j})|$ many non-zero entries. Summing over the work for all $i,j$, we get $\tO{n\log W}$ since for each $G_i$, the partition of the edges into expanders takes in total at most $\tO{n}$ many nodes according to \Cref{lem:dynamicExpanderDecomposition}, and there are at most $\log W$ different $G_i$'s.
        
	For step 2 (finding $e$ satisfying  \eqref{eq:query:new goal}), the depth is $\tO{1}$ since we can check all the edges in parallel. The work is proportional to the total number of edges being checked, which is $O(\sum_{i,j}\sum_{v\in\hat{V}_{i,j}}\deg_{G_{i,j}}(v))$ where $\hat{V}_{i,j}$ is defined as all the node $v\in V(G_{i,j})$ with $|h'_v|\ge 0.5\delta$. According to the same proof as in \cite{BrandLN+20}, we can show that the complexity is bounded by $\tO{\|\mdiag(g) \mA h\|_2^2 \epsilon^{-2}}$, so the lemma follows.

	\paragraph{\textsc{Sample}.}
	Again, let $G_{i,j}$ be an expander of the maintained decomposition
	and denote $\mD^{(i,j)}$ its diagonal degree matrix.
	Then we define $h^{(i,j)}$ as in \eqref{eq:query:modify_h} for each $G_{i,j}$
	which satisfies $(h_v - h_u)^2 = (h^{(i,j)}_v - h^{(i,j)}_u)^2$ and $h^{(i,j)} \bot \mD^{(i,j)} \onevec$.
	
	Next, let 
	$$
	Q = \frac{K}{
		\sum_{i,j} 2^{2i} \sum_{v \in V(G_{i,j})} 
		(h^{(i,j)}_v)^2 \deg_{G_{i,j}}(v)
	}
	$$ 
	then we perform the following procedure for each expander $G_{i,j}$ of our decomposition. For each node $v$ in $G_{i,j}$, we sample each edge incident to $v$ with probability
	$$
	\min\{ Q \cdot 2^{2i} (h^{(i,j)}_v)^2, 1\}.
	$$
	If an edge $(u,v)$ is sampled twice (i.e. once for $u$ and once for $v$), 
	then it is included only once in the output.
	
	Computing all $h^{(i,j)}$'s takes $\tO{n \log W}$ work and $\tO{1}$ depth as in the \textsc{HeavyQuery} analysis. Computing $Q$ can be done in $\tO{n\log W}$ work and $\tO{1}$ depth as well since we can compute $(h^{(i,j)}_v)^2 \deg_{G_{i,j}}(v)$ for difference $i,j,v$ in parallel and sum them up.
        Sampling of edges can be implemented 
	such that the work is bounded by the number of sampled edges: first sampling a binomial for each node indicating how many adjacent edges we should sample for that vertex, 
	and then pick the corresponding number of incident edges uniformly at random. It is proven in \cite{BrandLN+20} that the number of sampled edges is at most $O(K\log n)$ with high probability, so the complexity follows. The correctness follows trivially because we did not change the algorithm.

	\paragraph{\textsc{Probability}.}
	As discussed in the analysis of \textsc{Sample}, we can compute $Q$ and all $h^{(i,j)}$'s in $\tO{n\log W}$ work and $\tO{1}$ depth, then for each edge $e\in I$, we can look up the subgraph $G_{i,j}$ it belongs to, and compute its probability $q_e$ used in \textsc{Sample} in $O(1)$ times, which takes $\tO{|I|}$ work in total for all edges in $I$. We do them in parallel, so the depth is $\tO{1}$. 
	\paragraph{\textsc{LeverageScoreSample}.}
	We will sample the edges in each expander $G_{i,j}$ separately in parallel. According to the proof in \cite{BrandLN+20}, it suffices to sample each edge incident to $v$ with the following probability for $O(\log n)$ times, and the correctness will follow, while the expected number of sampled edges can be bounded by $O(K'n\log^{11}n\log W)$\[
	p_v = \min\left\{\frac{16K'\phi^{-2}}{\deg^{(i,j)}_v},1\right\},
	\]
    To implement the sample, we first compute the value $P_v$ for each $v$ in parallel, cost $\tO{n\log W}$ work in total and $\tO{1}$ depth, then sample edges using work proportional to the number of sampled edges as in the last operation. 
\end{proof}

By following the same idea as in Lemma F.1 of \cite{BrandLLSS0W21}, we get the following corollary.
\begin{corollary}\label{cor:heavyhitter}
    If we assume $A$ is by removing a column of an edge incident matrix of a graph $G$, \Cref{lem:large_entry_datastructure} still holds for the operations \textsc{Initialize}, \textsc{Scale}, \textsc{HeavyQuery}.
\end{corollary}
\begin{proof}
    See the proof of Lemma F.1 of \cite{BrandLLSS0W21}, removing a column makes $A$ be like an incident matrix of a graph (which can be maintained by \Cref{lem:large_entry_datastructure}) concatenated with at most $n$ rows of vectors with only $1$ non-zero entry, for which we maintain them explicitly: every \textsc{HeavyQuery} we just explicitly computes $\left(\mdiag(g) \mA h\right)_e$ where $e$ corresponds to a row for $A$. 
\end{proof}

\paragraph{\textsc{LeverageScoreBound}.} We can get the exact probability an edge $e$ is sampled in the previous implementation of \textsc{LeverageScoreSample}. Suppose $e=(u,v)$ is in subgraph $G_{i,j}$, we can compute the probability $p_u,p_v$ in parallel in $\tO{|I|}$ work and $\tO{1}$ depth. 

%% file: LeverageScore.tex
\section{Maintaining Regularized Lewis-Weights}\label{sec:MaintainingRegularizedLewis-Weights}

In this section, we provide an efficient data structure for maintaining Lewis-Weights that supports updating and querying, which will be used in our main algorithm. It requires a subroutine maintaining leverage score, which is provided in \Cref{subsec:Maintainingleveragescore}.

\begin{theorem}\label{thm:lewis_weight_maintenance}
Let $z\in\R^m$ with $z_i \ge n/m$ for every $i\in [m]$.
Let $\tau(\mG \mA) \in \R^m$ such that $\tau(\mG \mA) = \sigma(\tau(\mG \mA)^{1/2-1/p} \mG \mA) + z$ for $p \in (0,2]$,
i.e. $\tau(\mG \mA)$ are regularized Lewis weights of $\mG \mA$.
There exists a Monte-Carlo data-structure (\Cref{alg:lewis_weight_maintenance}), 
that works against an adaptive adversary, 
with the following procedures: 
\begin{itemize}
\item \textsc{Initialize}$(\mA \in \R^{m \times n}, g \in \R_{\ge 0}^m, z\in\R_{>0}^m, p \in [1/2,2), 
	\delta >1, \epsilon \in (\frac{1}{\polylog(n)},\frac{1}{2^{10}\delta\cdot\log n}])$:
	The data structure initializes for the given matrix $\mA \in\{-1,0,1\}^{m \times n}$ which is from deleting one column of an adjacency matrix of a directed graph with $m$ edges and $n$ vertices, 
	scaling $g \in \R_{\ge0}^m$,
	regularization parameter $z \in \R^m$,
	Lewis weight parameter $p \in [1/2,2)$, 
	and target accuracy $\epsilon \in (\frac{1}{\polylog(n)}, \frac{1}{2^{10}\delta\cdot\log n}]$.
	The parameter $\delta<\text{poly}(n)$ is a bound on how much the vector $g$ is allowed to change per iteration:
	Let $g\t$ be the vector $g$ during the $t$-th call to \textsc{Query} (with $g^{(0)}$ during \textsc{Initialization}),
	then we assume that $g\t \approx_{\delta\epsilon} g^{(t-1)}$ for all $t$. The work is $\tO{m}$ and depth is $\tO{1}$. 
\item \textsc{Scale}$(I \subseteq [m], b \in \R_{\ge0}^I)$: 
	Sets $g_{i} \leftarrow b_i$ for each $i\in I$. The amortized work for each call is $\tO{\frac{m}{n\eps^{O(\log \delta)}}\cdot\sum_{i\in[I]}\tau(\mG^{(t)}\mA)_i}$ and the depth is $\tO{1}$. 
\item \textsc{Query}$()$: 
	W.h.p.~in $n$ the data structure outputs a vector $\otau \in \R^m$
	with the property $\otau \approx_{\epsilon} \sigma(\otau^{1/2-1/p} \mG \mA) + z$
	and therefore $\otau \approx_{\epsilon} \tau(\mG \mA)$.
	The vector $\otau$ is returned as a pointer and the data structure also returns a set $I \subset [m]$ of indices $i$ where $\otau_i$ has changed compared to the last call to \textsc{Query}. The amortized work is $\tO{m+\frac{m}{\sqrt{n}}}$ and the depth is $\tO{1}$. 
\end{itemize}

For the claimed complexity to hold the following condition must be satisfied. There exists a sequence $\tg\t$ such that for all $t$ %
\begin{align}
g\t \in (1\pm1/(10^5 \log n))& ~\tg\t \label{eq:lw:nearby_sequence}\\
\| (\tw\t)^{-1}(\tw\t - \tw^{(t+1)})\|_{\tau(\tmG\t \mA)} &= O(1) \label{eq:lw:nearby_sequence_stable}
\end{align}
where we define $\tw\t := \tau(\tmG\t \mA)^{1/2-1/p}\tg\t$.
\end{theorem}

\begin{algorithm}[ht]
\caption{Algorithm of \Cref{thm:lewis_weight_maintenance} \label{alg:lewis_weight_maintenance}}
\SetKwProg{Members}{members}{}{}
\SetKwProg{Proc}{procedure}{}{}
\Members{}{
	$D_j$ for $j=1,...,O(1/p)$ \Comment{Data structure of \Cref{thm:leverage_score_maintenance}}\\
	$\ov^{(j)}\in \R^m$ for $j=1,...,O(1/p)$ \\ 
	$g \in \R^m$, $p \in [1/2,2)$, $L \in \N$, $\epsilon > 0$
}
\Proc{\textsc{Initialize}$(\mA \in \R^{m \times n}, g \in \R^m, z \in \R^m, p \in [1/2,2), \delta>1, \epsilon \in (0,1/(2^{10}\delta\log n)])$}{
	Compute $\ov^{(1)}$ with $\ov^{(1)} \approx_\epsilon \sigma((\omV^{(1)})^{1/2-1/p} \mG \mA) + z$ \label{line:lw:initialize_v1}\\
	$L \leftarrow \lceil(\log_{4/3}(200\delta)+1\rceil$ \\
	\For{$j = 1,...,L$}{
		$D_j.\textsc{Initialize}(\mA, (\omV^{(j)})^{1/2-1/p} g, z, \epsilon /(40L))$ \label{line:lw:initialize_D}\\
		$\ov^{(j+1)} \leftarrow ((\ov^{(j)})^{2/p-1} D_j.\textsc{Query}())^{p/2}$
	}
	$g \leftarrow g$, $p \leftarrow p$, $\epsilon \leftarrow \epsilon$
}
\Proc{\textsc{Scale}$(I, b)$}{
	$g_I \leftarrow b$ \\
	$D_j.\textsc{Scale}(I, (\ov^{(j)}_I)^{1/2-1/p} b)$ for $j=1,...,L$ \label{line:lw:scale_Dj}
}
\Proc{\textsc{Query}$()$}{
	\tcp{Maintain $\ov^{(j+1)} = ((\ov^{(j)})^{2/p-1} \osigma^{(j)})^{p/2}$ for $j=1,...,L-1$}
	\For{$j=1,...,L-1$}{
		$I_\osigma^{(j)}, \osigma^{(j)} \leftarrow D_j.\textsc{Query}()$ 
		\label{line:lw:DjQuery}\\
		$\ov^{(j+1)}_i \leftarrow ((\ov^{(j)}_i)^{2/p-1} \osigma^{(j)}_i)^{p/2}$ for $i \in I_\osigma^{(j)}$\\
		$D_{j+1}.\textsc{Scale}(i, (\ov^{(j+1)}_i)^{1/2-1/p} g_i)$ for $i \in I_\osigma^{(j)}$ \label{line:lw:update_Dj}\\
	}
	\tcp{Maintain $\ov^{(1)}_i \approx_\epsilon \tau(\mG \mA)_i$ for all $i\in[m]$,}
	\tcp{but update only if $\tau(\mG \mA)_i$ changed sufficiently.}
	\For{$i \in I_\osigma^{(L-1)}$ with $\ov^{(L)}_i \not\approx_{\epsilon/10} \ov^{(1)}_i$}{ \label{line:lw:check_change}
		$\ov^{(1)}_i \leftarrow \ov^{(L)}_i$ \label{line:lw:update_v1}\\
		$D_1.\textsc{Scale}(i, (\ov^{(1)}_i)^{1/2-1/p} g_i)$ \label{line:lw:update_D1}
	}
	\Return $I_\osigma^{(L)}, \ov^{(1)}$
}
\end{algorithm}

\begin{proof}%[Proof of \Cref{thm:lewis_weight_maintenance}]
    The correctness follows from the same proof in Section 5 of \cite{BrandLLSS0W21} since we did not change the algorithm. We only need to provide the parallel implementation of each step of the data structure and show the complexity.

    Notice that according to Lemma 5.6 of \cite{BernsteinBGNSS022}, all calls to \Cref{thm:leverage_score_maintenance} satisfies the requirements, so the complexity in \Cref{thm:leverage_score_maintenance} holds.

    \paragraph{\textsc{Initialize}} The first step requires us to solve the equation $\ov^{(1)} \approx_\epsilon \sigma((\omV^{(1)})^{1/2-1/p} \mG \mA) + z$ to compute $\bar{v}^{(1)}$. As shown at the end of Section 5 in \cite{BrandLLSS0W21}, this can be done by computing $\tO{1}$ times of approximate leverage scores, which according to \Cref{thm:leverage_score_maintenance}, can be done in $\tO{m}$ work and $\tO{1}$ depth. Moreover, for each $j\in[L]$, we need to initialize the data structure in \Cref{thm:leverage_score_maintenance} can call \textsc{Query} once. Since $\tO{L}=\tO{1}$, according to \Cref{thm:leverage_score_maintenance}, the total work is $\tO{m}$ and the total depth is $\tO{1}$.

    \paragraph{Depth for \textsc{Scale}} This step can be implemented by simply $L=\tO{1}$ call to the \textsc{Query} operation of \Cref{thm:leverage_score_maintenance}, which takes $\tO{1}$ depth. We will analyze the amortized work later.

    \paragraph{Depth for \textsc{Query}} The step for calling $D_j.\textsc{Query}$ takes $\tO{1}$ depth. After getting the indices $I^{(j)}_{\bar{\sigma}}$, we can compute $\bar{v}_i^{(j+1)}$ for every $i$ in parallel in $\tO{1}$ depth, then the call to $D_{j+1}.\textsc{Scale}$ uses $\tO{1}$ depth. In the end, for each $i\in I_{\bar{\sigma}}^{(L-1)}$, we can check the if condition and do the $D_1.\textsc{Scale}$ in parallel, the total depth is $\tO{1}$.

    \paragraph{Analysis for amortized work.} The amortized work analysis for \textsc{Scale} and \textsc{Query} is more complicated but is identical to the time analysis of Section 5 in \cite{BrandLLSS0W21}. So we only scratch the critical points as follows. The full proof can be found in Section 5 in \cite{BrandLLSS0W21}. Notice that the vector $c$ in their paper denotes the maximum number of non-zero entries of each row of $A$, which in our case is $2$ for every entry of $c$.

    We first give upper bounds for the amortized work for each call to $D_j.\textsc{Scale}$ and $D_j.\textsc{Query}$ by calling \Cref{thm:leverage_score_maintenance}. By Lemma 5.6 of \cite{BrandLLSS0W21}, the condition of \Cref{thm:leverage_score_maintenance} for the complexity to hold is satisfied. Thus, we can compute the work for a call to $D_j.\textsc{Scale}$ as
    $$
\tilde{O}\left(
\frac{m}{n} \sigma((\mV^{(j)})^{1/2-1/p} \mG \mA)_i + 1
\right)
=
\tilde{O}\left(
\frac{m}{n} \tau(\mG \mA)_i
\right)
$$
The amortized work to a call to $D_j.\textsc{Query}$ is
$$
\tilde{O}\left(
n+\frac{m}{\sqrt{n}}\right).
$$

As in \cite{BrandLLSS0W21}, we will first give the total work for running $T=\sqrt{n}$ iterations and charge them into each single call to \textsc{Scale} and \textsc{Query}. After every $\sqrt{n}$ iterations we re-initialize the data structure, so the initialization cost should also be charged (into \textsc{Query}). 

Firstly, the total work for \textsc{Scale} is
\begin{align}
\tilde{O}\left(
\sum_{t=1}^T \sum_{i\in[m]} \left(
\frac{m}{n} \tau(\mG^{(t-1)} \mA)_i
\right) \mathbf{1}_{g\t_i \neq g^{(t-1)}_i}
\right).
\label{eq:lw:cost:scale}
\end{align}

The total work for calling $D_j.\textsc{Query}$ inside \textsc{Query} is
\[\tO{T\cdot\left(n+\frac{m}{\sqrt{n}}\right)}.\]

As analyzed by \cite{BrandLLSS0W21}, the total work by calling $D_1.\textsc{Scale}$ in Line~\ref{line:lw:update_D1} is bounded by
\[\tO{\frac{T^2m}{n}}.\]

As analyzed by \cite{BrandLLSS0W21}, the total work by calling $D_1.\textsc{Scale}$ in Line~\ref{line:lw:update_Dj} is bounded by
\[\tO{\frac{m}{n\eps^{O(\log \delta)}}\cdot\left(\sum_{t=1}^T\sum_{i\in[m]}\tau(\mG^{(t)}\mA)_i \mathbf{1}_{g\t_i \neq g^{(t-1)}_i}\right)}.\]

The total work over $T$ iterations is thus bounded by
\[\tO{\underbrace{T\cdot\left(n+\frac{m}{\sqrt{n}}\right)+\frac{T^2m}{n}}_{\textsc{Query}}+\underbrace{\frac{m}{n\eps^{O(\log \delta)}}\cdot\left(\sum_{t=1}^T\sum_{i\in[m]}\tau(\mG^{(t)}\mA)_i \mathbf{1}_{g\t_i \neq g^{(t-1)}_i}\right)}_{\textsc{Scale}}}.\]

We marked which part to be charged to \textsc{Query} or \textsc{Scale}. Remember that we call \textsc{Initialization} for every $T=\sqrt{n}$ runs of query and charge the cost to query, so the amortized work for query is clearly

\[\tO{\left(n+\frac{m}{\sqrt{n}}\right)+\frac{Tm}{n}+\frac{m}{T}}=\tO{n+\frac{m}{\sqrt{n}}}.\]

The charged work for \textsc{Scale} depends on $\mathbf{1}_{g\t_i \neq g^{(t-1)}_i}$, but notice that this is exactly what is changed by the input $I$, so when we charge it to each individual call to \textsc{Scale}, the amortized work becomes

\[\tO{\frac{m}{n\eps^{O(\log \delta)}}\cdot\sum_{i\in[I]}\tau(\mG^{(t)}\mA)_i}.\]
\end{proof}

\subsection{Maintaining Leverage Score}\label{subsec:Maintainingleveragescore}

This section contains the subroutine used for maintaining the leverage score.

\begin{theorem}\label{thm:leverage_score_maintenance}
There exists a Monte-Carlo data-structure (\Cref{alg:leverage_score_maintenance}), 
that works against an adaptive adversary, 
with the following procedures. 
\begin{itemize}
\item \textsc{Initialize}$(\mA \in \{-1,0,1\}^{m \times n}, v \in \R^m, z\in\R^m, \epsilon \in (0,1) )$: 
	The data structure initializes on $A$ which is from deleting a column of an edge incidence matrix of a directed graph $G$ of $m$ edges and $n$ vertices, 
	scaling $v \in \R^m$,
	target accuracy $\epsilon>0$,
	and regularization parameter $z \in \R^m$ with $z_i \ge 3n/m$ for every $i\in[m]$,
	and returns a vector $\osigma \in \R^m$ with $\osigma \approx_\epsilon \sigma(\mV \mA) + z$. It takes $\tO{m/\eps^2}$ work and $\tO{1}$ depth.

\item \textsc{Scale}$(I \in [m], c \in \R^I_{\ge0})$: 
	For given $c \in \R^I_{\ge0}$ satisfying $c_i \approx_{0.25} v_i$ for every $i\in I$, set $v_{i} \leftarrow c_i$ for every $i\in I$. Suppose $v$ is the member vector when \textsc{Scale} is called, then it takes $\tO{\frac{m}{n\eps^4}\sum_{i\in I}\sigma(\mV\mA)_i+\frac{|I|}{\eps^2}}$ amortized work, and $\tO{1}$ depth.
\item \textsc{Query}$()$: 
	W.h.p.~in $n$ the data-structure outputs a vector $\osigma \in \R^m$
	such that $\osigma \approx_{\epsilon} \sigma(\mV \mA) + z$.
	The vector $\osigma$ is returned as a pointer and the data structure also returns a set $I \subset [m]$ of indices $i$ where $\osigma_i$ has changed compared to the last call to \textsc{Query}. It takes $\tO{\frac{n}{\eps^2}+\frac{m}{\eps^4\sqrt{n}}}$
\end{itemize}

For the complexity to hold, the following condition must be satisfied. Let $v^{(t)}$ be the vector $v$ during the $t$-th call to \textsc{Query} (and $v^{(0)}$ during the initialization). The data structure assumes there exists a sequence $\tilde{v}^{(t)}$ such that for all $t$
\begin{align}
v\t \in (1\pm1/(64 \log n)) \tv\t \label{eq:ls:nearby_sequence}\\
\| (\tv\t)^{-1}(\tv\t - \tv^{(t+1)})\|_{\sigma(\tmV\t \mA)} = O(1) \label{eq:ls:nearby_sequence_stable}
\end{align}
for all $t$.
\end{theorem}

\begin{algorithm}
\caption{Data structure for maintaining leverage scores \label{alg:leverage_score_maintenance}}
\SetKwProg{Members}{members}{}{}
\SetKwProg{Proc}{procedure}{}{}
\Members{}{
$t,T,r \in \N$, $v \in \R^m$, $\Delta^{(j)} \in \R^m$ and $S_j,C_j \subset [m]$ for $j=1,...,0.5 \log n$
}
\Proc{\textsc{Initialize}$(\mA, v^{\init}, z, \epsilon)$}{
	$t \leftarrow 0$, 
	$v \leftarrow v^{\init}$, 
	$z \leftarrow z$,
	$T \leftarrow \lceil\epsilon^2\sqrt{n}\rceil$ \\
	Compute $\osigma \approx_\epsilon \sigma(\mV \mA) + z$ \label{line:ls:init_sigma} \\
	Let $r = O(\log n)$ be such that a $r \times m$ JL-matrix yields a $1/2$-approximation. \\
	\For{$j=0,...,\log T$}{
            Let $D_j$ be a \textsc{HeavyHitter} data structure from \Cref{lem:large_entry_datastructure} \\
		$D_j.\textsc{Initialize}(\mA,v\cdot z^{-1/2})$ \\
		$S_j\leftarrow \emptyset$\\
		$C_j \leftarrow \emptyset$ \\
		$\Delta^{(j)} \leftarrow \zerovec_m$ \label{line:ls:init_Delta_j}
	}
	\Return $\osigma$
}
\Proc{\textsc{FindIndices}$(h\in\R^n)$}{
	$I \leftarrow \emptyset$\\
	\For{$j=\log T,...,0$}{
		\If{$2^j | t$}{
			$\mS_{i,i}=1$ for $i \in S_j \cup C_j$, 
			and for other $i$ we set $\mS_{i,i} = 1/p_i$ with probability 
			$p_i = \min(1,~c \osigma_i \epsilon^{-2} \log n \log\log n)$ 
			for some large constant $c > 0$, 
			and $\mS_{i,i} = 0$ otherwise. \label{line:ls:S_findindices}\\
			$\mR \leftarrow$ $m \times r$ JL-matrix \label{line:ls:JL_find_indices}\\
			Let $\mM' \approx_{1/(64 n)}(\mA^\top(\mV - \Delta^{(j)})^2 \mS^2 \mA)^{-1}$ and $\mM \approx_{1/(64 n)} (\mA^\top\mV^2 \mS^2 \mA)^{-1}$ \\
			$\mH \leftarrow (\mM' \mA^\top (\mV - \Delta^{(j)}) \mS \mR 
				- \mM \mA^\top \mV \mS \mR$ \label{line:ls:H_difference}\\
			$I \leftarrow I \cup D_j.\textsc{HeavyQuery}(\mH \unitvec_k, \epsilon /(48 r\log n))$ for all $k \in [r]$ \label{line:ls:add_query}\\
			$\Delta^{(j)} \leftarrow \zerovec_m$ \label{line:ls:reset_Delta_j}\\
			$D_j.\textsc{Scale}(i, v_i z_i^{-1/2})$ for $i \in S_j$ \label{line:ls:scale_D_j}\\
			$I \leftarrow I \cup S_j$, $S_j \leftarrow \emptyset$, $C_j \leftarrow \emptyset$ \label{line:ls:add_trivial_sets}
		}
	}
	\Return $I$
}
\Proc{\textsc{UpdateIndices}$(I)$}{
	$\mS_{i,i} = 1/\sqrt{p_i}$ with probability $p_i = \min(1,~ c \epsilon^{-2} \osigma_i \log n)$ for some large enough constant $c > 0$, and $\mS_{i,i} = 0$ otherwise. \\
	$\mR \leftarrow$ $\exp(\pm\epsilon/16)$-accurate JL-matrix \\
	$\mH \leftarrow \mM \mA^\top \mV \mS \mR$ for any $\mM \approx_{\epsilon/(64 \log n)} (\mA^\top \mV^2 \mS^2 \mA)^{-1}$ \label{line:ls:H_updateindices}\\
	$I' \leftarrow \emptyset$ \\
	\For{$i \in I$}{
		\If{$\osigma_i \not\approx_{3\epsilon/8} \|e_i^\top \mV \mA \mH\|_2^2 + z_i$}{ \label{line:ls:check_change}
			$\osigma_i \leftarrow \| e_i^\top \mV \mA \mH \|_2^2 + z_i$ \label{line:ls:update_osigma}\\
			$C_j \leftarrow C_j \cup \{ i \}$ for $j = 0,...,\log n$ \label{line:ls:add_C}\\
			$I' \leftarrow I' \cup \{ i \}$
		}
	}
	\Return $I'$
}
\Proc{\textsc{Scale}$(I, c)$}{
	\For{$j=0,...,\log T$}{
		$S_j \leftarrow S_j \cup \{I\}$ \label{line:ls:add_S}\\
		$D_j.\textsc{Scale}(I, 0)$ \label{line:ls:scale_D_j_0}\\
		$\Delta^{(j)} \leftarrow \Delta^{(j)} + c - v$ \Comment{Maintain $v^{(t)} = v^{(t_j)} + \Delta^{(j)}$} \label{line:ls:update_Delta_j}
	}
    
        \For{$i\in I$}{
        $v_i \leftarrow c_i$ }
    }
\Proc{\textsc{Query}$()$}{
	\lIf{$t = T$}{\Return $[m]$, \textsc{Initialize}$(\mA, v, w, \epsilon)$}
	$t \leftarrow t + 1$\\
	$I \leftarrow \textsc{FindIndices}()$ \label{line:findIndices}\\
	$I \leftarrow \textsc{UpdateIndices}(I)$ \label{line:verify_I}\\
	\Return $I$, $\osigma$
}
\end{algorithm}

The correctness of \Cref{alg:leverage_score_maintenance} is from Section C.1 in \cite{BrandLLSS0W21}. We only need to show how to implement it in parallel. We first prove the following lemma, which is frequently used in our implementation.
\begin{lemma}\label{lem:solver}
    Suppose vector $v\in\mR^m$ and $\mA\in\{-1,0,1\}^{m\times n}$ has at most $2$ non-zero entries in each row, and suppose we are given access to non-zero entries of $v,\mA$, then we can solve $(\mA^T\mV\mA)^{-1}b$ for any $b\in\mR^n$ in $\tO{\nnz(v)}$ work and $\tO{1}$ depth.
\end{lemma}
\begin{proof}
    We first compute all the non-zero entries of $\sqrt{\mV}\mA$ in $\tO{\nnz(v)}$ work and $\tO{1}$ depth: we write $\mA$ as a block matrix of rows $a_1,...,a_m$, so that $\mV\mA=\sum_{i\in[m]}\sqrt{v_i}a_i$. Notice that computing $a_i^Ta_i$ costs $\tO{1}$ work since $\mA$ only contains rows with at most $2$ non-zero entries. After that, we sum them up in parallel. so the total work is $\tO{\nnz(v)}$ and depth is $\tO{1}$. The outputs are all the non-zero entries of $\sqrt{\mV}\mA$, which has size at most $\tO{v}$.

    After having $\mM:=\sqrt{\mV}\mA$, we use \Cref{lem:parallelsolver} to compute $(\mM^T\mM)^{-1}b=(\mA^T\mV\mA)^{-1}b$ in $\tO{\nnz(\mM)}=\tO{v}$ work and $\tO{1}$ depth. 
\end{proof}
\paragraph{\textsc{Initialize}} We assume the inputs are given by non-zero entries. The initialization requires to compute $\bar{\sigma}$, which can be done by approximating $\sigma(\mV \mA)$. It is known that approximating the leverage score of $\mV\mA$ can be achieve by solving $\tO{1/\eps^2}$ instances of $(\mA^T\mV^T\mV\mA)^{-1}b$ and $\mA c$ for some vectors $b,c$ (for example, see \cite{LeeS13a}, end of Section B.2). According to \Cref{lem:solver}, solving $(\mA^T\mV^T\mV\mA)^{-1}b$ requires $\tO{m}$ work and $\tO{1}$ depth. 
Compute $Ac$ for some vector $c$ can be easily done in $\tO{\nnz(\mA)}=\tO{m}$ work and $\tO{1}$ depth. So the total work is $\tO{m/\eps^2}$ and depth is $\tO{1}$.

Then, we need to initialize $T$ instances of \textsc{HeavyHitter}. Each one of them receives the parameter $A$ and $v\cdot z^{-1/2}$ (which can be computed easily in parallel $\tO{m}$ work and $\tO{1}$ depth), according to \Cref{lem:large_entry_datastructure}, the initialization takes $\tO{m}$ work and $\tO{1}$ depth.

\paragraph{\textsc{FindIndices}} We can bound the work similar to Lemma C.13 of \cite{BrandLLSS0W21}. Notice that $\log T$ is hidden in $\tO{1}$, so we only need to analyze the inner loop. The inner loop contains two major parts, cost for computing $\mH$, and costs for operations on the heavy hitter data structure $D_j$. To compute $\mH$, we need to first computed $a=\mA^\top (\mV - \Delta^{(j)}) \mS \mR$ and $b=\mA^\top \mV \mS \mR$, there $\mR$ can be easily computed in $\tO{n}$ work and $\tO{1}$ by randomly sample $O(\log n)$ vectors; then the matrix multiplication can be done in $\tO{1}$ depth, and work depends on non-zero entries of the following two matrix $\bar{A}:=\mS\mV\mA$ and $\bar{A}':=\ms(V-\Delta^{(j)})\mA$. The last operation is to multiply $\mM$ and $\mM'$ to the corresponding $a,b$, which according to \Cref{lem:solver}, can be done in depth $\tO{1}$ and work $\tO{\nnz(\bar{A})+\nnz(\bar{A}')}$. In summary, the work for computing $\mH$ is $\tO{\nnz(\bar{A})+\nnz(\bar{A}')}$. This value over $T$ iterations can be bounded by the following equation, and the proof is the same as the proof of Lemma C.13 in \cite{BrandLLSS0W21} for analyzing running time. Notice that the vector $c\in\mR^m$ in \cite{BrandLLSS0W21} represents the number of non-zero entries of each row of $\mA$, where in our case $c$ is simply $2$ for every entry. %Moreover, computing $\mH$ is called every $2^j$ calls to \textsc{FindIndices}, so the total work over $T$ calls, according to the same computation of Lemma C.13 of \cite{BrandLLSS0W21}, is bounded by 

\[\tilde{O}\left(
T n \epsilon^{-2}
+ \sum_{t=1}^T \sum_{v^{(t)}_i \neq v^{(t-1)}_i} \left(\frac{m}{n \epsilon} \sigma(\mV^{(t-1)} \mA)_i+1\right)
\right)\]

The costs for $D_j.\textsc{HeavyQuery}$ and $D_j.\textsc{Scale}$ can be derived from \Cref{lem:large_entry_datastructure,cor:heavyhitter},which gives us in total $\tO{1}$ depth and work following the same computation as in Lemma C.13 of \cite{BrandLLSS0W21}, which gives us amortized work in the $t$-th call as

\[\tilde{O}\left(
\frac{m}{n\epsilon^2} \cdot \left(T + \sum_{v^{(t)}_i \neq v^{(t-1)}_i} \sigma(\mV^{(t-1)} \mA)_i\right)+n\right).\]

In total, the depth is $\tO{1}$, and the amortized work for the $t$-call is is
\begin{align*}
\tilde{O}\left(
\underbrace{n \epsilon^{-2}
+ \frac{m}{n\epsilon^2} \cdot T
+ n}_{\textsc{Query}}
+ \underbrace{\sum_{v^{(t)}_i \neq v^{(t-1)}_i} \left(\frac{m}{n \epsilon^2}\sigma(\mV^{(t-1)} \mA)_i+1\right)}_{\textsc{Scale}}
\right).
\end{align*}

In the equation, we also marked which part of the work should be charged to the amortized work of \textsc{Query} and also \textsc{Scale.}

\paragraph{\textsc{UpdateIndices}} According the proof of Lemma C.11 of \cite{BrandLLSS0W21}, the matrix $\mS$ has $O(\eps^{-2}n\log n)$ non-zero entries and $\mR$ has $\eps^{-2\log n}$ columns (of random entries), so $b:=\mA^T\mV \mS \mR$ can be computed in parallel in $\tO{n/\eps^4}$ work and $\tO{1}$ depth. To compute $\mH$, we need to solve $(\mA^T\mV^2\mS^2\mA)^{-1}b$ approximately, which according to \Cref{lem:solver} can be done in $\tO{1}$ depth, and the work depends on the number of non-zero entries of $\mV^2\mS^2$, which is at most $\tO{n/\eps^2}$ according to $\mS$.

After computing $\mH$, we need to do the loop for each $i\in I$ in parallel. Notice that computing $e_i^\top \mV \mA \mH$ can be done by picking the $i$-th column of $\mV\mA\mH$, which only costs $\tO{1/\eps^2}$ work. We can check the if condition and update $\bar{\sigma_i},\mC_j,I'$ by using standard parallel set data structure. The total depth is $\tO{1}$. The work is $\tO{|I|/\eps^2}$. To give a more transparent bound on $|I|$, we use Lemma C.10 of \cite{BrandLLSS0W21}, which shows that the amortized work for the $t$-th call is 

\[\underbrace{\tO{\frac{Tm}{n\eps^4}}}_{\textsc{Query}}+\underbrace{\sum_{v^{(t)}_{i}\not=v^{(t+1)}_i}\left(\frac{1}{\eps^2}+\frac{m}{n\eps^4}\sigma\left(\mV^{(t-1)}\mA\right)_i\right)}_{\textsc{Scale}}.\]

\paragraph{\textsc{Scale}} Updating $S_j$ can be done by standard parallel set data structure. The call to $D_j.\textsc{Scale}(I,0)$, according to \Cref{cor:heavyhitter}, can be done in $\tO{|I|}$ work and $\tO{1}$ depth. Updating $\Delta$ and $v_i$ can be done trivially in $\tO{|I|}$ work and $\tO{1}$ depth.

Notice that we have the amortized work charged from \textsc{HeavyQuery} and \textsc{UpdateIndices} every time $v^{(t)}$ changed by the input $I$, so the total work is

\[\tO{\frac{m}{n\eps^4}\sum_{i\in I}\sigma(\mV^{(t-1)}\mA)_i+\frac{|I|}{\eps^2}}.\]

Here $v^{(t-1)}$ is the member vector $v$ when \textsc{Scale} is called.

\paragraph{\textsc{Query}} For every $T$ iterations, we call \textsc{Initialize} again, which uses $\tO{1}$ and the amortized work depends on the work for \textsc{Initialize}, can be calculated as $\tO{m/\sqrt{n}\eps^4}$.

After that, \textsc{FindIndices} and \textsc{UpdateIndices} can be done in $\tO{1}$ depth, the corresponding work charged to each call of \textsc{Query} along with \textsc{Initialization} can be computed as

\[\tO{n \epsilon^{-2}
+ \frac{m}{n\epsilon^2} \cdot T
+ n+\frac{Tm}{n\eps^4} + \frac{m}{\sqrt{n}\eps^4}}=\tO{\frac{n}{\eps^2}+\frac{m}{\eps^4\sqrt{n}}}.\]

%% file: primalGradientMaintenance.tex
\section{Parallel Primal and Gradient Maintenance}\label{sec:primal}

In their IPM for solving LPs with two-sided constraints, as well as some special cases, such as the Minimum-Cost Flow Problem, \cite{BrandLLSS0W21} introduce an IPM that operates by taking $\tO{\sqrt{n}}$ steps in the direction of the steepest descent. In each step, a potential function is decreased to get closer to the optimal feasible point. To decrease this potential function, they use the gradient of the mentioned potential function. In order to compute and maintain (an approximation of) this gradient, as well as maintaining an approximation of the primal solution $x \in \R^m$, \cite{BrandLLSS0W21} discuss a data structure, namely the Primal/Gradient Maintenance data structure. This data structure was originally introduced by \cite[Lemma 7.2]{BrandLN+20} and modified by \cite[Lemma D.1]{BrandLLSS0W21}.

In this section, we discuss how this data structure can be implemented in the parallel setting. The main result of this section is \Cref{thm:gradient_maintenance}.

%\begin{restatable*}[Primal/Gradient Maintenance]{theorem}{gradientMaintenance}
\begin{theorem}[Parallel Primal/Gradient Maintenance, {\cite[Lemma D.1]{BrandLLSS0W21}}]
\label{thm:gradient_maintenance} 
There exists a deterministic data-structure that supports the following operations
\begin{itemize}
\item $\textsc{Initialize }(\mA\in\R^{m\times n}, x^{\init} \in \R^m, g\in\R^{m}, \ttau\in\R^{m}, z\in\R^{m}, w \in [0,1]^m, \epsilon>0)$:
	The data-structure preprocesses the given matrix $\mA\in\R^{m\times n}$, vectors $x^{\init},g,\ttau,z\in\R^{m}$, and the accuracy parameters $w \in [0,1]^m$ and $\epsilon>0$ in $\tO{\nnz(\mA)}$ work and $\tO{1}$ depth. We denote $\mG$ the diagonal matrix $\mdiag(g)$. 
	The data-structure assumes $0.5\le z\le2$ and $n/m\le\ttau\le2$.
\item $\textsc{Update}(I \subseteq [m], b \in \R^{I}, c \in \R^{I}, d \in \R^{I})$: 
	Sets $g_{i}\leftarrow b_i$, $\ttau_{i} \leftarrow c_i$ and $z_i \leftarrow d_i$ for each $i \in I$. This operation takes $\tO{\sum_{i \in I}\nnz(a_i)}$ work and $\tO{1}$ depth. 
	The data-structure assumes $0.5\le b_i\le2$ and $n/m\le c_i\le2$ for all $i \in I$. 
\item $\textsc{SetAccuracy}(I, \delta \in [0,1]^{I})$
	Sets $w_i \leftarrow \delta_i$ for all $i \in I$ in $\tO{|I|}$ work and $\tO{1}$ depth. %
\item $\textsc{QueryProduct}()$: 
	Returns a vector $\ov \in \R^n$ such that there exists a $\otau \in \R^m$ with $\otau \approx_\epsilon \ttau$, and a $\oz \in \R^m$ with $\|\oz-z\|_{\infty}\le \epsilon$ such that
    \begin{align*}
        \ov = \mA^{\top}\mG(\nabla\Psi(\oz))^{\flat(\otau)} \in \R^n,
    \end{align*}
    where $x^{\flat(\otau)} := \argmax_{\|w\|_{\otau + \infty} \le 1} \langle x, w \rangle$.
	Every call to \textsc{QueryProduct} must be followed by a call to \textsc{QuerySum},
	and we bound their complexity together (see \textsc{QuerySum}).
\item $\textsc{QuerySum}(h \in \R^m)$:
	Let $v^{(\ell)}$ be the vector $\mG\nabla\Psi(\oz)^{\flat(\otau)}$ used for the result of the $\ell$-th call to \textsc{QueryProduct}.
    Let $h^{(\ell)}$ be the input vector $h$ given to the $\ell$-th call to \textsc{QuerySum}.
	We define 
	\begin{align*}
	   x^{(t)} := x^{\init} + \sum_{\ell=1}^{t} \left( v^{(\ell)} + h^{(\ell)} \right).  
	\end{align*}
	Then the $t$-th call to \textsc{QuerySum} returns a vector $\ox \in \R^m$ with
	$$\|w^{-1}(\ox - x^{(t)})\|_\infty \le \epsilon.$$

	Assuming the input vector $h$ is given in a sparse representation (e.g. a list of non-zero entries), then after $T$ calls to \textsc{QuerySum} and \textsc{QueryProduct} the total work for all calls together is bounded by
	\begin{align*}
	\tO{
		T n 
		+ \sum_{\ell=0}^T \|h^{(\ell)}\|_0 
		+ T \cdot \sum_{\ell=1}^T \|v^{(\ell)}/w^{(\ell-1)}\|_2^2
	},
	\end{align*} 
    and total depth by $\tO{T}$.
	The output $\ox \in \R^m$ is returned in a compact representation to reduce the size. In particular, the data-structure returns a pointer to $\ox$ 
and a set $J \subset [m]$ of indices which specifies which entries of $\ox$ have changed 
	between the current and previous call to \textsc{QuerySum}.
\item $\textsc{ComputeExactSum}()$:
	Returns the exact $x^{(t)}$ in $\tO{m}$ work and $\tO{1}$ depth.
\item $\textsc{Potential}()$:
	Returns $\Psi(z)=\sum_i \cosh(\lambda z_i)$ in $\tO{1}$ work and depth.  
\end{itemize}
\end{theorem}
%\end{restatable*}

As \cite{BrandLLSS0W21} describe, this data structure is essentially the combination of two other data structures. Therefore, in order to prove \Cref{thm:gradient_maintenance}, we first need to discuss the other two data structures, namely the Gradient Reduction data structure and the Gradient Accumulator data structure.

\subsection{Gradient Maintenance}

As \cite{BrandLLSS0W21} discuss, the main idea for maintaining the gradient efficiently is that, for the $m$-dimensional gradient $\nabla(z)^{\flat(\ttau)}$, we can distribute the $m$ dimensions into $O(\epsilon^{-2} \log n)$ buckets, so that the indices in the same bucket have similar $\ttau_i$ and $z_i$ values. This would yield a $\nabla(\oz)^{\flat(\otau)}$ for some $\oz$ and $\otau$, where $\|\oz - z\|_{\infty} \leq \epsilon$ and $\otau \approx_\epsilon \ttau$. This is the main task of the Gradient Reduction data structure.

Before discussing the implementation of this data structure in the parallel setting, we first state a tool used for gradient reduction. \Cref{lem:project_mixed_norm,cor:compute_flat} introduce a maximizer that computes the steepest descent step with respect to a custom norm.

\begin{lemma}[{\cite[Theorem 62 and Algorithm 8]{lee2020solvinglinearprogramssqrtrank}}]
\label{lem:project_mixed_norm} 
For any $a\in \R^n$ and $l\in\R_{>0}^{n}$, there is an algorithm which outputs a solution to
\begin{align*}
    \max_{\norm{x}_{2}+\norm{l^{-1} x}_{\infty} \leq 1}\left\langle a,x\right\rangle \label{eq:project_our_form}
\end{align*}
in $O(n\log n)$ total work and $O(\log n)$ depth.
\end{lemma}

\begin{corollary}\label{cor:compute_flat}
Given $x \in \R^n$, $v \in \R^n$, we can compute
\begin{align*}
u = \underset{ \|v w\|_2 + \|w\|_\infty \le 1 }{ \argmax } \langle x,w \rangle 
\end{align*}
in $\tO{n}$ total work and $\tO{1}$ depth. 
\end{corollary}
\begin{proof}
    This is almost exactly the Corollary 7.4 of \cite{BrandLN+20}. As they show, $u$ can be computed via the algorithm by \Cref{lem:project_mixed_norm} in combination with an entry-wise multiplication of two vectors of length $n$. This process takes $O(n)$ additional work as well as $\tO{1}$ additional depth.
\end{proof}

Now, we are ready to state the main result of this subsection:

\begin{lemma}[Gradient Reduction, {\cite[Lemma 7.2]{BrandLN+20} and \cite[Lemma D.2]{BrandLLSS0W21}}]
\label{lem:gradient_reduction} 
There exists a deterministic data-structure that supports the following operations
\begin{itemize}
\item $\textsc{Initialize }(\mA\in\R^{m\times n},g\in\R^{m},\ttau\in\R^{m},z\in\R^{m},\epsilon>0)$:
    The data-structure preprocesses the given matrix $\mA\in\R^{m\times n}$, vectors $g,\ttau,z\in\R^{m}$, and accuracy parameter $\epsilon>0$ in total $\tO{\nnz(\mA)}$ work and $\tO 1$ depth. The data-structure assumes $0.5\le z\le2$ and $n/m\le\ttau\le2$. The output is a partition $\bigcup_{k=1}^K I_k = [m]$ with $K = O(\epsilon^{-2} \log n)$.
\item $\textsc{Update}(M \subseteq [m], b \in \R^M, c \in \R^M, d \in \R^M)$: 
	For $i \in M$ sets $g_{i}=b_i$, $\ttau_{i}=c_i$ and $z_i=d_i$ in total $\tO{\sum_{i \in M}\nnz(a_i)}$ work and $\tO{1}$ depth.
	The data-structure assumes $0.5\le z\le2$ and $n/m\le\ttau\le2$.
	The indices in $M$ might be moved to a different partition set, so the data-structure returns a list $u$ such that $i \in I_{u_i}$ for all $i \in M$.
\item $\textsc{Query}()$: 
	Returns a vector $\ov \in \R^n$ such that there exists a $\otau \in \R^m$ with $\otau \approx_\epsilon \ttau$, and a $\oz \in \R^m$ with $\|\oz-z\|_{\infty}\le \epsilon$ such that
    \begin{align*}
        \ov = \mA^{\top}\mG(\nabla\Psi(\oz))^{\flat(\otau)} \in \R^n,
    \end{align*}
    where $x^{\flat(\otau)} := \argmax_{\|w\|_{\otau + \infty} \le 1} \langle x, w \rangle$.
    The data structure further returns the low-dimensional representation $s \in \R^K$ such that
    \begin{align*}
        \sum_{k=1}^K s_k \mathbf{1}_{i \in I_k} = \left( \nabla\Psi(\oz)^{\flat(\otau)}\right)_i
    \end{align*}
    for all $i \in [m]$, in total $\tO{n}$ work and $\tO{1}$ depth.
\item $\textsc{Potential}()$
	Returns $\Psi(z)$ in $\tO{1}$ work and depth. 
\end{itemize}
\end{lemma}

\begin{proof}[Proof of \Cref{lem:gradient_reduction}]
    The data structure is presented in \Cref{alg:gradient_reduction}. This data structure, along with \Cref{alg:gradient_reduction}, represents the parallel versions of the data structure introduced by \cite[Lemma 7.2 and Algorithm 5]{BrandLN+20}. Consequently, the correctness of the data structure follows from Lemma 7.2 by \cite{BrandLN+20}. Therefore, the focus here is on the \textbf{implementation} and \textbf{complexity} of this algorithm in the EREW model.

    \paragraph{\textsc{Initialize}.} 
    First, assuming the inputs of this operation, $\mA$, $g$, and $z$, are specified by their non-zero entries, and assuming the matrix $\mA$ has full rank ($\nnz(\mA) \geq m$), $\mA$ and $g$ can be stored in $\mA$ and $g$ of the data structure using $\tO{\nnz(\mA)}$ work and $\tO{1}$ depth. The value of $\Psi(z) = \sum_{i \in [1,m]} \cosh(\lambda z_i)$ can be computed and stored in $p$ with $\tO{m}$ work and $\tO{1}$ depth. The value $\epsilon$ will be stored using $O(1)$ work and depth. Thus, Line~\ref{line:gradient_reduction:init:A} requires an overall $\tO{\nnz(\mA)}$ work and $\tO{1}$ depth.

    In Line~\ref{line:gradient_reduction:init:IW}, we initialize the sets $I^{(k,\ell)}$ and vectors $w^{(k, \ell)}$ for all $(k, \ell)$. The vectors $w^{(k, \ell)}$ will be initialized as vectors of length $n$. This requires a total work of $n \cdot \log^{-1}_{(1-\epsilon)}(n/m) \cdot (1.5/\epsilon) = \tO{n}$ and $\tO{1}$ depth. Furthermore, we maintain the sets $I^{(k, \ell)}$ as lists, as discussed in \Cref{lem:parallelSortedList}. Thus, initializing these empty sets takes $\log^{-1}_{(1-\epsilon)}(n/m) \cdot (1.5/\epsilon) = \tO{1}$ work and depth. Assuming that $m > n$ ($\mA$ has more rows than columns), Line~\ref{line:gradient_reduction:init:IW} requires an overall work of $\tO{m}$ and $\tO{1}$ depth.

    In the loop in Line~\ref{line:gradient_reduction:init:fori}, for each $i \in [1,m]$, we compute the corresponding $k_i$ and $\ell_i$ in order to place this index $i$ in the right bucket $I^{(k_i, \ell_i)}$. To achieve this, we make a temporary sorted list (\Cref{lem:parallelSortedList}) of $[1,m]$ using $(k_i, \ell_i)$ for each $i$ as the sorting criterion. Note that it is possible to compute $(k_i, \ell_i)$ for each $i \in [1,m]$ with overall $\tO{m}$ work and $\tO{1}$ depth. After making this list, we iterate over it, and for each $(k,\ell)$ in parallel, we add the segment of this list with $\{i \mid (k_i,\ell_i) = (k,\ell)\}$ to $I^{(k,\ell)}$ using the \textsc{Add} operation. It is easy to see that initializing the temporary list as well as adding the indices to sets $I^{(k,\ell)}$ can be done in $\tO{m}$ work and $\tO{1}$ depth, considering there are overall $m$ indices $i \in [1,m]$, and thus $\sum_{(k,\ell)} |I^{(k,\ell)}| = m$. 

    Finally, in the loop in Line~\ref{line:gradient_reduction:init:forW}, we need to compute the value of $w^{(k,\ell)} = \sum_{i \in I^{(k,\ell)}} \mA^\top g_i e_i$ for each $(k, \ell)$. To achieve this in a parallel setting, for each $(k,\ell)$ in parallel, we make a list $t^{(k,\ell)}$ of length $|I^{(k,\ell)}|$, and for each $i \in I^{(k,\ell)}$ in parallel, we compute $\mA^\top g_i e_i$ and set $t^{(k,\ell)}_i = \mA^\top g_i e_i$. Finally, for each $(k,\ell)$ in parallel, we sum the values in $t^{(k,\ell)}$ and add it to $w^{(k,\ell)}$. Since we work with the non-zero entries of $\mA$ and thus the non-zero entries of the $i$-th row of $\mA$, i.e., $a_i = \mA^\top e_i$, we can compute $\mA^\top g_i e_i$ for each $i$ with $\tO{\|e_i^\top \mA\|_0}$ work and $\tO{1}$ depth, and sum them up with $\tO{\sum_{i \in I^{(k,\ell)}} \|e_i^\top \mA\|_0}$ work and $\tO{1}$ depth. Thus, computing $w^{(k,\ell)} = \sum_{i \in I^{(k,\ell)}} \mA^\top g_i e_i$ takes an overall $\tO{|I^{(k,\ell)}| + \sum_{i \in I^{(k,\ell)}} \|e_i^\top \mA\|_0}$ work and $\tO{1}$ depth. This results in an overall $\tO{\sum_{(k,\ell)} |I^{(k,\ell)}| + \sum_{(k,\ell)} \sum_{i \in I^{(k,\ell)}} \|e_i^\top \mA\|_0} = \tO{m + \nnz(\mA)}$ work and $\tO{1}$ depth for this loop.

    As a result, assuming $\mA$ has full rank ($\nnz(\mA) \geq m$), the \textsc{Initialize} operation requires an overall $\tO{\nnz(\mA)}$ work and $\tO{1}$ depth.

    \paragraph{\textsc{Update}.}
    Unlike the \textsc{Update} operation defined by \cite{BrandLN+20}, which updates the data structure based on changes to a single index $i$, we define the \textsc{Update} operation to update the data structure based on a batch of indices $i \in M$, to enable parallelization. The new values are provided through lists $b$, $c$, and $d$ (\Cref{lem:parallelSortedList}).

    An important detail for implementing this data structure in the parallel setting is the list $u$ (Line~\ref{line:gradient_reduction:update:defu}). This list stores the new buckets of $i$'s for $i \in M$ whenever the \textsc{Update} operation is called. 
    For this list, we use the data structure discussed in \Cref{lem:parallelSortedList}. For each $i \in M$, we map $i$ to their new bucket (this computation can be done in $\tO{|M|}$ work and $\tO{1}$ depth, see Line~\ref{line:gradient_reduction:update:newkl}), and initialize $u$ (\Cref{lem:parallelSortedList}, \textsc{Initialize}) in $\tO{|M|}$ work and $\tO{1}$ depth.

    Now, we analyze the complexity of this operation. The operation primarily consists of one main loop (Line~\ref{line:gradient_reduction:update:loop}). For each $i \in M$, we:
        \begin{enumerate}
            \item Update the values of $z$ and $p = \Psi(z)$ (Lines~\ref{line:gradient_reduction:update:p} and \ref{line:gradient_reduction:update:z}).
            \item Determine the bucket $(k, \ell)$ in which $i$ previously resided and where it should move to (Lines~\ref{line:gradient_reduction:update:oldw} and \ref{line:gradient_reduction:update:neww}), and update the $w^{(k,\ell)}$ vectors accordingly (Lines~\ref{line:gradient_reduction:update:oldw} and \ref{line:gradient_reduction:update:newkl}).
            \item Update $g$ (Line~\ref{line:gradient_reduction:update:g}).
            \item Update $u$ (Line~\ref{line:gradient_reduction:update:u}).
        \end{enumerate}

    These iterations can be parallelized as follows:
    Updating $z_i$ and $g_i$ (Lines~\ref{line:gradient_reduction:update:z} and \ref{line:gradient_reduction:update:g}) require $\tO{1}$ work and depth for each $i \in M$, resulting in an overall complexity of $\tO{|M|}$ work and $\tO{1}$ depth.
    Updating $p$ (Line~\ref{line:gradient_reduction:update:p}) can be done in parallel as follows: We make a new temporary list, mapping each index $i \in M$ to $-\cosh{\lambda z_i} + \cosh{\lambda d_i}$. Finally, we iterate through this list, sum these values, and update $p$. It is easy to see that this takes $\tO{|M|}$ work and $\tO{1}$ depth, since $-\cosh{\lambda z_i} + \cosh{\lambda d_i}$ can be computed in parallel for each $i$ with total $\tO{|M|}$ work and $\tO{1}$ depth.

    The remaining parts are Lines~\ref{line:gradient_reduction:update:oldkl} to \ref{line:gradient_reduction:update:neww}, where we move $i \in M$ to new buckets and update $w^{(k,\ell)}$ accordingly. We implement this in parallel as follows: For each bucket in parallel, we use the search operation ($I^{(k,\ell)}.\textsc{Search}(M)$), to find the old buckets of $i \in M$ in $\tO{K \cdot |M|} = \tO{1}$ work and $\tO{1}$ depth. Finding the new buckets can also be done in parallel for each $i \in M$ in $\tO{|M|}$ work and $\tO{1}$ depth (Line~\ref{line:gradient_reduction:update:newkl}). 
        
    Then, for each $(k,\ell)$, we make a new list $tw^{(k,\ell)}$ of length $|M|$ in $\tO{K \cdot |M|} = \tO{|M|}$ work and $\tO{1}$ depth, and for each $i \in M$ in parallel, we set $tw^{(k,\ell)}_i = -\mA^{\top} g_{i} e_i$ if $i$ was previously in $(k,\ell)$-th bucket, $tw^{(k,\ell)}_i = \mA^{\top}b_i e_i$ if $i$ has to be moved to $(k,\ell)$-th bucket, and $tw^{(k,\ell)}_i = \zerovec$ otherwise. Note that we work with the non-zero entries of $\mA$ (and $tw^{(k,\ell)}_i$), so computing these values requires $\nnz (tw^{(k,\ell)}_i) = \nnz (a_i)$ work and $\tO{1}$ depth. Thus, the total work and depth are $\tO{|M| + \sum_{i \in M}\nnz(a_i)}$ and $\tO{1}$, respectively.

    Finally, for each $(k, \ell)$ in parallel, sum the values in $tw^{(k,\ell)}$ and update $w^{(k,\ell)}$ accordingly. Again, considering we work with the non-zero entries of $tw^{(k,\ell)}_i$, this addition across all $(k,\ell)$ takes $\tO{|M| + \sum_{i \in M}\nnz(a_i)}$ work and $\tO{1}$ depth.

    Considering the fact that we can attribute the necessary work and depth of creating vectors $b, c, d$ ($\tO{m}$ work and $\tO{1}$ depth) to \textsc{Initialize}, the overall complexity of the \textsc{Update} operation is $\tO{\sum_{i \in M}\nnz(a_i)}$ work and $\tO{1}$ depth.

    \paragraph{\textsc{Query}.}
    With a quick look at the implementation of this operation in \Cref{alg:gradient_reduction}, it is easy to see that, in order to execute this operation efficiently, it is crucial to be able to access the length of $I^{(k,\ell)}$ for each $(k,\ell)$ efficiently, as it is relevant for Lines~\ref{line:gradient_reduction:query:x} and \ref{line:gradient_reduction:query:v}. To implement this efficiently in the parallel setting, we define a vector $l$ of length $\log^{-1}_{(1-\epsilon)}(n/m) \cdot (1.5/\epsilon)$, such that for each $(k,\ell)$: $l_{(k,\ell)} = |I^{(k,\ell)}|$. 

    We assume this vector is initialized with \textsc{Initialize} and maintained with each call to \textsc{Update}, whenever the length of a bucket changes. To be more concrete, analogous to how we update the values of $w^{(k,\ell)}$ each time an $i$ is added to a bucket and/or removed from one (\textsc{Initialize}, Line~\ref{line:gradient_reduction:init:setW}, and \textsc{Update}, Lines~\ref{line:gradient_reduction:update:oldw} and \ref{line:gradient_reduction:update:neww}), we also increase or decrease the value of $l_{(k,\ell)}$, respectively. Note that this causes additional $\tO{m}$ work and $\tO{1}$ depth to \textsc{Initialize}, and additional $\tO{|I|}$ work and $\tO{1}$ depth to \textsc{Update}, which do not change the overall complexity of these operations.

    Using $l$, we can construct the vectors $x, v \in \R^K$ with $K = O(\epsilon^{-2} \log n)$ (Lines~\ref{line:gradient_reduction:query:x} and \ref{line:gradient_reduction:query:v}), in $\tO{1}$ work and depth. Furthermore, based on \Cref{cor:compute_flat}, computing the maximizer $s \in \R^K$ (Line~\ref{line:gradient_reduction:query:s}) requires $O(K) = \tO{1}$ work and depth. Finally, constructing $\ov = \mA^{\top}\mG(\nabla\Psi(\oz))^{\flat(\otau)} = \sum_{k,\ell} s_{(k,\ell)} w^{(k,\ell)}$ (Line~\ref{line:gradient_reduction:query:ov}) takes $\tO{n \cdot \log^{-1}_{(1-\epsilon)}(n/m) \cdot (1.5/\epsilon)} = \tO{n}$ work and $\tO{1}$ depth, since we need to scale and sum $\log^{-1}_{(1-\epsilon)}(n/m) \cdot (1.5/\epsilon)$ vectors of length $n$.

    This results to an overall $\tO{n}$ work and $\tO{1}$ depth for operation \textsc{Query}.

    \paragraph{\textsc{Potential}.}
    Returning $p \in \R$ takes $\tO{1}$ work and depth.
\end{proof}

\begin{algorithm}[h]%[p!]
\caption{Algorithm for reducing the dimension of $\nabla\Psi(\oz)^{\flat}$ and maintaining $\mA^{\top}\mG\nabla\Psi(\oz)^{\flat}$ (\Cref{lem:gradient_reduction}) \label{alg:gradient_reduction}, \cite[Algorithm 5]{BrandLN+20}}

\SetKwProg{Members}{members}{}{}
\SetKwProg{Proc}{procedure}{}{}
\Members{}{
$I^{(k,\ell)}$ for all $1 \leq k \leq \log^{-1}_{(1-\epsilon)}(n/m)$ and $0 \leq l \leq 1.5/\epsilon$: Partition of $[m]$. \\
$w^{(k,\ell)}\in\R^{n}$ for all $1 \leq k \leq \log^{-1}_{(1-\epsilon)}(n/m)$ and $0 \leq l \leq 1.5/\epsilon$ \tcp*{Maintained to be $\mA^{\top}\mG \mathbf{1}_{i \in I^{(k,\ell)}}$}
$g,z \in \R^m$, $p, \epsilon \in \R$, \tcp*{$p$ is maintained to be $\Psi(z)$}
}
\Proc{\textsc{Initialize}$(\mA\in\R^{m\times n},g\in\R^{m},\ttau\in\R^{m},z\in\R^{m},\epsilon>0)$}{
	$\mA\leftarrow\mA$, $g\leftarrow g$, $z\leftarrow z$, $p \leftarrow \Psi(z)$, $\epsilon \leftarrow \epsilon$ \label{line:gradient_reduction:init:A} \\
	$I^{(k,\ell)}\leftarrow \emptyset, w^{(k,\ell)}\leftarrow \zerovec \quad \forall k=1,\ldots,\log^{-1}_{(1-\epsilon)}(n/m)$ and $\ell=0,...,(1.5/\epsilon)$ \label{line:gradient_reduction:init:IW} \\
	\For{$i\in[1,m]$ \label{line:gradient_reduction:init:fori}}{
		Find $k_i,\ell_i$ such that $0.5+\ell_i\epsilon/2\le z_{i}<0.5+(\ell_i+1)\epsilon/2$ and $(1-\epsilon)^{k_i+1}\le\ttau_{i}\le(1-\epsilon)^{k_i}.$\\
		Add $i$ to $I^{(k_i,\ell_i)}$ \\ %and set $\otau_{i}\leftarrow(1-\epsilon)^{k_i+1}$ \\
	}
    \For{$(k, \ell)$ \label{line:gradient_reduction:init:forW}}{
        $w^{(k,\ell)} \leftarrow w^{(k,\ell)} + \sum_{i \in I^{(k, \ell)}} \mA^\top g_i e_i$ \label{line:gradient_reduction:init:setW}
    }
	
	\Return $(I^{(k,\ell)})_{k,\ell\ge0}$
}
\Proc{\textsc{Update}$(M\subseteq[m],b\in\R^M,c\in\R^M,d\in\R^M)$}{
    $u \leftarrow \emptyset$ \label{line:gradient_reduction:update:defu}\\ 
    \Comment{$u$ is a vector of length $m$ and for $i\in M$} stores their new bucket $(\bar{k}_i, \bar{\ell}_i)$\\
    \For{$i \in M$ \label{line:gradient_reduction:update:loop}}{
        $p \leftarrow p - \cosh{\lambda z_i} + \cosh{\lambda d_i}$ \label{line:gradient_reduction:update:p} \\
        $z_i \leftarrow d_i$ \label{line:gradient_reduction:update:z} \\
        Find $k_i,\ell_i$ such that $i \in I^{(k_i,\ell_i)}$, then remove $i$ from $I^{(k_i,\ell_i)}$. \label{line:gradient_reduction:update:oldkl} \\
        $w^{(k_i,\ell_i)}\leftarrow w^{(k_i,\ell_i)}-\mA^{\top} g_{i} e_i$ \label{line:gradient_reduction:update:oldw}\\    
        Find $\bar{k}_i,\bar{\ell}_i$ such that $0.5+\bar{\ell}_i\epsilon/2\le d_i <0.5+(\bar{\ell}_i+1)\epsilon/2$ and $(1-\epsilon)^{\bar{k}_i+1}\le c_i \le(1-\epsilon)^{\bar{k}_i}$, then insert $i$ into $I^{(\bar{k}_i,\bar{\ell}_i)}$. \label{line:gradient_reduction:update:newkl}\\
        $w^{(\bar{k}_i,\bar{\ell}_i)}\leftarrow w^{(\bar{k}_i,\bar{\ell}_i)}+\mA^{\top}b_i e_i$ \label{line:gradient_reduction:update:neww}\\
        $g_{i}\leftarrow b_i$ \label{line:gradient_reduction:update:g}\\
        $u.\textsc{Add}(\{(\bar{k}_i, \bar{\ell}_i) \mid i \in M\}) $ \label{line:gradient_reduction:update:u}
    }
	\Return $u$
}
\Proc{\textsc{Query}$()$}{
	\Comment{Construct scaled low dimensional representation of $\nabla\Psi(\oz)$}\\
	Let $x_{k,\ell} = |I^{(k,\ell)}|\left(\lambda\sinh(\lambda(0.5+l\epsilon/2))\right)$ \label{line:gradient_reduction:query:x} \\
	Interpret $x$ as an $O(\epsilon^{-2} \log n)$ dimensional vector. \\
	\Comment{Construct scaled low dimensional representation of $\otau$, here $C$ is the constant when define $\|\cdot\|_{\tau+\infty}$}
	Let $v$ be the vector with $v_{k,\ell} = \sqrt{|I^{(k,\ell)}|(1-\epsilon)^{k+1}}/C$. \label{line:gradient_reduction:query:v} \\
	$s \leftarrow \argmax_{y:\|vy\|_{2}+\|y\|_{\infty}\leq 1} \left\langle x,y\right\rangle $ via \Cref{cor:compute_flat}. \label{line:gradient_reduction:query:s} \\
	\Return $\ov = \sum_{k,l} s_{(k,\ell)} w^{(k,\ell)}$ and $s$ \label{line:gradient_reduction:query:ov}
}
\Proc{\textsc{Potential}$()$}{
	\Return $p$
}
\end{algorithm}

\subsection{Primal Maintenance}

In addition to maintaining the gradient, it is crucial for the IPM to also maintain a per-coordinate accurate approximation of the primal solution $x \in \mathbb{R}^m$. In this subsection, we discuss the Gradient Accumulator data structure in the parallel setting. The main result of this subsection is as follows:

\begin{lemma}[Gradient Accumulator, {\cite[Lemma D.3]{BrandLLSS0W21}}]
\label{lem:gradient_accumulator}
There exists a deterministic data-structure that supports the following operations
\begin{itemize}
\item $\textsc{Initialize }(x^{\init}\in\R^m,g\in\R^m,(I_k)_{1\le k \le K},\epsilon\in (0,1]^m)$:
	The data-structure initialized on the given vectors $x^{\init},g\in\R^m$,
	the partition $\bigcup_{k=1}^K I_k = [m]$ where $K=O(\epsilon^{-2}\log n)$, and the per-coordinate accuracy parameter $\epsilon \in (0,1]^m$ in $\tO{m}$ work and $\tO{1}$ depth.
    
\item $\textsc{Scale}(M \subseteq [m], a \in \R^{M})$: 
	For index $i \in M$, sets $g_{i} = a_i$. This operation takes $\tO{|M|}$ work and $\tO{1}$ depth.
    
\item $\textsc{Move}(M \subseteq [m],k\in[1,K]^{M})$: 
	For index $i \in M$, moves index $i$ to set $I_{M_i}$. This operation needs $\tO{|M|}$ work and $\tO{1}$ depth.
    
\item $\textsc{SetAccuracy}(M \in [m],\delta\in (0,1]^{M})$: 
	For index $i \in M$, sets $\epsilon_{i} = \delta_i$. This operation requires $\tO{|M|}$ work and $\tO{1}$ depth.
    
\item $\textsc{Query}(s \in \R^K, h \in \R^m)$: 
	Let $g^{(\ell)}$ and $\epsilon^{(\ell)}$ be the state of the vectors $g$ and $\epsilon$, respectively, during the $\ell$-th call to \textsc{Query}, and let $s^{(\ell)}$ and $h^{(\ell)}$ be the input arguments of the respective call. The vector $h$ is always provided as a sparse vector, so the locations of the non-zeros in the vector are known. Define $y^{(\ell)} = \mG^{(\ell)} \sum_{k=1}^K I_k^{(\ell)} s_k^{(\ell)}$ and $x^{(t)} = x^{\init} + \sum_{\ell=1}^t h^{(\ell)} + y^{(\ell)}$. The $t$-th call to \textsc{Query} returns a vector $\ox$ satisfying $|\ox_i - x^{(t)}_i| \leq \epsilon^{(t)}_i$ for all $i \in [m]$.

    After $T$ calls to \textsc{Query}, the total work of all $T$ calls is bounded by
    \begin{align*}
        \tO{TK + \sum_{\ell = 1}^{T} \|h^{(\ell)}\|_0 + TK \sum_{\ell = 1}^{T} \|y^{(\ell)}/\epsilon^{(\ell - 1)}\|_2^2},
    \end{align*}
    and the total depth by $\tO{T}$.

    The vector $\ox \in \R^m$ is returned as a pointer, and additionally, a set $J \subset [m]$ is returned containing the indices where $\ox$ differs from the result of the previous \textsc{Query} call.

\item $\textsc{ComputeExactSum}()$:
	Returns the current exact vector $x^{(t)}$ in $\tO{m}$ work and $\tO{1}$ depth.
\end{itemize}
\end{lemma}

\begin{proof}[Proof of \Cref{lem:gradient_accumulator}]
    The data structure is presented in \Cref{alg:gradient_accumulator}. This data structure, along with \Cref{alg:gradient_accumulator}, constitutes the parallel versions of the data structure discussed by \cite[Lemma D.3 and Algorithm 8]{BrandLLSS0W21}. Consequently, its correctness follows from Lemma D.3 in \cite{BrandLLSS0W21}. Here, we focus on the \textbf{implementation} and \textbf{complexity} of this data structure in the EREW model.

        \paragraph{\textsc{Initialize}.}
        In Line~\ref{line:gradient_accumulator:init}, it is clear that initializing $\ox$, $g$, and $\epsilon$ requires $\tO{m}$ work and $\tO{1}$ depth, as these are vectors of length $m$. Furthermore, initializing $f^{(0)}$ as a vector of length $K$ takes $K = \tO{1}$ work and depth, since $K = O(\epsilon^{-2} \log n)$. Setting the variable $t$ also requires $\tO{1}$ work and depth.  

        For the sets $I_1, \dots, I_K$, as well as $J$, we use the data structure discussed in \Cref{lem:parallelSortedList}. Considering that each partition set has at most $m$ entries, initializing $I_1, \dots, I_K$ takes an overall work of $K \cdot \tO{m} = \tO{m}$ work and $\tO{1}$ depth. Thus, the overall complexity of \textsc{Initialize} is $\tO{m}$ work and $\tO{1}$ depth.

        \paragraph{\textsc{Query}.} 
        Before analyzing this operation, we first discuss some details relevant to efficiently implementing this operation in the parallel setting.

        First, as discussed in \cite{BrandLLSS0W21}, to compute Line~\ref{line:gradient_accumulator:query:loopkK} efficiently, we need to maintain two sorted lists of the indices $i$ in set $I_k$ for each $k \in [1, K]$: one sorted by $\Delta^{(high)}_i$ and the other by $\Delta^{(low)}_i$. For these lists, we use the data structure discussed in \Cref{lem:parallelSortedList}. These lists are initialized in parallel for each $k \in [1, K]$ during the data structure initialization and are updated whenever an index $i$ is moved to another set $I_k$ and/or the values of $\Delta^{(high)}_i$ or $\Delta^{(low)}_i$ are updated.

        To be more concrete, during the initialization of the data structure, for each $k \in [1, K]$ in parallel, we compute $\Delta^{(high)}_i = +|\epsilon_i / (10 g_i)|$ and $\Delta^{(low)}_i = -|\epsilon_i / (10 g_i)|$ in parallel for each $i \in I_k$, and initialize two lists $l_k^{(high)}$ and $l_k^{(low)}$. Since there are $m$ indices, by \Cref{lem:parallelSortedList}, this adds $\tO{K + m} = \tO{m}$ work and $\tO{1}$ depth to \textsc{Initialize}, which does not affect its overall complexity.

        Furthermore, \Cref{alg:gradient_accumulator} defines the private procedures \textsc{ComputeX} and \textsc{UpdateDelta}, which are used in other operations (including \textsc{Query}). We analyze the complexity of these procedures below.

        \textsc{ComputeX} computes the new value for $\ox$ and updates $\hat{\ell}$ at the given indices in $M$, and adds these indices to $J$. In Line~\ref{line:gradient_accumulator:computex:findk}, for each $i$, we find $k_i$ such that $i \in I_{k_i}$. This can be achieved in parallel: for each $k \in [1, K]$, we use the search operation (see \Cref{lem:parallelSortedList}) to check whether $i \in I_k$ for $i \in M$. If so, we set $\bar{k}_i = k$ in a temporary vector $\bar{k} \in [1,K]^{|M|}$. This takes $\tO{K \cdot |M|} = \tO{|M|}$ work and $\tO{1}$ depth. Computing the new value of $\ox_i$ and updating $\hat{\ell}_i$ (Lines~\ref{line:gradient_accumulator:computex:updatex} and \ref{line:gradient_accumulator:computex:updatel}) can be done in parallel for $i \in M$, requiring $\tO{|M|}$ work and $\tO{1}$ depth. Note that in Line~\ref{line:gradient_accumulator:computex:updatex}, in order to compute $f^{(\hat{\ell}_i)}_{k_i}$, we use $f^{(\hat{\ell}_i)}_{k_i} = (\Delta^{(high)}_i + \Delta^{(low)}_i) / 2$. Updating $J$ uses $J.\textsc{Insert}(M)$ (\Cref{lem:parallelSortedList}) and also takes $\tO{|M|}$ work and $\tO{1}$ depth. Thus, \textsc{ComputeX} has a total complexity of $\tO{|M|}$ work and $\tO{1}$ depth.

        \textsc{UpdateDelta} operates similarly to \textsc{ComputeX}. For each $i \in M$, we find $k_i \in [1, K]$ such that $i \in I_{k_i}$ (Line~\ref{line:gradient_accumulator:updatedelta:findk}), in $\tO{|M|}$ work and $\tO{1}$ depth. Computing $\Delta^{(high)}_i$ and $\Delta^{(low)}_i$ (Lines~\ref{line:gradient_accumulator:updatedelta:deltahigh} and \ref{line:gradient_accumulator:updatedelta:deltalow}) in parallel requires $\tO{|M|}$ work and $\tO{1}$ depth. Finally, updating the values of $\Delta^{(high)}_i$ and $\Delta^{(low)}_i$ in the sorted lists $l_{k_i}^{(high)}$ and $l_{k_i}^{(low)}$, respectively, can be done using \textsc{Insert} and \textsc{Delete} operations (\Cref{lem:parallelSortedList}) in $\tO{|M|}$ work and $\tO{1}$ depth. Hence, \textsc{UpdateDelta} also requires $\tO{|M|}$ work and $\tO{1}$ depth.

        We now analyze \textsc{Query} in the parallel setting. Increasing $t$ and resetting $J$ (Line~\ref{line:gradient_accumulator:query:updatetJ}) requires $\tO{1}$ work and depth. Here, the old $J$ is returned as a result of \textsc{Query}, and $J$ is initialized as a new empty list.

        Updating the vector $f^{(t)}$ (Line~\ref{line:gradient_accumulator:query:group_bin_ft}) can be done in parallel for each $f^{(t)}_i$, requiring $\tO{K}$ work and $\tO{1}$ depth.

        In Line~\ref{line:gradient_accumulator:query:nonzero_hi}, performing \textsc{ComputeX} and \textsc{UpdateDelta} for non-zero $h_i$'s requires $\tO{\|h^{(\ell)}\|_0}$ work and $\tO{1}$ depth, assuming $h^{(\ell)}$ is the input vector $h$ during the $\ell$-th call to \textsc{Query}.

        In Line~\ref{line:gradient_accumulator:query:loopkK}, for each $k \in [1, K]$, we find all indices $i \in I_k$ such that $f^{(t)}_k > \Delta^{(high)}_i$ or $f^{(t)}_k < \Delta^{(low)}_i$, and perform \textsc{ComputeX} and \textsc{UpdateDelta} for these indices. To do this efficiently, for each $k \in [1, K]$, we iterate through the sorted lists $l_{k}^{(high)}$ and $l_{k}^{(low)}$ in order, stopping as soon as the condition no longer holds. As shown in \cite{BrandLLSS0W21}, the total number of times the condition $f^{(t)}_k > \Delta^{(high)}_i$ or $f^{(t)}_k < \Delta^{(low)}_i$ is satisfied over $T$ calls to \textsc{Query} is bounded by:
        \begin{align*}
            O\left(T \sum_{\ell=1}^T \|\mG^{(\ell)} (\sum_k I_k^{(\ell)} s_k^{(\ell)}) / \epsilon^{(\ell - 1)}\|_2^2\right).
        \end{align*}
        Considering this, the loop in Line~\ref{line:gradient_accumulator:query:loopkK} requires an overall:
        \begin{align*}
            \tO{K \cdot T \sum_{\ell=1}^T \|\mG^{(\ell)} (\sum_k I_k^{(\ell)} s_k^{(\ell)}) / \epsilon^{(\ell - 1)}\|_2^2}
        \end{align*}
        work and $\tO{1}$ depth.

        Combining all steps, the total work and depth for \textsc{Query} are:
        \begin{align*}
            \tO{TK + \sum_{\ell=1}^T \|h^{(\ell)}\|_0 + TK \sum_{\ell=1}^T \|\mG^{(\ell)} (\sum_k I_k^{(\ell)} s_k^{(\ell)}) / \epsilon^{(\ell - 1)}\|_2^2}
        \end{align*}
        work and $\tO{1}$ depth.

        \paragraph{\textsc{Scale}, \textsc{Move}, and \textsc{SetAccuracy}.}
        Examining the implementation of these operations in \Cref{alg:gradient_accumulator}, we observe that they use the procedures \textsc{ComputeX} and \textsc{UpdateDelta} as their core, resulting in $\tO{|M|}$ work and $\tO{1}$ depth for all three operations. In addition to this:
        \begin{itemize}
            \item \textsc{Scale} updates the vector $g$ at the indices $i \in M$ (Line~\ref{line:gradient_accumulator:scale:scale}), requiring $\tO{|M|}$ work and $\tO{1}$ depth.

            \item \textsc{Move} transfers each of the indices $i \in M$ to new sets $I_k$ (Line~\ref{line:gradient_accumulator:move:move}). Using the list operations \textsc{Insert} and \textsc{Delete} (\Cref{lem:parallelSortedList}), this is achievable in $\tO{|M|}$ work and $\tO{1}$ depth. Note that when moving an index $i$ to a new set $I_{k_i}$, it is also necessary to update $i$ in $l_{k_i}^{(high)}$ and $l_{k_i}^{(low)}$. Overall, \textsc{Move} requires $\tO{|M|}$ work and $\tO{1}$ depth.

            \item \textsc{SetAccuracy} updates the value of $\epsilon$ at the indices $i \in M$ (Line~\ref{line:gradient_accumulator:seracc:setacc}), requiring $\tO{|M|}$ work and $\tO{1}$ depth. 
        \end{itemize}
        
        In summary, each of these three operations requires $\tO{|M|}$ work and $\tO{1}$ depth.

        \paragraph{\textsc{ComputeExactSum}.}
        For this operation, all $i$ indices with $\hat{\ell}_i < t$ are identified and added to a list in parallel. Subsequently, \textsc{ComputeX} and \textsc{UpdateDelta} are performed. This operation requires at most $\tO{m}$ work and $\tO{1}$ depth.

\end{proof}

 \scalebox{0.75}{
    \begin{minipage}{\linewidth}
\begin{algorithm}[H]%[p!]
\caption{Algorithm for accumulating $\mG \nabla\Psi(\ov)^{\flat}$ 
(\Cref{lem:gradient_accumulator}) \label{alg:gradient_accumulator}, \cite[Algorithm 6]{BrandLLSS0W21}}
\SetKwProg{Members}{members}{}{}
\SetKwProg{Proc}{procedure}{}{}
\SetKwProg{Priv}{private procedure}{}{}
\Members{}{
	$I_1,...,I_K$ \tcp*{Partition $\bigcup_k I_k = [m]$}
    $g, \epsilon \in \R^m$ 
	$t \in \N, \ox \in \R^m$ \tcp*{\textsc{Query} counter and approximation of $x^{(t)}$}
	$\hat{\ell}\in\N^{m}$ \tcp*{$\hat{\ell}_i$ is value of $t$ when we last update $\ox_i \leftarrow x_i$}
	$f^{(t)} \in \R^K$ \tcp*{Maintain $f^{(t)}=\sum_{k=1}^t s^{(k)}$}
	$\Delta^{(high)},\Delta^{(low)} \in \R^m$ \tcp*{Maintain $\Delta_i = f^{(\hat{\ell}_i)}_k\pm |\epsilon_i / (10 g_i)| $ if $i\in I_k$} 
    $J$ \tcp*{Changed entries between 2 consecutive calls to \textsc{Query}}
}
\Proc{\textsc{Initialize}$(x^{\init} \in \R^m, g \in \R^m, (I_k)_{1\le k \le K}, \epsilon \in (0,1]^m)$}{
	$\ox \leftarrow x^{\init}$, 
	$(I_k)_{1\le k \le K} \leftarrow (I_k)_{1\le k \le K}$,
	$t \leftarrow 0$,
	$f^{(0)} \leftarrow \zerovec_K$,
	$g \leftarrow g$,
	$\epsilon \leftarrow \epsilon$, 
    $J \leftarrow \emptyset$ \label{line:gradient_accumulator:init}
}
\Priv{\textsc{ComputeX}$(M \subseteq [m], h \in \R^{M})$}{
    \For{$j \in M$}{
        Let $k_i$ be such that $i \in I_{k_i}$ \label{line:gradient_accumulator:computex:findk}\\
	    $\ox_i \leftarrow \ox_i + g_i \cdot (f^{(t)}_{k_i} -  f^{(\hat{\ell}_i)}_{k_i})+ h_j$  \label{line:gradient_accumulator:computex:updatex}\\
	    $\hat{\ell}_i \leftarrow t$ \label{line:gradient_accumulator:computex:updatel}\\
    }
    $J \leftarrow J \cup M$ \label{line:gradient_accumulator:computex:updateJ}\\
}
\Priv{\textsc{UpdateDelta}$(M)$}{
    \For{$i \in M$}{
	    Let $k_i$ be such that $i \in I_{k_i}$. \label{line:gradient_accumulator:updatedelta:findk} \\
    	$\Delta^{(high)}_i \leftarrow  f^{(\hat{\ell}_i)}_{k_i} +|\epsilon_i / (10 g_i)| $ \label{line:gradient_accumulator:updatedelta:deltahigh}\\
	    $\Delta^{(low)}_i \leftarrow  f^{(\hat{\ell}_i)}_{k_i} -|\epsilon_i / (10 g_i)| $ \label{line:gradient_accumulator:updatedelta:deltalow}
    }
}
\Proc{\textsc{Move}$(M \subseteq [m], k \in [1,K]^{M})$}{
    \textsc{ComputeX}$(M,\zerovec)$ \\
    Move index $M_i$ to set $I_{k_i}$ for all $i \in M$ \label{line:gradient_accumulator:move:move} \\
    $\textsc{UpdateDelta}(M)$
}
\Proc{\textsc{Scale}$(M \subseteq [m], a \in \R^{M})$}{
    \textsc{ComputeX}$(M,\zerovec)$ \\
	$g_{i} \leftarrow a_i$ for all $i \in M$ \label{line:gradient_accumulator:scale:scale}  \\
	$\textsc{UpdateDelta}(M)$
}
\Proc{\textsc{SetAccuracy}$(M \subseteq [m], \delta \in (0,1]^{M})$}{ %
    \textsc{ComputeX}$(M,\zerovec)$ \\
	$\epsilon_{i} \leftarrow \delta_i$ for all $i \in M$ \label{line:gradient_accumulator:seracc:setacc}  \\
	$\textsc{UpdateDelta}(M)$
}
\Proc{\textsc{Query}$(s \in \R^K, h \in \R^m)$}{
	$t \leftarrow t + 1$,	$J \leftarrow \emptyset \label{line:gradient_accumulator:query:updatetJ}$ \tcp*{Collect all entries that have changed since the last call to \textsc{Query}} 
	$f^{(t)} \leftarrow f^{(t-1)} + s$\label{line:gradient_accumulator:query:group_bin_ft} \\
	\textsc{ComputeX}$(\{i \mid h_i \neq 0\}, h)$, $\textsc{UpdateDelta}(\{i \mid h_i \neq 0\})$ \label{line:gradient_accumulator:query:nonzero_hi} \\
	\For{$k=1,\ldots,K$ \label{line:gradient_accumulator:query:loopkK}}{
		\textsc{ComputeX}$(\{i \in I_k \mid f^{(t)}_k > \Delta^{(high)}_i \text{ or } f^{(t)}_k < \Delta^{(low)}_i\},\zerovec)$ \label{line:gradient_accumulator:query:computeX_change_x_f} \\ 
        $\textsc{UpdateDelta}(\{i \in I_k \mid f^{(t)}_k > \Delta^{(high)}_i \text{ or } f^{(t)}_k < \Delta^{(low)}_i\})$ \label{line:gradient_accumulator:query:updatedelta_x_f} \\
	}
	\Return $\ox$, $J$
}
\Proc{\textsc{ComputeExactSum}$()$}{
    $\textsc{ComputeX}(\{i \in [m] \mid \hat{\ell}_i<t\},\zerovec)$ \\
    $\textsc{UpdateDelta}(\{i \in [m] \mid \hat{\ell}_i<t\})$ \\
	\Return $\ox$
}
\end{algorithm}
\end{minipage}
}

Now, we are ready to prove \Cref{thm:gradient_maintenance}:

\begin{proof}[Proof of \Cref{thm:gradient_maintenance}]
    The algorithm for this data structure is illustrated in \Cref{alg:gradient_maintenance}. It is essentially a combination of the data structures in \Cref{lem:gradient_reduction} and \Cref{lem:gradient_accumulator}. Thus, it is straightforward to verify the correctness and complexity of this data structure:

        \paragraph{\textsc{Initialize}.}
        We initialize one instance of \Cref{lem:gradient_reduction} and \Cref{lem:gradient_accumulator} in $\tO{m + \nnz(\mA)}$ work and $\tO{1}$ depth. Assuming $\mA$ has full rank, this results in $\tO{\nnz(\mA)}$ work and $\tO{1}$ depth. 

        \paragraph{\textsc{Update}.}
        Based on [\Cref{lem:gradient_reduction}, \textsc{Update}] and [\Cref{lem:gradient_accumulator}, \textsc{Scale} and \textsc{Move}], this operation takes $\tO{\sum_{i \in I}\nnz(a_i) + |I|}$ work and $\tO{1}$ depth. Assuming $\mA$ has no zero rows, this means \textsc{Update} takes $\tO{\sum_{i \in I}\nnz(a_i)}$ work and $\tO{1}$ depth.

        \paragraph{\textsc{SetAccuracy}.}
        Considering [\Cref{lem:gradient_accumulator}, \textsc{SetAccuracy}], this operation takes $\tO{|I|}$ work and $\tO{1}$ depth.

        \paragraph{\textsc{QueryProduct} and \textsc{QuerySum}.} 
        \textsc{QueryProduct} uses \textsc{Query} of $D^{(reduction)}$, which takes $\tO{n}$ work and $\tO{1}$ depth. Setting $s$ (Line~\ref{line:gradient_maintenance:qproduct:s}) takes $\tO{K}$ work and $\tO{1}$ depth, which is dominated by \textsc{Query}. Thus, $T$ calls to \textsc{QueryProduct} can be performed in $\tO{T n}$ work and $\tO{T}$ depth. 

        \textsc{QuerySum} uses \textsc{Query} of $D^{(accumulator)}$. After $T$ calls to \textsc{QuerySum}, the overall work is 
        \begin{align*}
	       \tO{
            T K 
		      + \sum_{\ell=0}^T \|h^{(\ell)}\|_0 
		      + T K \sum_{\ell=1}^T \|v^{(\ell)}/w^{(\ell-1)}\|_2^2
	       },
	    \end{align*} 
        with a total depth of $\tO{T}$.

        Since $K = O(\epsilon^{-2} \log n)$, $T$ calls to \textsc{QueryProduct} and \textsc{QuerySum} require 
        \begin{align*}
	       \tO{
		      T n 
		      + \sum_{\ell=0}^T \|h^{(\ell)}\|_0 
		      + T \cdot \sum_{\ell=1}^T \|v^{(\ell)}/w^{(\ell-1)}\|_2^2
	       },
	    \end{align*}
        total work and $\tO{T}$ depth.

        \paragraph{\textsc{ComputeExactSum}.}
        Based on [\Cref{lem:gradient_accumulator}, \textsc{ComputeExactSum}], this takes $\tO{m}$ work and $\tO{1}$ depth.

        \paragraph{\textsc{Potential}.}
        Based on [\Cref{lem:gradient_reduction}, \textsc{Potential}], this takes $\tO{1}$ work and depth.

\end{proof}

\begin{algorithm}[h]
\caption{Algorithm for maintaining primal $\ox$ and gradient $\nabla\Psi(\oz)^{\flat}$ (\Cref{thm:gradient_maintenance}) \label{alg:gradient_maintenance}}

\SetKwProg{Members}{members}{}{}
\SetKwProg{Proc}{procedure}{}{}
\Members{}{
    $D^{(reduction)}$ instance of gradient reduction data structure (\Cref{lem:gradient_reduction}) \\
    $D^{(accumulator)}$ instance of gradient accumulator data structure (\Cref{lem:gradient_accumulator}) \\
    $\epsilon > 0, s \in \R^K \text{ with } K=O(\epsilon^{-2}\log n)$
}
\Proc{\textsc{Initialize}$(\mA\in\R^{m\times n}, x^\init \in \R^m,g\in\R^{m},\ttau\in\R^{m},z\in\R^{m},w \in [0,1]^m,\epsilon>0)$}{
    $\epsilon \leftarrow \epsilon$ \\
	$(I_k)_{k \in [1,K]} = D^{(reduction)}.\textsc{Initialize}(\mA, g, \ttau, z, \epsilon)$ \\
    $D^{(accumulator)}.\textsc{Initialize}(x^\init,g,(I_k)_{k \in [1,K]},(w_i \cdot \epsilon)_{i \in [m]})$
}
\Proc{\textsc{Update}$(I\subseteq[m],b\in\R^{I},c\in\R^{I},d\in\R^{I})$}{
    $k = D^{(reduction)}.\textsc{Update}(I, b,c,d)$ \\ 
    \Comment{$k \in [1,K]^{I}$ are the new sets of indices $i \in I$} \\
    $D^{(accumulator)}.\textsc{Scale}(I,b)$ \\
    $D^{(accumulator)}.\textsc{Move}(I, k)$ 
}
\Proc{\textsc{SetAccuracy}$(I\subseteq[m], \delta \in [0,1]^{I})$}{
	$D^{(accumulator)}.\textsc{SetAccuracy}(I, (\delta_i \cdot \epsilon)_{i \in I})$ 
}
\Proc{\textsc{QueryProduct}$()$}{
	$\ov, s = D^{(reduction)}.\textsc{Query}()$ \\  
    \Comment{$s \in \R^K$ is low-dimensional representation of $\nabla\Psi(z)^{\flat(\ttau)}$} \\
    $s \leftarrow s$ \label{line:gradient_maintenance:qproduct:s}\\
    \Return $\ov$
}
\Proc{\textsc{QuerySum}$(h \in \R^m)$}{
	$\ox, J = D^{(accumulator)}.\textsc{Query}(s, h)$ \\ 
    \Return $\ox, J$
}
\Proc{\textsc{ComputeExactSum}$()$}{
	\Return $D^{(accumulator)}.\textsc{ComputeExactSum}()$
}
\Proc{\textsc{Potential}$()$}{
	\Return $D^{(reduction)}.\textsc{Potential}()$
}
\end{algorithm}

%% file: graphDataStructures.tex
\section{Graph Data Structures}\label{sec:graphDS}

In this section, we discuss some data structures introduced by \cite{BrandLN+20} and \cite{BrandLLSS0W21}, which are required for their IPM, and demonstrate how to implement them efficiently in the parallel setting. These data structures are specifically designed for graph instances. Therefore, throughout this section, we assume that the matrix $\mA \in \R^{m \times n}$ is an edge-vertex incidence matrix of a directed graph with $n+1$ vertices and $m$ edges, where one of its columns is removed (to ensure the matrix has full rank). 

\subsection{Dual Maintenance}

%\begin{restatable*}[Dual Maintenance]{theorem}{dualSlackMaintenance}
\begin{theorem}[Parallel Dual Maintenance, {\cite[Theorem E.1]{BrandLLSS0W21}}]
\label{thm:dual_maintenance}
There exists a data structure (\Cref{alg:dual_maintenance}) that supports the following operations:
\begin{itemize}
    \item \textsc{Initialize($\mA\in\R^{m\times n}, v^{\init}\in \R^m, w^{\init} \in [0,1]^m, \epsilon \in [0,1]$)}:  
    The data structure preprocesses the given matrix $\mA \in \R^{m \times n}$, the vector $v^{\init} \in \R^m$, the accuracy vector $0 < w^{\init} \le 1$, and the accuracy parameter $\epsilon \in [0,1]$ in $\tilde{O}(m)$ work and $\tO{1}$ depth. It assumes that the matrix $\mA$ is obtained by removing one column from an incidence matrix of a directed graph with $n+1$ vertices and $m$ edges. 
    
    \item \textsc{SetAccuracy}($I, \delta \in [0,1]^{I}$):  
    Sets $w_i \leftarrow \delta_i$ for $i \in I$ in $\tilde{O}(|I|)$ amortized work and $\tO{1}$ depth.
    
    \item \textsc{Add($h\in\R^n$)}:  
    Suppose this is the $t$-th time the \textsc{Add} operation is called, 
    and let $h^{(k)}$ be the vector $h$ given when the \textsc{Add} operation is called for the $k$-th time. 
    Define $v^{(t)}\in\R^m$ as:
    \[
    v^{(t)} = 
    v^{\init} + \mA \sum_{k=1}^t h^{(k)}.
    \]
    Then, the data structure returns a vector $\ov^{(t)} \in \R^m$ such that
    \[
    \|w^{-1}(\ov^{(t)} - v^{(t)}) \|_\infty \le \epsilon.
    \]
    The output is given in a compact representation to reduce its size.  
    In particular, the data structure returns a pointer to $\ov$ 
    and a set $I \subset [m]$ of indices $i$
    where $\ov^{(t)}_i$ has changed compared to $\ov^{(t-1)}_i$, i.e., the result of the previous call to \textsc{Add}.  
    The total work after $T$ calls to \textsc{Add} is bounded by
    \[
    \tilde{O}\left(
    Tn\log W + T \epsilon^{-2} \cdot \sum_{t=1}^T \| (v^{(t)}-v^{(t-1)})/w^{(t)}\|_2^2 \right).
    \]
    and its depth by $\tO{T}$. Here, $W$ is the ratio of the largest to the smallest nonzero entry in $w$. 
    
    \item \textsc{ComputeExact()}:  
    Returns $v^{(t)}\in \R^m$ in $\tO{\nnz(\mA)}$ work and $\tO{1}$ depth, 
    where $t$ is the number of times \textsc{Add} has been called so far 
    (i.e., $v^{(t)}$ is the state of the exact vector $v$ after the most recent call to \textsc{Add}).
\end{itemize}
\end{theorem}

%\end{restatable*}

\begin{proof}[Proof of \Cref{thm:dual_maintenance}]
The algorithm for this data structure is given in \Cref{alg:dual_maintenance}. Since this algorithm is (almost) identical to Algorithm 9 of \cite{BrandLLSS0W21}, its correctness follows from there. In the following, we primarily focus on its parallel implementation and analyze its complexity.  

Note that for lists and sets, we use the data structure from \Cref{lem:parallelSortedList}, and for \textsc{HeavyHitter}s, we use the data structure from \Cref{lem:large_entry_datastructure}.  

\paragraph{\textsc{Initialize}.} 
Initializing $\hat{f} \in \R^n, w, \ov \in \R^m, \epsilon \in \R$ and $t \in \N$ (Line~\ref{line:dual:init:vars}) takes overall $\tO{m}$ work and $\tO{1}$ depth. In the loop in Line~\ref{line:dual:init:loop}, we initialize a \textsc{HeavyHitter} data structure for each $i \in [1,\log T]$ (Line~\ref{line:dual:init:heavyhitterinit}), as well as the vector $f^{(j)} \in \R^n$ and the set $F_j$. Assuming we use the list from \Cref{lem:parallelSortedList} for $F_j$'s, based on \Cref{lem:large_entry_datastructure}, initializing \textsc{HeavyHitter}'s, and thus, this loop requires $\tO{m \log T} = \tO{m}$ work and $\tO{1}$ depth in total, as the loop can run in parallel for each $i \in [1, \log T]$. 

\paragraph{\textsc{SetAccuracy}.}
It is not hard to see that the components of this data structure, i.e., updating $w$ (Line~\ref{line:dual:setacc:delta}), updating $F_j$'s (Line~\ref{line:dual:setacc:djscale}), and scaling the \textsc{HeavyHitter}'s (Line~\ref{line:dual:setacc:djscale}), require overall $\tO{|I| \log T} = \tO{|I|}$ work and $\tO{1}$ depth. 

Note that adding entries to $F_j$'s causes additional work in the \textsc{Add} operation when \textsc{VerifyIndex} is called (Line~\ref{line:dual:add:verify_F}) and when the \textsc{HeavyHitter}'s, $D_j$'s, are rescaled (Line~\ref{line:dual:add:reweight}). So, in order to make the analysis of \textsc{Add} easier, we attribute the extra work caused by $i$'s added to $F_j$'s by \textsc{SetAccuracy} to \textsc{SetAccuracy}. As noted by \cite{BrandLLSS0W21}, we can flag such $i$'s and skip the loop in Line~\ref{line:dual:verifyIndex:loopJ} of \textsc{VerifyIndex} for these $i$'s, since the loop does not modify anything for them. Consequently, in \textsc{VerifyIndex}, for these $i$'s, we only need to compute the value of $(v^\init + \mA \hat{f})_i$. Assuming $\mA$ is given by its non-zero entries, this computation can be performed independently for each $i$ in parallel in $\tO{1}$ work and depth. As a result, we get an additional $\tO{|I|}$ work and $\tO{1}$ depth for the $i$'s that are added to $F_j$'s by \textsc{SetAccuracy}. Further, the additional work by $D_j.\textsc{Scale}$ in Line~\ref{line:dual:add:reweight} can also be bounded by $\tO{|I|}$, based on \Cref{lem:large_entry_datastructure}.

\paragraph{\textsc{Add}.}
The main costs of this operation are the subprocedures that it uses, i.e., \textsc{FindIndices} (Line~\ref{line:dual:add:findIndices}) and \textsc{VerifyIndex} (Lines~\ref{line:dual:add:verify_I} and \ref{line:dual:add:verify_F}). So, we start by analyzing these first. 

From Lemma E.5 of \cite{BrandLLSS0W21}, it follows that after $T$ calls to \textsc{Add}, the total work of \textsc{FindIndices} and \textsc{VerifyIndex} can be bounded by 
\begin{align*} 
    \tO{T \epsilon^{-2} \sum_{t = 1}^T \|(v^{(t)} - v^{(t-1)})/w^{(t)}\|_2^2 + Tn\log W }
\end{align*}
and their depth by $\tO{T}$. 

The proof of this is analogous to the one in Lemma E.5 of \cite{BrandLLSS0W21} and Lemma 6.5 of \cite{BrandLN+20}. The main idea of the proof is that, in \textsc{FindIndices}, for each $j \in [1, \log T]$, we call $D_j.\textsc{HeavyQuery}(f^{(j)}, 0.2\epsilon/\log n)$ in every $2^j$ iterations, and thus we can show that its cost is bound by 
\begin{align*}
    \tO{T \epsilon^{-2} \sum_{t = 1}^T \|(v^{(t)} - v^{(t-1)})/w^{(t)}\|_2^2 + T 2^{-j} n\log W }.
\end{align*}
We can then sum this for $j = 1,\dots,\log T$, which gives us the bound mentioned above, as $T = O(\sqrt{n})$.

Furthermore, it is not hard to see that \textsc{VerifyIndex} requires $\tO{|I|}$ work and $\tO{1}$ depth: Computing $(v^\init + \mA \hat{f})_i$ (Line~\ref{line:dual:verifyIndex:computev}) as well as updating the set $J$ (Line~\ref{line:dual:verifyIndex:updateJ}) can be done in parallel for $i \in I$. Moreover, updating $F_j$'s (Line~\ref{line:dual:verifyIndex:updateFj}) and scaling $D_j$'s (Line~\ref{line:dual:verifyIndex:scale}) require $\tO{|I| \log T} = \tO{|I|}$ work and $\tO{1}$ depth, since $D_j.\textsc{Scale}$ takes $\tO{|I|}$ work and $\tO{1}$ depth based on \Cref{lem:large_entry_datastructure}. 

Now, if we consider the occurrences of \textsc{VerifyIndex} in \textsc{Add}, there are two of them. First, Line~\ref{line:dual:add:verify_I}, which can be bounded by the size of $I$, which, as a result of \textsc{FindIndices}, is bounded by the work of \textsc{FindIndices}. Second, Line~\ref{line:dual:add:verify_F}, which is called for the $i$'s for which 
\begin{enumerate}
    \item $w_i$ has been updated over the last $2^j$ iterations (added in Line~\ref{line:dual:setacc:fj} in \textsc{SetAccuracy}).
    \item $v_i$ has changed by more than $0.2 w_i \epsilon / \log n$ over the last $2^j$ iterations (added in Line~\ref{line:dual:verifyIndex:updateFj} in \textsc{FindIndices}).
\end{enumerate}
The work caused by the first case is already covered by \textsc{SetAccuracy}. For the second case, according to \cite{BrandLLSS0W21}, the number of such $i$'s is bounded by 

\begin{align*}
    \tO{T \epsilon^{-2} \sum_{t = 1}^T \|(v^{(t)} - v^{(t-1)})/w^{(t)}\|_2^2}.
\end{align*}
Thus, we obtain the bounds mentioned above for $T$ calls to \textsc{FindIndices} and \textsc{VerifyIndex}. 

Reinitializing the data structure (Line~\ref{line:dual:add:init}) takes $\tO{m}$ work and $\tO{1}$ depth once every $T$ iterations, resulting in an amortized $\tO{m/n}$ work and $\tO{1}$ depth per iteration. The cost for Line~\ref{line:dual:add:reweight} is dominated by the cost of Line~\ref{line:dual:add:verify_F}. Updating $t$ and $\hat{f}$ in Line~\ref{line:dual:add:findIndices} requires $\tO{n}$ work and $\tO{1}$ depth. 

This results in a total work of 
\begin{align*} 
    \tO{T \epsilon^{-2} \sum_{t = 1}^T \|(v^{(t)} - v^{(t-1)})/w^{(t)}\|_2^2 + Tn\log W}
\end{align*}
and depth of $\tO{T}$ over $T$ calls to \textsc{Add}.

\paragraph{\textsc{ComputeExact}.} This is effectively a matrix-vector product, which, in the parallel setting, can be done in $\tO{\nnz (\mA)} = \tO{m}$ work and $\tO{1}$ depth. 

\end{proof}

\scalebox{0.75}{
\begin{algorithm}[H]%[p!]
\caption{\label{alg:dual_maintenance}Algorithm for \Cref{thm:dual_maintenance}, {\cite[Algorithm 8]{BrandLLSS0W21}}} %
\SetKwProg{Members}{members}{}{}
\SetKwProg{Proc}{procedure}{}{}
\SetKwProg{PrivateProc}{private procedure}{}{}
\Members{}{
$\hat{f}\in \R^n, w, \ov\in\R^m, \epsilon \in [0,1], t\in \N$ \tcp*{$t$ is the \textsc{Add} counter}
$T =  O(\sqrt{n}), D_j$ for $0\le j\le \log T$ \tcp*{$D_j$ are \textsc{HeavyHitter} (\Cref{lem:large_entry_datastructure})} 
$f^{(j)} \in \R^n$ and $F_j\subset [m]$ for $0\le j\le \log T$ \\
}
\Proc{\textsc{Initialize}$(\mA, v^{\init}, w^{\init},\epsilon$)}{
	$\ov \leftarrow v^{\init}$, $\hat{f} \leftarrow \zerovec_n$, $w \leftarrow w^{\init}$, $\epsilon \leftarrow \epsilon$, $t \leftarrow 0$ \label{line:dual:init:vars} \\
	\For{$j=0,\dots,\log T$ \label{line:dual:init:loop}}{
		$D_j.\textsc{Initialize}(\mA,w^{-1})$ (\Cref{lem:large_entry_datastructure}) \label{line:dual:init:heavyhitterinit}\\
		$f^{(j)} \leftarrow \zerovec_n$, $F_j\leftarrow \emptyset$ \label{line:dual:init:initF}\\
	}
}
\PrivateProc{\textsc{FindIndices}$(h\in\R^n)$}{
	$I \leftarrow \emptyset$\\
	\For{$j=\log T,...,0$}{
		$f^{(j)} \leftarrow f^{(j)} + h$ 
		\Comment{When $2^j | t$, then $f^{(j)} = \sum_{k=t-2^j+1}^t h^{(k)}$}\\
		\If{$2^j | t$}{
			$I \leftarrow I \cup D_j.\textsc{HeavyQuery}(f^{(j)}, 0.2 \epsilon/\log n)$ \label{line:detect_changes}\\
			$f^{(j)} \leftarrow \zerovec_n$
		}
	}
	\Return $I$
}
\Proc{\textsc{SetAccuracy}$(I \subseteq [1,m], \delta \in \R^{I}$)}{%
	$w_i \leftarrow \delta_i$ for all $i \in I$ \label{line:dual:setacc:delta} \\
	\For{$j=0,...,\log T$ \label{line:dual:setacc:loop}}{
		$F_j \leftarrow F_j \cup I$ \label{line:dual:setacc:fj} \\ 
        $D_j.\textsc{Scale}(I, \zerovec_{|I|})$ \label{line:dual:setacc:djscale}
	}
}
\PrivateProc{\textsc{VerifyIndex}$(I \subseteq [1,m])$}{
    $J \leftarrow \emptyset$ \\
    \For{$i \in I$}{
	\If{$|\ov_i - (v^\init + \mA \hat{f})_i| \ge 0.2 w_i \epsilon /\log n$}{ \label{line:check_change}
		$\ov_i \leftarrow (v^\init + \mA \hat{f})_i$ \label{line:dual:verifyIndex:computev}\\
		$J \leftarrow J \cup \{i\}$ \label{line:dual:verifyIndex:updateJ}
	}
    }
    \For{$j=0,...,\log T$ \label{line:dual:verifyIndex:loopJ}}{
		$F_j \leftarrow F_j \cup J$ \label{line:dual:verifyIndex:updateFj} \Comment{Notify other $D_j$'s to stop tracking $i$.}\\
		$D_j.\textsc{Scale}(J, \zerovec_{|J|})$ \label{line:dual:verifyIndex:scale} 
	}
    \Return $J$
}
\Proc{\textsc{Add}$(h\in\R^n)$}{
	\lIf{$t = T$ \label{line:dual:add:init}}{\Return $\textsc{Initialize}(\mA, \mA (\hat{f}+h), w, \epsilon)$}%
	$t \leftarrow t + 1$, $\hat{f} \leftarrow \hat{f} + h$,
	$I \leftarrow \textsc{FindIndices}(h)$ \label{line:dual:add:findIndices}\\
	$I \leftarrow \textsc{VerifyIndex}(I)$ \label{line:dual:add:verify_I}\\
	\lFor{$j:2^j | t$ \label{line:dual:add:firstLoop}}{
		$I \leftarrow I \cup \textsc{VerifyIndex}(F_j)$ \label{line:dual:add:verify_F}
	}
	\For{$j:2^j | t$ \label{line:dual:add:secondLoop}}{
		%//\For{$i\in I \cup F_j$}{
		%//	$D_j.\textsc{Scale}(i, 1/w_i)$  \label{line:dual:add:reweight}
		%}
        $I' \gets  I \cup F_j$\\
        $D_j.\textsc{Scale}(I', 1/w_{I'})$  \label{line:dual:add:reweight}\\
		$F_j \leftarrow \emptyset$ \label{line:dual:add:empty_F}\\
	}
	\Return $I$, $\ov$
}
\Proc{\textsc{ComputeExact}$()$}{
	\Return $v^\init + \mA \hat{f}$
}
\end{algorithm}
}

\subsection{HeavySampler}

To guarantee the progress of their IPM, \cite{BrandLLSS0W21} use a sampling procedure that samples a random diagonal scaling $\mR \in \R^{m \times m}$. This sampling procedure satisfies certain properties, such as bounded variance, bounded maximum, and mean preservation. To handle the sampling procedure and compute $\mR$ efficiently, they introduce the \textsc{HeavySampler} data structure. \Cref{thm:heavysampler} represents a parallel version of \textsc{HeavySampler}. 

It is important to note that the definition of \textsc{HeavySampler} in \cite{BrandLLSS0W21} (Definition 6.3) and their sampling scheme (Definition 4.13) are more general, as their IPM is designed for solving general LPs with two-sided constraints. Nevertheless, while our version is more specific, based on Lemma F.3 of \cite{BrandLLSS0W21}, it still yields a feasible \textsc{HeavySampler} for min-cost flow as well as the max flow problem.

\begin{theorem}[Parallel \textsc{HeavySampler}, {\cite[Definiton 6.3 and Lemma F.3]{BrandLLSS0W21}}] \label{thm:heavysampler}
    There is a data structure, that supports the following operations:
    \begin{itemize}
        \item \textsc{Initialize}$(\mA \in \R^{m \times n}, g \in \R^m_{>0}, \tau \in \R^m_{>0})$: 
        Initializes the data structure in $\tO{m}$ work and $\tO{1}$ depth for the matrix $\mA$, which is obtained by removing one column from an incidence matrix of a directed graph with $n+1$ vertices and $m$ edges.

        \item \textsc{Scale}$(I \in [m], a \in \R^I_{>0}, b \in \R^I_{>0})$: 
        Sets $g_i \leftarrow a_i$ and $\tau_i \leftarrow b_i$ for all $i\in I$ in $\tO{|I|}$ amortized work and $\tO{1}$ depth. 

        \item \textsc{Sample}$(h \in \R^m)$: 
        Assume $C_1, C_2, C_3$ are given. %\jan{if we have this assumption, then we can remove $C_1,C_2,C_3$ from time complexity} are certain constants provided by the IPM. 
        This operation returns a random diagonal matrix $\mR \in \R^{m \times m}$ where independently for all $i$ we have $\mR_{i,i} = 1/p_i$ with probability $p_i$ and $\mR_{i,i} = 0$ otherwise for 
        $$p_i \ge \min \left\{1, C_1 \frac{m}{\sqrt{n}} \cdot \frac{(\mG\mA h)_i^2}{\|\mG\mA h\|_2^2} + C_2 \frac{1}{\sqrt{n}} + C_3 \frac{n\tau_i}{\|\tau\|_1} \right\}.$$ 
        With high probability, the output size of $\mR$ is bounded by 
        $$\tO{(C_1+C_2)\frac{m}{\sqrt{n}} + C_3 n}.$$
        The work of this operation can w.h.p. be bounded by
        $$\tO{(C_1+C_2)\frac{m}{\sqrt{n}} + C_3 n + n \log W}$$ 
        and its depth by $\tO{1}$, where $W$ is a bound on the ratio of largest to smallest entry in $g$ and $\tau$.
    \end{itemize}
\end{theorem}

\begin{algorithm}[ht]
\caption{Algorithm for \textsc{HeavySampler} (\Cref{thm:heavysampler}), inspired by {\cite[\textsc{SamplePrimal} - Algorithm 7]{BrandLN+20}} \label{alg:heavysampler}}

\SetKwProg{Members}{members}{}{}
\SetKwProg{Proc}{procedure}{}{}
\Members{}{
    $D^{(\sample)}$ instance of \textsc{HeavyHitter}, \Cref{lem:large_entry_datastructure} \\
    $D^{(\tau-\sample)}$ instance of \textsc{$\tau$-Sampler}, \Cref{thm:tausampler} \\
    $g \in \R^m_{>0}$ 
}
\Proc{\textsc{Initialize}$(\mA\in\R^{m\times n},g \in \R^m_{>0},\tau \in \R^m_{>0})$}{
    $g \leftarrow g$ \label{line:heavysample:init:g} \\
    %$\tau \leftarrow \tau$ \label{line:heavysample:init:tau}\\
    $D^{(\sample)}.\textsc{Initialize}(\mA, g)$ \label{line:heavysample:init:sample}\\
    $D^{(\tau-\sample)}.\textsc{Initialize}(\tau)$ \label{line:heavysample:init:tausample}
    %$D^{(\sample, \sigma)}.\textsc{Initialize}(\mA, \tau^{1/2} g) \label{line:heavysample:init:sampleLS}$ 
}
\Proc{\textsc{Scale}$(I\subseteq[m],a\in\R^{I},b\in\R^{I})$}{
    $g_i \leftarrow a_i$ for all $i \in I$ \label{line:heavysample:scale:g} \\
    %$\tau_i \leftarrow b_i$ for all $i \in I$ \label{line:heavysample:scale:otau} \\
    $D^{(\sample)}.\textsc{Scale}(I, a)$ \label{line:heavysample:scale:sample} \\
    $D^{(\tau-\sample)}.\textsc{Scale}(I, b)$ \label{line:heavysample:scale:tausample} \\
    %$D^{(\sample, \sigma)}.\textsc{Scale}(I, b^{1/2}a)$ \label{line:heavysample:scale:sampleLS} 
}
\Proc{\textsc{Sample}$(h \in \R^n, C_1 \in R, C_2 \in \R, C_3 \in \R)$}{
    %$I_u \leftarrow D^{(\sample, \sigma)}.\textsc{LeverageScoreSample}(3 C_3)$ \label{line:heavysample:sample:levscoresample} \\
    $I_u \leftarrow D^{(\tau-\sample)}.\textsc{Sample}(3 C_3)$ \label{line:heavysample:sample:tausample} \\
	$I_v \leftarrow D^{(\sample)}.\textsc{Sample}(h, 3 C_1 m/\sqrt{n} \cdot (16\log^8(n)))$ \label{line:heavysample:sample:sample} \\
	$I_w \subseteq [m]$, where $\P[i \in I_w] = 3 C_2 / \sqrt{n}$ independently for each $i \in [m]$. \label{line:heavysample:sample:normalsample} \\
	$I \leftarrow I_u \cup I_v \cup I_w$ \label{line:heavysample:sample:I} \\
    Let $u,v,w$ be $\zerovec_m$. \label{line:heavysample:sample:uvw} \\
	\LineComment{Set $u_i = \P[i \in I_u]$, $v_i = \P[i \in I_v]$, $w_i = \P[i \in I_w]$ for $i \in I$}
	%$u_I \leftarrow D^{(\sample, \sigma)}.\textsc{LeverageScoreBound}(3 C_3, I)$ \label{line:heavysample:sample:levScoreBound} \\
    $u_I \leftarrow D^{(\tau-\sample)}.\textsc{Probability}(I, 3 C_3)$  \label{line:heavysample:sample:tauprob} \\
    $v_I \leftarrow D^{(\sample)}.\textsc{Probability}(I, h, 3 C_1 m/\sqrt{n} \cdot (16\log^8(n)))$ \label{line:heavysample:sample:prob} \\
	$w_i \leftarrow 3C_2 / \sqrt{n}$ for $i \in I$ \label{line:heavysample:sample:normalset}\\
	$\mR \leftarrow \mathbf{0}_{m \times m}$ \\
    \For{$i \in I$ \label{line:heavysample:sample:loop}}{
		$\mR_{i,i} \leftarrow 1/\min\{1, u_i + v_i + w_i\}$ 
			with probability $\frac{\min\{1, u_i + v_i + w_i\}}{1-(1-u_i)(1-v_i)(1-w_i)}$
			\label{line:fix_probability}
	}
	\Return $\mR$
}
\end{algorithm}

\begin{proof}[Proof of \Cref{thm:heavysampler}]
    The algorithm for this data structure is presented in \Cref{alg:heavysampler}. The correctness of \textsc{Initialize} and \textsc{Scale} is trivial. The correctness of \textsc{Sample} follows from Lemma 8.2 of \cite{BrandLN+20}, as the distribution of matrix $\mR$ can be analyzed analogously. So, we just need to discuss the implementation of \Cref{alg:heavysampler} in the parallel setting and analyze this algorithm's complexity. 
    
        \paragraph{\textsc{Initialize}.}
        Initializing the $m$-dimensional vector $g$ (Line~\ref{line:heavysample:init:g}) takes $\tO{m}$ work and $\tO{1}$ depth. In Lines~\ref{line:heavysample:init:sample} and \ref{line:heavysample:init:tausample}, we initialize an instance of \textsc{HeavyHitter} data structure and an instance of \textsc{$\tau$-Sampler} data structure, respectively, which based on \Cref{lem:large_entry_datastructure} and \Cref{thm:tausampler} takes $\tO{m}$ work and $\tO{1}$ depth.

        \paragraph{\textsc{Scale}.}
        Updating the vector $g$ at given indices $i \in I$ (Line~\ref{line:heavysample:scale:g}) can be done in parallel with $\tO{|I|}$ work and $\tO{1}$ depth. Scaling $D^{(\sample)}$ (Line~\ref{line:heavysample:scale:sample}), based on \Cref{lem:large_entry_datastructure}, requires amortized $\tO{|I|}$ work and $\tO{1}$ depth. Similarly, scaling $D^{(\tau-\sample)}$ (Line~\ref{line:heavysample:scale:tausample}) takes $\tO{m}$ work and $\tO{1}$ depth, based on \Cref{thm:tausampler}.

        \paragraph{\textsc{Sample}.}
        First, we bound the number of entries of $I$: In Line~\ref{line:heavysample:sample:I}, we set $I = I_u \cup I_v \cup I_w$, thus 
        \begin{align*}
            |I| \leq |I_u| + |I_v| + |I_w|.
        \end{align*}
        As $I_u$ and $I_v$ are the results of the \textsc{Sample} operation of \textsc{$\tau$-Sampler} and \textsc{HeavyHitter}, their sizes can be bounded by \Cref{thm:tausampler} and \Cref{lem:large_entry_datastructure}, respectively. W.h.p., we have:
        \begin{align*}
            |I_u| \leq \tO{C_3 n} \quad \text{and} \quad 
            |I_v| \leq \tO{C_1 m / \sqrt{n}}.
        \end{align*}
        Furthermore, in Line~\ref{line:heavysample:sample:normalsample}, we sample each edge with probability $3C_2/\sqrt{n}$, thus, w.h.p., $|I_w| \leq \tO{C_2 m / \sqrt{n}}$.
        
        So, we have that, w.h.p., $|I|$ is bounded by 
        \begin{align*}
            \tO{C_3 n} + \tO{C_1 m / \sqrt{n}} + \tO{C_2 m / \sqrt{n}} = \tO{(C_1 + C_2)m/\sqrt{n} + C_3 n}. 
        \end{align*}
        Note that, since in Line~\ref{line:heavysample:sample:loop}, we set the values of $\mR_{i,i}$ for $i \in I$, w.h.p., the number of non-zero entries of $\mR$ can also be bounded by $\tO{(C_1 + C_2)m/\sqrt{n} + C_3 n}$. 

        For the implementation of sets $I_u, I_v, I_w$, as well as $I$, we can use the list data structure discussed in \Cref{lem:parallelSortedList}. Considering this, by \Cref{thm:tausampler} and \Cref{lem:large_entry_datastructure}, the \textsc{Sample} operations in Lines~\ref{line:heavysample:sample:tausample} and \ref{line:heavysample:sample:sample}, w.h.p., take $\tO{C_3 n + \log W}$ work and $\tO{C_1 m/\sqrt{n} + n \log W}$ work, respectively, and $\tO{1}$ depth each. 
        
        In Line~\ref{line:heavysample:sample:normalsample}, we sample each edge independently with probability $3 C_2 / \sqrt{n}$. Similar to edge sampling in each bucket in \Cref{thm:tausampler}-\textsc{Sample}, we can first sample a binomial and then pick the corresponding number of edges uniformly at random. Thus, the work for this sampling can be bounded by the number of sampled edges. Note that redundant edges can be removed by the list data structure, as the list stays sorted. Therefore, Line~\ref{line:heavysample:sample:normalsample}, w.h.p., takes $\tO{C_2 m/\sqrt{n}}$ work and $\tO{1}$ depth. 
        
        Consequently, w.h.p., we have in total $\tO{(C_1 + C_2)m/\sqrt{n} + C_3 n + n\log W}$ work and $\tO{1}$ depth for computing $I$.

        As vectors $u, v, w$ can be initialized with \textsc{Initialize} (attribute the complexity to \textsc{Initialize} as well) and later each referred to by a pointer (Line~\ref{line:heavysample:sample:uvw}), setting $u$ and $v$  (Lines~\ref{line:heavysample:sample:tauprob} and \ref{line:heavysample:sample:prob}), by \Cref{thm:tausampler} and \Cref{lem:large_entry_datastructure}, takes (w.h.p.) $\tO{|I|}$ and $\tO{|I| + n \log W}$ work, respectively, and $\tO{1}$ depth. Setting $w$, w.h.p., also takes $\tO{|I|}$ work and $\tO{1}$ depth, as we can set $w_i$ for each $i \in I$ in parallel. Finally, the loop in Line~\ref{line:heavysample:sample:loop} takes $\tO{|I|}$ work and $\tO{1}$ depth, since $\mR_{i,i}$ can be computed independently for each $i \in I$ in parallel (to be more concrete, we store $\mR$ by its non-zero entries. To achieve this, we can again use the list data structure of \Cref{lem:parallelSortedList}, as it represents a list that allows list operations in parallel). 

        Overall, w.h.p., the work complexity of \textsc{Sample} can be bounded by 
        \begin{align*}
            \tO{(C_1 + C_2)m/\sqrt{n} + C_3 n + n \log W} + \tO{|I| + n \log W}.
        \end{align*}
        As we already showed that $|I|$, w.h.p., can be bounded by $\tO{(C_1 + C_2)m/\sqrt{n} + C_3 n}$, we have an overall complexity of
        \begin{align*}
            \tO{(C_1 + C_2)m/\sqrt{n} + C_3 n + n \log W} \text{ work} \quad \text{and} \quad \tO{1} \text{ depth}. 
        \end{align*}

\end{proof}

Note that, in the following, we assume that $C_1, C_2, C_3 \in \tO{1}$ are provided by the IPM. Under this assumption, the \textsc{Sample} operation can be further simplified: With high probability, the output size of \textsc{Sample} can be bounded by $\tO{m/\sqrt{n} + n}$, and its work by $\tO{m/\sqrt{n} + n \log W}$.